%% file: sn-article.tex
\documentclass[pdflatex,sn-mathphys-ay,iicol]{sn-jnl}

\usepackage{graphicx}%
\usepackage{multirow}%
\usepackage{amsmath,amssymb,amsfonts}%
\usepackage{amsthm}%
\usepackage{mathrsfs}%
\usepackage[title]{appendix}%
\usepackage{xcolor}%
\usepackage{textcomp}%
\usepackage{manyfoot}%
\usepackage{booktabs}%
\usepackage{algorithm}%
\usepackage{algorithmicx}%
\usepackage{algpseudocode}%
\usepackage{listings}%
\usepackage{dsfont}
\usepackage{mathtools}
\usepackage{float}
\usepackage{url}
\usepackage{mathabx}
\usepackage{subcaption}
\usepackage{kantlipsum,cuted}
\usepackage{tikz}
\usetikzlibrary{positioning, arrows.meta}
\usetikzlibrary{calc}
\usepackage{accents}
\usepackage{caption} 
\usepackage{siunitx} 
\usepackage{rotating}
\usepackage{soul}

\newcommand{\rme}{\mathrm{e}}
\newcommand{\rms}{\mathrm{s}}
\newcommand{\BF}[1]{\mathbf{#1}}
\newcommand{\BS}[1]{\boldsymbol{#1}}
\newcommand{\bX}{\BF{X}}
\newcommand{\bZ}{\BF{Z}}

\newcommand{\bSigma}{\BS{\Sigma}}
\newcommand{\bOmega}{\BS{\Omega}}
\newcommand{\balpha}{\BS{\alpha}}

\newcommand{\bE}{\BF{E}}
\newcommand{\bm}{\BF{m}}
\newcommand{\bS}{\BF{S}}

\newcommand{\bC}{\BF{C}}

\newcommand{\bV}{\BF{V}}

\newcommand{\bP}{\BF{P}}

\newcommand{\bU}{\BF{U}}

\newcommand{\offset}{\mathrm{O}}

\newcommand{\rmm}{\mathrm{m}}
\newcommand{\bphi}{\BS{\varphi}}

\newcommand{\btheta}{\BS{\theta}}

\newcommand{\bmu}{\BS{\mu}}

\newcommand{\softmax}[1]{\sigma (#1)}
\newcommand{\rmX}{\mathrm{X}}
\newcommand{\rmZ}{\mathrm{Z}}
\newcommand{\rmV}{\mathrm{V}}
\newcommand{\rmY}{\mathrm{Y}}
\newcommand{\mgf}{\mathrm{M}}
\newcommand{\vect}[1]{\mathrm{Vect}\left(#1\right)}

\newcommand{\cholesky}{\BF{L}}

\newcommand{\RR}{\mathds{R}}
\newcommand{\NN}{\mathds{N}}
\newcommand{\PP}{\mathds{P}}
\newcommand{\KL}[2]{\mathrm{D}_{\mathrm{KL}}\left[#1 \Vert #2\right]}

\newcommand{\indicator}[1]{\mathds{1}_{#1}}

\newcommand{\entropy}{\mathrm{H}}

\newcommand{\multinomial}[2]{\mathcal{M}\left(#1, #2\right)}

\newcommand{\gaussian}[2]{\mathcal{N}\left(#1, #2\right)}
\newcommand{\gaussianat}[3]{\mathcal{N}\left(#1;#2, #3\right)}

\newcommand{\poisson}[1]{\mathcal{P}(#1)}

\newcommand{\eqsp}{\;}

\newcommand{\diag}{\mathrm{diag}}
\newcommand{\identity}{\mathbf{I}}
\newcommand{\EE}[2]{\mathds{E}_{#1}\left[#2\right]}

\newcommand{\trace}{\mathrm{tr}}

\newcommand{\Sigmahat}{\widehat{\bSigma}}
\newcommand{\simplex}[1]{\mathcal{S}^{#1}}

\newcommand{\taxonomy}{\mathcal{T}}
\newcommand{\child}[3]{\widecheck{\BF{#1}}{}^{#3}_{#2}}
\newcommand{\childindex}[2]{\mathcal{C}_{#2}^{#1}}
\newcommand{\parent}[3]{\widehat{#1}^{#3}_{#2}}

\newcommand{\ELBO}{\mathcal{L}}
\newcommand{\kk}[1]{k}
\newcommand{\pkk}[1]{k}
\newcommand{\jj}[1]{j}
\newcommand{\vv}[1]{v}
\newcommand{\distrib}[1]{{p_{\btheta,{#1}}}}
\newcommand{\posterior}[1]{{p_{\btheta,{#1}}}}
\newcommand{\prior}[1]{{p_{\btheta,{#1}}}}
\newcommand{\prioroffset}{{\distrib{\offset}{}}}
\newcommand{\varapprox}[1]{{q_{\bphi,{#1}}}}
\newcommand{\varapproxZ}[1]{{q_{\bphi,{#1}}^{\bZ}}}

\newcommand{\varapproxoffset}{{q_{\bphi}^{\offset}}}

\newcounter{hypH}
\newenvironment{hypH}{%
    \refstepcounter{hypH} 
    \noindent 
    (\textbf{H\arabic{hypH}}) 
}{%
    \par 
}

\newtheorem{theorem}{Theorem}
\newtheorem{proposition}{Proposition}
\newtheorem{corollary}[proposition]{Corollary}
\newtheorem{lemma}[proposition]{Lemma}

\begin{document}

\title[Article Title]{Tree-based variational inference for Poisson log-normal
models}


\author*[1]{\fnm{Alexandre} \sur{Chaussard}}\email{alexandre.chaussard@sorbonne-universite.fr}

\author[1]{\fnm{Anna} \sur{Bonnet}}\email{anna.bonnet@sorbonne-universite.fr}

\author[2]{\fnm{Elisabeth} \sur{Gassiat}}\email{elisabeth.gassiat@universite-paris-saclay.fr}

\author[1]{\fnm{Sylvain} \sur{Le Corff}}\email{sylvain.le\_corff@sorbonne-universite.fr}

\affil*[1]{\orgdiv{CNRS, Laboratoire de Probabilit\'es, Statistique et Mod\'elisation, LPSM}, \orgname{Sorbonne Universit\'e, F-75005 Paris, France}}

\affil[2]{\orgdiv{CNRS, Laboratoire de Math\'ematiques d'Orsay, LMO}, \orgname{Universit\'e Paris-Saclay, Orsay, France}}

\abstract{\input{abstract}}




\maketitle

\section{Introduction}
\label{sec:intro}
\input{intro}

\section{Background}
\label{sec:background}
\subsection{Notations}
\label{sec:notations}
\input{notations}
\input{background}

\input{plntree}

\section{Experiments}
\label{sec:exp}
\input{experiments}

\section{Discussion}
\label{sec:conclusion}
\input{conclusion}

\backmatter

\bmhead{Acknowledgements}
\input{acknowledgments}
\section*{Declarations}
\bmhead{Funding} PhD thesis public funds provided by Institute of Computing and Data Sciences (ISCD) from Sorbonne Universit\'e.
\bmhead{Author contributions} All authors conceived the ideas, contributed to the design methodology and investigated the formal analysis. Alexandre Chaussard developed the code, analyzed the data and wrote the first draft of the manuscript. Reviewing and editing has been performed by all authors.
\bmhead{Data availability} The microbiome dataset from \cite{pasolli2016machine} is available online on open access. The generated datasets are reproducible from the provided GitHub of PLN-Tree.
\bmhead{Conflict of interest} The authors declare no conflict of interest.

\bibliography{24-plntree.bib}

\onecolumn
\begin{appendices}
\input{appendix}




\end{appendices}



\end{document}

%% file: abstract.tex
When studying ecosystems, hierarchical trees are often used to organize entities based on proximity criteria, such as the taxonomy in microbiology, social classes in geography, or product types in retail businesses, offering valuable insights into entity relationships. Despite their significance, current count-data models do not leverage this structured information. In particular, the widely used Poisson log-normal (PLN) model, known for its ability to model interactions between entities from count data, lacks the possibility to incorporate such hierarchical tree structures, limiting its applicability in domains characterized by such complexities. To address this matter, we introduce the PLN-Tree model as an extension of the PLN model, specifically designed for modeling hierarchical count data. By integrating structured variational inference techniques, we propose an adapted training procedure and establish identifiability results, enhancing both theoretical foundations and practical interpretability. Experiments on synthetic datasets and human gut microbiome data highlight generative improvements when using PLN-Tree, demonstrating the practical interest of knowledge graphs like 
the taxonomy in microbiome modeling. Additionally, we present a proof-of-concept implication of the identifiability results by illustrating the practical benefits of using identifiable features for classification tasks, showcasing the versatility of the framework.

%% file: intro.tex
Count data appear in various domains, such as ecology, metagenomics, retail, actuarial sciences, and social sciences. Their discrete nature and potential for overdispersion present challenges that standard Gaussian-based methods are not adequately designed to address. While some approaches rely on ill-defined log-transformations to mitigate these issues before applying standard Gaussian models \citep{sparcc,spiec}, specialized modeling approaches are generally preferred for their robust statistical groundings \citep{o2010not}. Consequently, a range of probabilistic models has been developed to capture the intrinsic properties of count data \citep{hilbe2014modeling,inouye2017review}. Among these, the Poisson Log-Normal (PLN) model, originally introduced by \cite{aitchison1989multivariate} and developed by \cite{chiquet_pln}, has proven particularly effective. By embedding Gaussian latent variables within a Poisson framework, PLN models naturally account for overdispersion while providing an interaction network between entities, thereby offering a statistically sound and interpretable approach for count data analysis.

Besides, count data often exhibit hierarchical structures where observations are organized in a tree graph reflecting compositional relationships between entities at different levels of the hierarchy, like the taxonomy in ecology, the social classes in geography, or product types in marketing. In cases where no natural hierarchical structure is established in the domain, or when alternative clustering insights are desired, practitioners often employ tree-inference approaches \citep{come2021hierarchical,blei2003latent,teh2004sharing,momal2020} to organize and describe entities in a comprehensible graph that incorporates domain-specific knowledge. In various applications, incorporating relevant hierarchical structures has been considered to enhance statistical models, resulting in improved performances in most cases \citep{silverman2017phylogenetic,crawford2020incorporating,oliver2023taxahfe,jiang2025phylomix}. However, adhering strictly to predefined hierarchical structures can sometimes hinder model performance, as shown by \cite{bichat2020incorporating} in the context of controlling the false discovery rate for the detection of differentially abundant microbial bacteria. This suggests the need for flexible modeling approaches that can exploit underlying tree graphs without being overly dependent on their structure. Yet, despite the potential interest of such hierarchical structures for multivariate counts modeling, existing models like PLN do not explicitly account for them, limiting their applicability in scenarios where hierarchical dependencies play a crucial role. 

To address this limitation, we introduce the PLN-Tree model, an extension of the PLN framework tailored to handle hierarchical count data represented by tree graphs. The PLN-Tree model leverages a top-down hidden Markov tree structure to capture hierarchical dependencies among counts, enabling more accurate and interpretable modeling of count data in hierarchical settings. While the observed counts are controlled by the underlying hierarchical structure in the PLN-Tree framework, the model maintains flexibility through a latent Markov chain to parameterize the counts which is not confined to the tree structure. Like its PLN parent, learning PLN-Tree models via maximum likelihood estimation is intractable, but this challenge can be circumvented using variational inference techniques \citep{blei2017variational}. Hence, leveraging the true form of the posterior distribution, we propose a structured variational inference method based on backward Markov chains. This approach is motivated by recent work in reinforcement learning \citep{campbell2021online} and theoretical guarantees from \cite{chagneux2024additive}, which underscore the potential efficiency of these approximations. To ensure modeling flexibility and scalability, we opt for deep learning architectures by parameterizing the distributions with neural networks, allowing for efficient inference of the variational approximation using amortized backward inference. Moreover, we introduce a residual amortized backward recurrent neural network architecture to parameterize the variational approximation, which outperforms the traditional Gaussian mean-field approximation in our experiments.

To ensure the interpretability of the latent variables in practical applications, we investigate the identifiability of the proposed model. Previous works on structured models, such as \cite{gassiat2020identifiability, halva2021disentangling}, have demonstrated the ability to uniquely identify latent data models in the presence of Markov dependency structures. Thus, we establish a class of identifiability within the PLN-Tree structured framework, ensuring its applicability in demanding contexts where accurate and interpretable modeling of count data is crucial. Additionally, by leveraging features derived from our identifiability result, we suggest novel preprocessing transformations of count data, which are evaluated on a disease classification problem using human gut microbiome samples \citep{pasolli2016machine}, showing performance improvements over traditional preprocessing methods.

Finally, we perform a thorough exploration of the generative capacities of the PLN-Tree models, comparing it against state-of-the-art approaches such as PLN and SPiEC-Easi. Although numerous heuristics have been proposed in various fields \citep{sajjadi2018assessing,yang2020evaluation,betzalel2024evaluation}, assessing the performance of generative models remains inherently challenging as no universal standard exists. Thus, in our ecology-oriented context, we propose to leverage the widely used $\alpha$-diversity and $\beta$-diversity statistics—commonly employed to characterize ecosystems \citep{gotelli2001quantifying,thukral2017review}—to evaluate the generative performance of the models. Benchmarks on both synthetic data and human gut microbiome datasets from \cite{pasolli2016machine} demonstrates significant improvements in ecosystem generation with PLN-Tree when compared to traditional competitors that do not incorporate the underlying tree graph. These findings highlight the effectiveness of our model in capturing hierarchical dependencies and underscore the inherent value of taxonomic information in microbiome modeling.

The main contributions of this paper are summarized below.
\begin{itemize}
    \item We develop an extension of PLN models tailored to hierarchical count data called PLN-Tree. We introduce a variational approximation that leverages the structure of the posterior distribution, and propose a residual amortized deep neural network implementation that outperforms the traditional Gaussian mean-field. Additionally, we show that this framework can be adapted to incorporate covariates and counts offsets.
    \item We establish a class of identifiability for PLN-Tree models. The practical benefits of this result are illustrated empirically through an application to a disease classification problem using microbiome compositions. In particular, we show that identifiable latent features lead to improved performance over traditional preprocessing methods such as the centered log-ratio (CLR) transform or raw PLN features.
    \item Upon conducting extensive evaluation of the generative capabilities of our model on both synthetic and human gut microbiome datasets, we show that PLN-Tree outperforms the standard PLN model and other alternatives in reproducing the diversity of ecosystems, thereby highlighting the advantages of exploiting underlying hierarchical structures to model complex ecological systems.
\end{itemize}

This paper is organized as follows. Section~\ref{sec:background} provides background on the PLN framework and structured variational inference techniques motivating our model. Section~\ref{sec:plntree} introduces the proposed PLN-Tree models and variational training procedures. Then, Section~\ref{sec:identifiability} displays the identifiability results for tree-based PLN models, and introduces the identifiable features used in our experiments. Finally, Section~\ref{sec:exp} provides synthetic and real-world applications, comparing the proposed backward variational approximation with the mean-field variant and other state-of-art interaction-based count data models like SPiEC-Easi \citep{spiec} and PLN. We namely demonstrate the practical utility of PLN-Tree models through a generative benchmark on human gut microbiome data from \cite{pasolli2016machine}.
As a proof-of-concept, we illustrate the interest of identifiable features obtained from our theoretical results in a one-vs-all disease classification problem in Section \ref{sec:data_preprocessing_benchmark}. Our implementation and experiments are freely available on our GitHub\footnote{\urlstyle{tt}\url{https://github.com/AlexandreChaussard/PLN-Tree}}.

%% file: notations.tex
Let $\taxonomy$ be a finite rooted tree with $L$ layers, where each layer $\ell \leq L$ comprises $K_{\ell}$ nodes. A branch contains at least one node in each layer, so that every branch has a depth equal to $L$. At layer $\ell \leq L$, the random variable associated with node $k \leq K_\ell$ is denoted by $\rmV_k^\ell$. For layer $\ell \leq L-1$ and node $k \leq K_\ell$, the vector of children of the random variable $\mathrm{V}_k^\ell$ is indexed by $\childindex{\ell}{k}$ and represented as $\child{\bV}{k}{\ell}$ $= (\rmV_j^{\ell+1})_{j \in \childindex{\ell}{k}}$. We generally denote the hierarchical counts by $\bX$ and the associated latent variables by $\bZ$. In our model, the hierarchy can be partially used to exclude all layers above a given one. Thus, the first considered level of the hierarchy is not necessarily the root of the tree and will be denoted in bold $(\bX^1, \bZ^1)$. A graphical representation is provided in Figure~ \ref{fig:hierarchical_count_data_example}.

If the distribution of a random variable $\rmV$ has a density parameterized by $\btheta$ with respect to a reference measure, it is denoted by $\distrib{\rmV}$. When there is no possible confusion, we may express the density as $p_{\btheta}(\rmV)$. If $\btheta$ is a vector, its $k$-th coordinate is denoted by $\btheta_k$, while for a diagonal matrix $\btheta$, the $k$-th diagonal term is denoted as $\btheta_k$. For a function $f_{\btheta}$ parameterized by $\btheta$ and taking values in  $\mathbb{R}^d$, $d > 0$, the $k$-th coordinate of any of its outputs is denoted by $f_{\btheta,k}$. The sequence of random variables $(\rmV^1, \dots, \rmV^L)$ is represented as $\bV^{1:L}$. For $\bV \in \mathbb{R}^d$, the exponential of $\bV$ is defined as $\exp(\bV) = (\exp(\rmV_j))_{1\leq j \leq d}$, and the multivariate Poisson distribution with parameters $\bV \in \mathbb{R}^d_{> 0}$ is denoted by $\poisson{\bV} = \otimes_{j=1}^d \poisson{\rmV_j}$. We denote by $\simplex{d}$ the simplex of dimension $d$, then if $\bV \in \simplex{d}$, we denote the multinomial distribution with total count $n$ and probabilities $\bV$ by $\multinomial{n}{\bV}$. Finally, for $\bV \in \RR^d$ we denote its projection on the simplex through the softmax transform by $\softmax{\bV} = (\rme^{\rmV_i} / \sum_{j=1}^d \rme^{\rmV_j})_{1 \leq i \leq d}$. In this paper, the proposed model is parameterized by $\btheta$, the probability density function of the latent variable is $p_{\btheta}(\bZ)$ and referred to as the prior, the conditional probability density function of the observation given the latent variable is $p_{\btheta}(\bZ \mid \bX)$ and referred to as the posterior, and the variational approximation of this posterior density is parameterized by $\bphi$ and written $q_{\bphi}(\bZ \mid \bX)$.
\begin{figure}
\centering
\scalebox{0.7}{%
    \begin{tikzpicture}[level distance=1.5cm,
      level 1/.style={sibling distance=5cm},
      level 2/.style={sibling distance=2cm},
      level 3/.style={sibling distance=1cm},
      ]
      \node[circle,draw, fill=red!15] (root) {147}
        child {node[circle,draw,fill=orange!15] (L1K1) {72}
          child {node[circle,draw,fill=yellow!15] (L2K1) {12}
            child {node[circle,draw,fill=white!15] (L3K1){3}}
            child {node[circle,draw,fill=white!15] (L3K2){9}}
          }
          child {node[circle,draw,fill=yellow!15] (L2K2) {60}
            child {node[circle,draw,fill=cyan!15] (L3K3) {60}}
            child {node[circle,draw,fill=cyan!15] (L3K4) {0}}
          }
        }
        child {node[circle,draw,fill=orange!15] (L1K2) {75}
            child {node[circle,draw,fill=blue!15] (L2K3) {42}
                child {node[circle,draw,fill=green!15] (L3K5) {12}}
                child {node[circle,draw,fill=green!15] (L3K6) {30}}
            }
            child {node[circle,draw,fill=blue!15] (L2K4) {13}
                child {node[circle,draw,fill=gray!15] (L3K7) {13}}
            }
            child {node[circle,draw,fill=blue!15] (L2K5) {20}
                child {node[circle,draw,fill=purple!15] (L3K8) {8}}
                child {node[circle,draw,fill=purple!15] (L3K9) {0}}
                child {node[circle,draw,fill=purple!15] (L3K10) {12}}
            }
        };

        \draw[->] (root) -- (L1K1);
        \draw[->] (root) -- (L1K2);
        \draw[->] (L1K1) -- (L2K1);
        \draw[->] (L1K1) -- (L2K2);
        \draw[->] (L1K2) -- (L2K3);
        \draw[->] (L1K2) -- (L2K4);
        \draw[->] (L1K2) -- (L2K5);
        \draw[->] (L2K1) -- (L3K1);
        \draw[->] (L2K1) -- (L3K2);
        \draw[->] (L2K2) -- (L3K3);
        \draw[->] (L2K2) -- (L3K4);
        \draw[->] (L2K3) -- (L3K5);
        \draw[->] (L2K3) -- (L3K6);
        \draw[->] (L2K4) -- (L3K7);
        \draw[->] (L2K5) -- (L3K8);
        \draw[->] (L2K5) -- (L3K9);
        \draw[->] (L2K5) -- (L3K10);
        
        \draw[purple,thick,dotted] ($(L1K1.north west)+(-0.3,0.3)$)  rectangle ($(L1K2.south east)+(0.3,-0.3)$);
        \draw[purple] ($(L1K2.east)+(1.5, 0)$) node {$\bX^1 = (\rmX^1_1, \rmX^1_2)$};

        \draw[blue,thick,dotted] ($(L2K3.north west)+(-0.3,0.3)$)  rectangle ($(L2K5.south east)+(0.3,-0.3)$);
        \draw[blue] ($(L2K5.east)+(0.7, 0)$) node {$\child{\bX}{2}{1}$};
        \draw[blue,thick,dotted] ($(L2K1.north west)+(-0.3,0.3)$)  rectangle ($(L2K2.south east)+(0.3,-0.3)$);
        \draw[blue] ($(L2K1.east)+(-1.4, 0)$) node {$\child{\bX}{1}{1}$};

    \end{tikzpicture}
    }
\caption{Example of a hierarchical count data with $L = 4$. Nodes of the same color are independent of the other nodes conditionally to their parent node and their respective latent variables.}
\label{fig:hierarchical_count_data_example}
\end{figure}

%% file: background.tex
\subsection{Poisson log-normal models}
The Poisson-Log Normal model, introduced by \cite{aitchison1989multivariate} and thoroughly extended by \cite{chiquet_pln}, is a standard network inference model that has become popular due to its ability to handle over-dispersed count data and capture complex dependencies among variables.
In its simplest form, for a sample $i$, the PLN approach models the interactions through a Gaussian latent variable $\bZ_i \in \RR^d$, with mean $\bmu \in \RR^d$ and precision matrix $\bOmega \in \RR^{d \times d}$. The observed counts $\bX_i \in \RR^{d}$ are modeled by a Poisson distribution such that $(\bZ_i,\bX_i)_{1\leq i\leq n}$ are independent and, for $1\leq i \leq n$,  conditionally on $\bZ_i$ and $\rmX_{ik}$, $1\leq k \neq j\leq d$, $\rmX_{ij}$ depends on $\rmZ_{ij}$ only:
\begin{align*}
    \text{latent space}& \quad \bZ_i \sim \gaussian{\bmu}{\bOmega^{-1}}\eqsp, \\
    \text{counts space}& \quad \bX_i \mid \bZ_i \sim \poisson{\exp(\bZ_i)} \eqsp.
\end{align*}
In the PLN model, the precision matrix $\bOmega$ yields the interaction network, as entailed by the faithful correlation property provided in \cite{chiquet_pln}. On the other hand, the mean parameter $\bmu$ models the fixed effects in the environment, such as the natural disproportion of species in an ecosystem. Individual-related environmental effects can also be accounted for in $\bmu$ by making it a function of covariates, or by adding sampling effort information through an offset, which can have a significant impact on the faithfulness of the reconstructed network, as shown numerically in \cite{chiquet_pln_network}.

Performing maximum likelihood estimation in such latent data models is challenging as the conditional distribution of the latent variables given the observations is not tractable.
Variational estimation \citep{blei2017variational} is an appealing alternative to computationally intensive Monte Carlo methods by approximating the posterior using a family of variational distributions, yielding the Evidence Lower Bound (ELBO) as a suboptimal optimization objective \citep{kingma2019_vae}. Consequently, \cite{chiquet_pln} proposed an inference method for PLN models based on variational inference called variational Expectation Maximization (VEM), which consists in maximizing the ELBO in an alternate optimization resembling the Expectation-Maximization (EM) algorithm \citep{dempster1977maximum}, except that the posterior is replaced by its variational counterpart. In \cite{chiquet_pln_network}, the variational approximation corresponds to the Gaussian mean-field approximation, where each sample is parameterized by a unique mean and diagonal covariance matrix, unlike usual neural network parameterizations \citep{kingma2019_vae}. This specific form enables fast inference, as it yields exact maximization steps of the true parameters given the variational parameters, making the inference process highly stable, efficient, and computationally expedient. However, it affects the model scalability to larger datasets as the number of parameters increases linearly with the number of samples.

PLN models have been extended to various applications. For instance, \cite{chiquet_pln_network} tackles network inference via a sparsity-informed penalty based on the graphical LASSO \citep{graphical_lasso}. Other notable variants involve PLN-PCA (Principal Components Analysis), PLN mixtures, and PLN-LDA (Linear Discriminant Analysis), as thoroughly explored in \cite{chiquet_pln}. Although these variants can be adapted to our PLN-Tree framework, they are not explored in this paper.

\subsection{Variational inference for structured data}
\label{sec:variational_inference_structured}
As underscored in the previous section, addressing the parameter inference problem for PLN models can be achieved by leveraging variational inference techniques, which requires choosing a variational family.

In scenarios devoid of specific structural constraints, the Gaussian mean-field approximation emerges as the prevalent choice for variational families. This approach entails modeling each latent coordinate with independent Gaussian densities, offering the advantage of explicit ELBO computation when the latent prior is Gaussian. The mean-field approximation has demonstrated efficacy across various applications, such as the Poisson Log-Normal network inference model \citep{chiquet_pln_network} and in Variational Auto-Encoders (VAE) \citep{kingma2019_vae}. However, its inherent lack of expressivity and dependency modeling has encouraged the development of alternative variational families including full-covariance Gaussian models \citep{kingma2019_vae}, Gaussian mixture models with variational boosting \citep{millerfa17}, and normalizing flows within the latent space to enhance posterior expressiveness \citep{kobyzev2020normalizing}.

In this context, we set the focus to another class of variational approximations that explicitly incorporate data structures. These structured variational approximations can be formulated based on prior assumptions, as seen in approaches like NVAE \citep{vahdat2020nvae}, or by deriving insights from the posterior distribution, like auto-regressive models \citep{marino2018general} or hidden Markov models \citep{campbell2021online}. While prior-based assumptions are pertinent to methodological advancements, structuring the variational approximation based on the posterior aligns more closely with statistical principles while encouraging model interpretability \citep{arrieta2020explainable}. Notably, when the latent process follows a hidden Markov model, an enhanced variational approximation beyond the mean-field approach can be derived, as demonstrated by \cite{johnson2016composing}, further illustrated and extended in \cite{lin2018variational,halva2021disentangling,schneider2023learnable}. Our work is closely related to advancements in this area, particularly in the context of hidden Markov models, where recent studies like \cite{campbell2021online,chagneux2024additive} have highlighted the utility of backward variational inference, showcasing both empirical improvements and theoretical guarantees. In particular, \cite{chagneux2024additive} employ recurrent neural networks to amortize the backward process in their Gaussian processes illustrations, suggesting that similar architectures could be adapted for more complex data structures. Moreover, the theoretical underpinnings laid out in \cite{chagneux2024additive,gassiatlecorff2024} and \cite{campbell2021online} regarding backward variational inference in Markov chains offer compelling motivations for its application in our specific context.

%% file: plntree.tex
\section{PLN-Tree models and inference}
\label{sec:plntree}

\paragraph{Tree compositionality constraint} Hierarchical count data are generated through the repeated aggregation of counts from the deepest-level entities in the hierarchy, moving from the bottom to the top layer of the tree. Formally, this process involves placing the observed counts at the deepest level of the tree, then summing these counts with their respective siblings to compute the counts at their parent node, continuing this process recursively up to the root layer. This construction induces the following tree compositionality constraint
\begin{equation}
    \label{eq:tree_compositionality}
    \forall \ell < L, \forall k \leq K_\ell, \quad \rmX_k^\ell = \sum_{j \in \childindex{\ell}{k}} \rmX_j^{\ell+1} \eqsp,
\end{equation}
which needs to be accounted for in the modeling, thereby preventing an independent modeling of the layers. Furthermore, this constraint motivates a top-down propagation dynamic of the counts in the observed space, as a bottom-up approach would rely solely on the final layer to determine the entire hierarchical count data, thus failing to incorporate the tree structure in the modeling. 

\paragraph{PLN-Tree model} The PLN framework models tabular count data, which only applies to one layer of the tree at a time. Therefore, learning one PLN model at each layer does not satisfy the tree compositionality constraint \eqref{eq:tree_compositionality} since it models independent layers. Consequently, we propose a new model tailored to hierarchical structures named PLN-Tree.
\begin{itemize}
    \item The variables $(\bZ_i, \bX_i)_{1 \leq i \leq n}$ are independent, and for $1\leq \ell\leq L-1$, conditionally on $\bZ_i, (\bX_i^v)_{1 \leq v \leq \ell}$, the random variables $(\child{X}{ik}{\ell})_{1\leq k \leq K_{\ell}}$ are independent and the conditional law of $\child{X}{ik}{\ell}$ depends only on $\child{Z}{ik}{\ell}$ and $\rmX_{i\kk{}}^\ell$.
    
    \item The latent process $(\bZ^\ell)_{1\leq \ell \leq L}$ is a Markov chain with initial distribution $\bZ^1 \sim \gaussian{\bmu_1}{\bSigma_1}$ and such that for all $1\leq \ell \leq L-1$, the conditional distribution of $\bZ^{\ell + 1}$ given $\bZ^{\ell}$ is Gaussian with mean $\bmu_{\btheta_{\ell+1}}(\bZ^{\ell})$ and variance $\bSigma_{\btheta_{\ell+1}}(\bZ^{\ell})$, both arbitrary functions parameterized by $\btheta_{\ell+1}$. Formally, the latent process up to $\ell < L$ writes
    \begin{align*}
        &\bZ^1 \sim \gaussian{\bmu_1}{\bSigma_1} \eqsp, \\
        &\bZ^{\ell+1} \mid \bZ^{\ell} \sim \gaussian{\bmu_{\btheta_{\ell+1}}(\bZ^{\ell})}{\bSigma_{\btheta_{\ell+1}}(\bZ^{\ell})} \eqsp.
    \end{align*}
    The probability density function of $\bZ^1$ is denoted by $\prior{1}$ and for all $1\leq \ell \leq L-1$, the conditional probability density function of $\bZ^{\ell+1}$ given $\bZ^{\ell}$ is denoted by $\prior{\ell+1\mid\ell}(\cdot\mid \bZ^{\ell})$.
    \item Conditionally on $\bZ^1$, $\bX^1 \sim \poisson{\rme^{\bZ^1}}$ and for all $1\leq \ell \leq L-1$, $1\leq k \leq K_\ell$, conditionally on $\rmX_{\kk{}}^\ell$ and $\child{Z}{k}{\ell}$, $\child{X}{k}{\ell}$ has a multinomial distribution with parameters $\softmax{\child{Z}{k}{\ell}}$ and $\rmX_{k}^\ell$, where $\softmax{\cdot}$ is the softmax transform (see Section~\ref{sec:notations}). Formally, the observed counts process up to $\ell < L$ writes
    \begin{align*}
        &\bX^1 \mid \bZ^1 \sim \poisson{\rme^{\bZ^1}} \eqsp, \\
        \forall k \leq K_\ell, \quad &\child{\bX}{k}{\ell} \mid \rmX_k^\ell, \child{\bZ}{k}{\ell} \sim \multinomial{\rmX_{k}^\ell}{\softmax{\child{Z}{k}{\ell}}} \eqsp.
    \end{align*}  
    The conditional probability density function of $\bX^1$ given $\bZ^1$ is denoted by $\posterior{1}(\cdot \mid \bZ^1)$ and for all $1\leq \ell \leq L-1$ and all $1\leq k \leq K_{\ell}$, the conditional probability density function of $\child{\bX}{k}{\ell}$ given $(\rmX_k^\ell, \child{\bZ}{k}{\ell})$ is denoted by $\posterior{k,\ell}(\cdot \mid \rmX_k^\ell, \child{\bZ}{k}{\ell})$.
\end{itemize}
The joint density of the PLN-Tree model is then given by:
\begin{equation}
    \begin{split}
        p_{\btheta}(\bX, \bZ) &= \prior{1}(\bZ^1) \prod_{\ell=1}^{L-1} \prior{\ell+1\mid\ell}(\bZ^{\ell+1} \mid \bZ^{\ell}) \\
        &\hspace{-.6cm}\times \posterior{1}(\bX^1 \mid \bZ^1) \prod_{\ell=1}^{L-1} \prod_{k=1}^{K_{\ell}} \posterior{k,\ell}(\child{\bX}{k}{\ell} \mid \rmX_k^\ell, \child{\bZ}{k}{\ell}) \eqsp.
    \end{split}
    \label{eq:plntree_joint_density}
\end{equation}
The latent Markovian dynamics incorporates the top-down structure and enables the modeling of interactions between all nodes of a given layer, not just siblings. Conversely, the observed counts are constrained to satisfy the tree compositionality constraint (\ref{eq:tree_compositionality}), effectively integrating subgroup structures. In particular, the multinomial conditional distribution of the observations $\child{X}{k}{\ell}$ for $1 \leq \ell < L$ is the conditional distribution of independent Poisson random variables with parameters $\exp({\child{Z}{k}{\ell}})$ conditioned on the event \{$\sum_{\jj{+1} \in \childindex{\ell}{\kk{}}}\rmX_{\jj{+1}}^{\ell+1} = \rmX_{\kk{}}^\ell$\}, effectively extending PLN models to satisfy hierarchical constraints.

\paragraph{Variational inference} Under PLN-Tree models, the posterior distribution $p_{\btheta}(\bZ \mid \bX)$ is a backward Markov chain. Since we approximate this distribution using a variational approximation, we suggest using a variational backward Gaussian Markov chain \citep{campbell2021online,chagneux2024additive}, effectively extending the mean-field approximation to incorporate the Markov structure of the true posterior:
\begin{align}
    \label{eq:top_down_markov_variational_approx}
    \begin{split}
        \varapprox{1:L}(\bZ|\bX) &= \eqsp \varapprox{L}(\bZ^L | \bX^{1:L}) \\ &\times \prod_{\ell=1}^{L-1} \varapprox{\ell \mid \ell+1}(\bZ^\ell | \bZ^{\ell+1}, \bX^{1:\ell}) \eqsp,
    \end{split}
\end{align}
where $\varapprox{L}(\cdot | \bX^{1:L})$ is the Gaussian density with mean $\BF{m}_{\bphi^{L}}(\bX^{1:L})$ and diagonal variance $\bS_{\bphi^{L}}(\bX^{1:L})$ and $q_{\bphi,\ell \mid \ell+1}(\cdot | \bZ^{\ell+1}, \bX^{1:\ell})$ is the Gaussian density with mean $\BF{m}_{\bphi^{\ell}}(\bZ^{\ell+1}, \bX^{1:\ell})$ and diagonal variance $\bS_{\bphi^{\ell}}(\bZ^{\ell+1}, \bX^{1:\ell})$ both arbitrary functions parameterized by $\bphi^\ell$.

Using the backward variational approximation \eqref{eq:top_down_markov_variational_approx}, we can compute the surrogate objective given by the evidence lower bound (ELBO) of PLN-Tree models, for which the complete derivation is provided in Appendix \ref{proposition:elbo_pln_tree}. Interestingly, the PLN-Tree ELBO shares similarities with a per-layer PLN ELBO, where the latent variables $(\bZ^{\ell})_{1 \leq \ell \leq L}$ would be treated as independent across layers. However, PLN-Tree relaxes this independence assumption, incorporating Markov dependencies between layers. These dependencies are reflected in the ELBO, which is expressed only up to an expectation rather than in closed form. Additionally, the propagation of multinomial distributions across children groups introduces distinctive terms between the root layer ($\ell = 1$) and deeper layers, setting PLN-Tree apart from traditional PLN models. As a result, the PLN-Tree optimization objective exhibits a greater complexity than a layer-wise PLN.

\paragraph{Residual amortized architecture} Numerically, handling the inputs of the neural networks parameterizing the variational distributions is a challenging task due to the increasing dimension of the chains $(\bX^{1:\ell})_{1 \leq \ell \leq L}$, and the value it takes relatively to the latent variables. To address this scalability issue, \cite{chagneux2024additive} suggests performing amortized inference by encoding the chain of counts using a recurrent neural network. This architecture enables to control the number of parameters while neutralizing the increasing dimension of the input. Moreover, considering the current observation's pivotal influence on the latent variable distribution at layer $\ell$, we introduce a residual connection yielding $\bX^\ell$ as input of the current variational parameters. Combined with the amortized setting, this approach yields the residual amortized backward architecture illustrated in figure \ref{fig:res_backward_architecture}. Problem-specific networks must then be tuned, as thoroughly explored in our experiments in Section~\ref{sec:exp}. While we focus on the residual amortized backward for its superior empirical performances in our experiment, other noteworthy methods could be employed for the variational parameters in certain cases, like the regular amortized backward, or a strongly amortized variant taking only the current level observation as input and the next latent.
\begin{figure*}
\centering
    \begin{tikzpicture}[
        >=Stealth, 
        node distance=2cm,
        rnn/.style={draw, rectangle, rounded corners, fill=blue!20, minimum width=2cm, minimum height=1cm, align=center},
        ffnn/.style={draw, rectangle, rounded corners, fill=red!20, minimum width=2cm, minimum height=1cm, align=center},
        concat/.style={draw, circle, inner sep=0pt, minimum size=0.8cm, fill=gray!20}
    ]

    \node (X) {$\bX^{1:\ell}$};
    
    \node[rnn, right of=X] (RNN) {RNN};
    \node[right=0.75cm of RNN] (E) {$\bE^{\ell}$};
    
    \node[concat, right of=E] (C) {$\bigcup$};
    
    \node[ffnn, right=1.0cm of C] (FFNN) {Neural network};
    \node[right=0.75cm of FFNN] (Z) {$\BF{m}_{\bphi^\ell}(\bZ^{\ell+1}, \bE^\ell, \bX^\ell)$};

    \node[below=0.5cm of C] (Res) {$\bX^{\ell}$};
    \node[above=0.5cm of C] (Zp1) {$\bZ^{\ell+1}$};
    
    \draw[->] (X) -- (RNN);
    \draw[->] (RNN) -- (E);
    \draw[->] (E) -- (C);
    \draw[->] (Zp1) -- (C);
    \draw[->] (Res) -- (C);
    \draw[->] (C) -- (FFNN);
    \draw[->] (FFNN) -- (Z);
    
    \end{tikzpicture}
\caption{Residual amortized backward architecture for the variational mean at layer $\ell \leq L$. The amortizing Recurrent Neural Network is denoted by RNN, while the symbol "$\cup$" indicates a concatenation of entries. The variable $\bE^\ell$ is the last output of the recurrent network after inputting the sequence $\bX^{1:\ell}$.}
\label{fig:res_backward_architecture}
\end{figure*}
\paragraph{Partial closed-form optimization} Learning PLN-Tree models can be accelerated by exploiting the variational EM algorithm from \cite{chiquet_pln} applied at the first layer, which holds an explicit optimum in $\btheta^1$ when $\bphi^1$ is known, so that at iteration $h+1$,
\begin{equation}
    \begin{split}
        &\bmu_1^{(h+1)} = \frac{1}{n}\sum_{i=1}^n \EE{q_{\bphi}}{\BF{m}_{\bphi_1^{(h)}}(\bZ^{2}, \bX_i^{1})} \eqsp, \\
        &\bSigma_1^{(h+1)} = \frac{1}{n} \sum_{i=1}^n \mathds{E}_{q_{\bphi}}\bigg[\left(\bmu_1^{(h+1)} - \BF{m}_{\bphi_1^{(h)}}(\bZ^{2}, \bX_i^{1})\right) \\ &\eqsp \times \left(\bmu_1^{(h+1)} - \BF{m}_{\bphi_1^{(h)}}(\bZ^{2}, \bX_i^{1})\right)^{\top} + \bS_{\bphi_1^{(h)}}(\bX_i^{1:L})\bigg]. 
    \end{split}
    \label{eq:optimal_update_first_layer}
\end{equation}
The availability of these closed-form expressions is critical for practical model training, as they significantly accelerate the optimization of the ELBO.

\paragraph{Offset modeling} 
Collecting count data within multiple ecosystems usually comes with a variable sampling effort in practice. This offset in the average total count often originates from the counting protocol or the difficulty of exploring an environment. In genomics for instance, the total count relates to the sequencing depth of the genome, which correlates with the counts of rarer species, introducing a bias in the data with higher total count \citep{lee2014rare,xu2017low}. As a result, the offset often reflects sampling protocols rather than the ecological properties of the environments being studied, making them unreliable as direct features.

To mitigate these effects, preprocessing techniques such as resampling (rarefaction) can be applied to reduce the influence of variable sampling efforts, albeit with some loss of data \citep{mcmurdie2014waste,weinroth2022considerations,schloss2024rarefaction}. An alternative approach is to model the offset directly within the statistical framework to avoid introducing spurious correlations \citep{chiquet_pln_network}. In PLN models \citep{chiquet_pln}, the offset is handled via a plug-in estimator that shifts the latent variable means based on the log of the total count in each sample. Extending this idea, we propose modeling the offset as a latent variable following a Gaussian mixture in the PLN-Tree framework. This formulation captures variability in sampling efforts both across different groups of samples and within groups, resolving the need for domain-specific assumptions. The flexibility of this approach comes with the introduction of an hyperparameter (the number of mixture components), which allows users to tailor the model to different offset scenarios but increases the complexity of parameter estimation during training. Interestingly, since the softmax is invariant by constant translation, adding the offset in the lower layers of the observed dynamics has no impact on the modeling, restricting its usage to the root layer. Details on the suggested variational approximation and the associated ELBO for PLN-Tree models with offset modeling can be found in Appendix~\ref{sec:plntree_elbo_offset}.

\paragraph{Conditional PLN-Tree with covariates}
In many practical settings, the distribution of the counts might be partially determined by environmental factors, often referred to as exogenous factors or covariates. For example, in microbiome analyses, covariates such as a patient’s diet have been shown to significantly influence the structure of the gut microbiome \citep{ross2024interplay}. Therefore, incorporating this exogenous information in PLN-Tree would allow to model how external factors affect the ecosystem's dynamics while possibly enhancing the generative capabilities of our model.
In the PLN framework, \cite{chiquet_pln} propose to model the impact of the covariates as linear effects on the mean of the latent variables, while keeping the interaction network independent of these covariates. Despite its simplicity, this approach effectively enhances the interaction network by eliminating spurious edges \citep{chiquet_pln_network}, while maintaining closed-form estimators to ensure efficient training. Extending their framework, we suggest to incorporate covariates into the PLN-Tree model at the initial layer using the same methodology, while covariates would also be integrated in each transition density with non-linear effects, as made possible by the arbitrary parameterization of the latent process. Formally, for \( n \) samples, let \( \bC = [{\bC_i}]_{1 \leq i \leq n} \in \RR^{n \times p} \) represent the covariates and \( \BF{B} \in \RR^{p \times K_1} \) denote the regression coefficients to be learned for the first layer, then the conditional latent dynamic is defined as follows for each sample \( i \leq n \) and for each level \( \ell < L \):
\begin{align*}
    &\bZ^{1}_i \sim \gaussian{\bC_i\BF{B}}{\bSigma_1} \eqsp, \\ &\bZ^{\ell+1}_i \mid \bZ^{\ell}_i \sim  \gaussian{\bmu_{\btheta_{\ell+1}}(\bZ^{\ell}_i, \bC_i)}{\bSigma_{\btheta_{\ell+1}}(\bZ^{\ell}_i, \bC_i)}.
\end{align*}
Similarly, covariates are incorporated with non-linear effects into the parameters of the variational approximations to estimate the conditional posterior distribution \( p_{\btheta}(\bZ \mid \bX, \bC) \). This conditional PLN-Tree model is then optimized using a conditional ELBO, derived analogously to Proposition \ref{proposition:elbo_pln_tree}, with all parameters now conditioned on \( \bC \), and a closed-form optimum for the regression coefficients at iteration $h+1$ derived from \eqref{eq:optimal_update_first_layer} as the ordinary least square estimator:
\[
\BF{B}^{(h+1)} = (\bC^\top \bC)^{-1} \bC^\top \mathbf{M}_{\bphi}(\bZ^2, \bX^1) \eqsp,
\]
with \(\mathbf{M}_{\bphi}(\bZ^2, \bX^1) = [\BF{m}_{\bphi_1^{(h)}}(\bZ^{2}_i, \bX^{1}_i)^\top]_{1 \leq i \leq n}\).
While the integration of covariates is straightforward within our framework, developing an efficient deep architecture for both the latent process and the variational approximation presents significant practical efforts. This complexity arises from the need to combine covariates with count data and latent variables as inputs to neural networks despite their heterogeneous nature. Among the extensive literature that covers specific combinations of data sources, attention-based models \citep{niu2021review} would appear as compelling options in our context (see \cite{gong2023multi}).

\section{Identifiability of PLN-Tree}
\label{sec:identifiability}
In a nutshell, identifiability ensures we can uniquely determine a model given the data, and thus infer the law of the latent variables solely from the law of the observations. In real-world applications, it was shown that the lack of identifiability can severely undermine performances  \citep{d2022underspecification}, and precludes the interpretability of the inferred networks. Fortunately, in many applications such as in \cite{halva2021disentangling,gassiat2020identifiability}, the dependency structure of the data can disentangle parameters using inductive biases. This section presents two identifiability results related to the PLN model and the PLN-Tree extension.

\paragraph{Latent Poisson model identifiability}
Lemma~\ref{lemma:poisson_markov_identifiability} shows the identifiability of models with  positive latent variables combined with a Poisson emission law. Corollaries involves the identifiability of PLN models and that of the first layer of PLN-Tree models.
\begin{lemma}
\label{lemma:poisson_markov_identifiability}
    Let $K > 0, \bZ = (\rmZ_k)_{1 \leq k \leq K}$  be a random variable supported on $(\RR^*_+)^K$. Assume that the observations $\bX = (\rmX_k)_{1 \leq k \leq K}$ are such that for all $1\leq k \leq K$, conditionally on $\bZ$ and $ (\rmX_v)_{1 \leq v \neq k \leq K}$, $\rmX_k$ follows a Poisson distribution with parameter $\rmZ_k$ and is independent of $(\rmZ_v, \rmX_v)_{1 \leq v \neq k \leq K}$. Then, the law of $\bZ$ is identifiable from the law of $\bX$.
\end{lemma}
\begin{proof}
    Proof is postponed to Appendix~~\ref{app:poisson_markov_identifiability}
\end{proof}

\paragraph{PLN-Tree identifiability} The previous result does not cover the whole scope of the PLN-Tree framework, as it yields the identifiability of independent layers conditionally to their respective latent variables at most. Instead, Theorem \ref{theorem:plntree_identifiability} establishes the identifiability of PLN-Tree models up to a softmax transform for the first three layers. 
\begin{theorem}
\label{theorem:plntree_identifiability}
    Let $\taxonomy$ a given tree, $\bZ = (\rmZ^1, \bZ^2, \bZ^3)$  be random variables such that $\rmZ^1> 0$, $\bZ^2 \in \simplex{K_2}$, for all $k \leq K_2, \child{Z}{k}{2} \in \simplex{\# \childindex{2}{k}}$. Suppose the observations $\bX = (\rmX^1, \bX^2, \bX^3)$ are such that: 
    \begin{itemize}
    \item conditionally on $\rmZ^1$, $\rmX^1$ has a Poisson distribution with parameter $\rmZ^1$;
    \item conditionally on $(\rmX^1,\bZ^2)$, $\bX^2 \sim \multinomial{\rmX^1}{\bZ^2}$;
    \item conditionally on $(\bX^2,\bZ^3)$, for all $1\leq k \leq K_2$, $\child{X}{k}{2} \sim \multinomial{\rmX^2_k}{\child{Z}{k}{2}}$, and $\child{X}{k}{2}$ is independent of $(\child{X}{j}{2})_{j\neq k}$.
    \end{itemize}
    Then, the law of $(\rmZ^1, \bZ^2, \bZ^3)$ is identifiable from the law of $(\rmX^1,\bX^2, \bX^3)$.
\end{theorem}
\begin{proof}
    Proof is postponed to Appendix~~\ref{proof:plntree_identifiability}.
\end{proof}
This identifiability result recursively extends to models with multiple conditionally independent Poisson roots and deeper multinomial dynamics within the hierarchy thanks to the Markov tree structure.

However, since the softmax function is constant along diagonals, obtaining the identifiability of $(\bZ^1, \dots, \bZ^L)$ is not a given if we do not set a constraint on the parameters space. Combining the previous result with Lemma \ref{lemma:identifiability_projected_softmax} in Appendix shows we can identify the law of the latent variables up to a linear projection.
Assuming the distribution of the latent variables is a Gaussian Markov chain, a direct application of the previous result yields the identifiability of every parent-children distribution of the PLN-Tree framework providing the parameters belong to a defined projection space.
\begin{corollary}
\label{corollary:identifiability_plntree_projected}
    Let $(\bZ^1, \bZ^2)$ and $(\tilde{\bZ}^1, \tilde{\bZ}^2)$ in $\RR^m \times \RR^d$ be such that  conditionally on $\bZ^1$ (resp. $\tilde \bZ^1$), $\bZ^2$ is Gaussian with mean $\bmu(\bZ^1)$ (resp. $\tilde \bmu(\tilde\bZ^1)$) and covariance $\bSigma(\bZ^1)$ (resp. $\tilde \bSigma(\tilde \bZ^1)$). Define $\bP = \identity_d - d^{-1} \BS{1}_{d\times d}$ the projector on $\vect{\BS{1}_d}^\perp$. Assume $(\bZ^1, \softmax{\bZ^2})$ has the same law as $(\tilde\bZ^1, \softmax{\tilde\bZ^2})$, then
    $$ 
    \bP \bmu(\BF{z}) = \bP \tilde \bmu(\BF{z}) \quad \mathrm{and}\quad  \bP \bSigma(\BF{z}) \bP = \bP \tilde\bSigma(\BF{z}) \bP\eqsp, 
    $$ 
    $\mathbb{P}_{\bZ^1}-a.s. \eqsp,$ where $\mathbb{P}_{\bZ^1}$ is the law of $\bZ^1$.
\end{corollary}
\begin{proof}
    Proof is postponed to Appendix~~\ref{proof:identifiability_plntree_projected}.
\end{proof}
For all $\ell \geq 2$, denoting by $\bP^\ell = \diag(\{\bP_k^\ell\}_{1 \leq k \leq K_{\ell-1}})$ with 
$$
\bP_k^\ell = \identity_{\# \childindex{\ell-1}{k}} - \frac{1}{\# \childindex{\ell-1}{k}} \BS{1}_{\# \childindex{\ell-1}{k} \times \# \childindex{\ell-1}{k}}\eqsp,
$$
we obtain from Theorem \ref{theorem:plntree_identifiability} and Corollary \ref{corollary:identifiability_plntree_projected} that all PLN-Tree model parameterized by the latent variables $(\bZ^1, \bP^2 \bZ^2, \dots, \bP^L \bZ^L)$ are identifiable.

\paragraph{Latent interactions modeling}
While the latent Gaussian process in PLN-Tree encodes layer-specific interactions through the covariance matrices of its Markov chain, interaction networks become interpretable only once projected into the identifiable space. For instance, conditionally on the previous latent variable, a diagonal latent covariance maps to a block-diagonal identifiable covariance whose blocks align exactly with the hierarchy’s subgroups. More generally, any latent block-diagonal covariance whose blocks follow this partition preserves its block structure after projection, thus confining latent interactions to the clusters defined by the hierarchy. Additional interaction patterns can be explored by imposing structural constraints on the identifiable parameters, namely by exploring sparse precision matrices or low-rank covariances as in \cite{chiquet_pln,chiquet_pln_network}.

\paragraph{Using identifiable features as counts preprocessing}
Using latent variables as inputs for machine learning tasks is a standard practice that can significantly improve performance. In the case of PLN-Tree, Theorem~\ref{theorem:plntree_identifiability} suggests that the identifiable latent variables $\bZ^1, (\bP^\ell\bZ^\ell)_{2\leq \ell \leq L}$ may provide  meaningful representations. This encoding process moves the data from a constrained and discrete space to a real-valued hyperplane satisfying the scale invariance principle formalized in \cite{aitchison1994principles}, thus making the latent features 
potentially strong candidates for tasks such as classification, PCA, or regression in the context of compositional data. However, it is difficult to directly associate a latent variable with a specific entity in the tree, rendering comparisons with the regular PLN impractical.

Based on this remark, we introduce a latent feature, referred to as the latent proportions (LP), which maps hierarchical count data to their latent representation $\bV$ such that: 
\begin{align} 
    \label{eq:latent_tree_counts}
    \begin{split}
        &\bV^1 = \softmax{\bZ^1} \eqsp, \\
        \forall \ell < L, k \leq K_\ell, \quad &\child{V}{k}{\ell} = \softmax{\child{Z}{k}{\ell}} \times \mathrm{V}_k^\ell\eqsp. 
    \end{split}
\end{align}
Since the latent proportions are compositional in nature, they can be further transformed using standard log transforms commonly employed in compositional data analysis \citep{ibrahimi2023overview}, such as the centered log-ratio (CLR) transform. By combining the LP with the CLR transform (LP-CLR), we can map the observed counts from their constrained compositional space into an unconstrained latent space, which can improve the performance of machine learning models. It can also serve as a foundation for estimating covariance matrices at different layers and for conducting network inference. Similarly, since Lemma \ref{lemma:pln_identifiability} establishes the identifiability of the latent variables in PLN models, an LP-CLR identifiable feature can be derived for PLN models by applying the CLR transform to the softmax of its latent variables, resulting in projecting the latent variables onto $\vect{\mathbf{1}_d}^\perp$. We refer to this LP-CLR feature derived from the PLN models as Proj-PLN.

The proposed identifiable LP-CLR transform of PLN and PLN-Tree latent features are benchmarked against the raw PLN features, true proportions, and their CLR transform of counts in Section \ref{sec:data_preprocessing_benchmark}.

%% file: experiments.tex
Across three generative benchmarks, we demonstrate the advantages of incorporating the underlying tree graph structure in count data modeling over unstructured approaches:
\begin{itemize}
    \item Section~\ref{sec:synthetic_data} involves a synthetic dataset generated along a PLN-Tree model to illustrate the performances of the backward variational approximation against the regular mean-field approximation, and its limits in an ideal inference framework. 
    \item Section~\ref{sec:markov_dirichlet} employs hierarchical count data generated from a Markovian Dirichlet procedure extended from the simulation protocol proposed in \cite{chiquet_pln_network}. This second experiment enables to evaluate PLN-Tree against non-hierarchical competitors in a fair yet controlled hierarchical setup.
    \item In Section~\ref{sec:microbiome_generation}, we evaluate the generative performances of the models on real metagenomics data from gut microbiome samples of disease-affected patients \citep{pasolli2016machine}, and demonstrate the practical benefits of exploiting taxonomic information to model metagenomics count data.
\end{itemize}
Secondly, Section~\ref{sec:data_preprocessing_benchmark} offers a practical perspective on the identifiability results presented in Section~\ref{sec:identifiability}. By leveraging identifiable features from PLN-based models, we demonstrate that such transform can serve as an alternative to standard compositional preprocessing methods, leading to improved performance in a one-versus-all disease classification task using metagenomics data.

\paragraph{Benchmarked models}  To assess the performance of PLN-Tree as a generative model, we benchmark it against other interaction-based count data models. However, state-of-art models like PLN \citep{chiquet_pln}, SparCC \citep{sparcc} or SPiEC-Easi \citep{spiec} are restricted to tabular data, allowing the modeling of only one layer of hierarchical count data at a time. Thankfully, by leveraging the hierarchical compositional constraint  \eqref{eq:tree_compositionality}, tabular count data models can generate valid hierarchical count data by modeling only the last layer of the tree, which is usually the one at stake for practitioners. This generative procedure involves sampling the abundances of the last layer under a given model and then exploiting the compositional constraints to derive the values of the parent nodes, allowing us to obtain hierarchical count data that satisfies  \eqref{eq:tree_compositionality}. 

In our experiments, PLN baselines are computed using the \textit{pyPLNmodels}\footnote{\urlstyle{tt}\url{https://github.com/PLN-team/pyPLNmodels}} Python implementation from \cite{batardiere2024pyplnmodels}. Conversely, SparCC and SPiEC-Easi were implemented within our package as generative models, as both methods usually only estimate the covariance and precision matrices of the log-centered ratio (CLR) transformation of compositional data. After estimating the mean of the normalized and CLR-transformed count data, we sample from the inferred Gaussian distribution and invert the CLR transformation using the softmax function, obtaining proportion data that can be used to generate count data via a multinomial distribution. Additionally, since our model does not involve sparsity, we set the sparsity parameter of the estimated matrices to $0$ in both SparCC and SPiEC-Easi, making both models equivalent. Consequently, we only compare PLN-Tree to PLN and SPiEC-Easi. Finally, in this benchmark, we compare the efficiency of the proposed backward approximation \eqref{eq:top_down_markov_variational_approx} against the regular Gaussian mean-field \citep{blei2017variational}, denoted as PLN-Tree (MF), with variational density 
$$
q_{\bphi}^{\mathrm{MF}}(\bZ \mid \bX) = \prod_{\ell=1}^{L} \gaussianat{\bZ^{\ell}}{\bm_{\bphi,\ell}^{\mathrm{MF}}(\bX^{\ell})}{\bS^{\mathrm{MF}}_{\bphi,\ell}(\bX^{\ell})} \eqsp,
$$
where for all $\ell \leq L$, $\bS^{\mathrm{MF}}_{\bphi}(\bX^{\ell})$ is diagonal positive definite. The PLN-Tree tag is retained for the residual backward variational approximation modeling. All  generated datasets and parameters are stored on the GitHub repository for reproducibility.

\paragraph{Metrics for model evaluation} 
In the context of variational deep generative models, comparing the quality of estimated parameters is often impractical due to variations in model architectures, which adds up to identifiability concerns in neural networks. Instead, we assess the generative performance of trained models by their ability to replicate the distribution of the original dataset faithfully. Yet, the evaluation of generative models remains inherently difficult, and while several metrics have been introduced in specific applications \citep{sajjadi2018assessing,yang2020evaluation,betzalel2024evaluation}, there is ultimately no consortium to the best of our knowledge. Therefore, to assess the generative performances in our ecology-oriented context, we use $\alpha$-diversity and $\beta$-diversity metrics that are commonly employed in ecosystem studies, as well as agnostic metrics such as the empirical Wasserstein on normalized counts (proportion hierarchical data) and correlation measures. 

First, $\alpha$-diversity metrics provide insights into species richness and evenness, thereby partially characterizing the diversity within an ecosystem (see Appendix \ref{sec:alpha_diversity} and \cite{gotelli2001quantifying}). Among these, the Shannon entropy and the Simpson index are widely employed. 
The Shannon index quantifies the uncertainty in predicting the entities in the ecosystem, while the Simpson index represents the probability that two entities chosen at random represent the same entity. Both estimators are qualified as robust and quantify complementary aspects of the ecosystems \citep{nagendra2002opposite}.
Our first objective is to ensure that the generated data closely approximates the $\alpha$-diversity distribution of the original dataset, as measured by the Wasserstein distance. Other distances or divergences are considered in the appendix for each experiment, such as the Kullback-Leibler divergence, Kolmogorov-Smirnov statistic, and total variation distance.

While $\alpha$-diversity metrics evaluate the intrinsic statistics of one ecosystem, $\beta$-diversity metrics enable the quantitative comparison of the composition of two ecosystems (see Appendix \ref{sec:beta_diversity}). These metrics are often referred to as dissimilarity measures, taking values between $0$ and $1$ to indicate the degree of dissimilarity between pairs of samples. Among the $\beta$-diversity metrics, the UniFrac \citep{lozupone2005unifrac} and Jaccard diversities can account for the hierarchical nature of the data, while the Bray-Curtis dissimilarity \citep{beals1984bray}, commonly applied in microbiological studies \citep{kleine2021data}, operates at a single level of the hierarchy. To ensure that the benchmark remains independent of the underlying tree structure, we restrict our assessment to the Bray-Curtis dissimilarity to evaluate the quality of the generations at each layer of the tree. To compare the $\beta$-diversity, we draw $n = 100$ samples from the true dataset and from the trained model, and compute the dissimilarity between each pair of samples. Repeating that sampling process $m = 50$ times, we obtain $m$ symmetric dissimilarity matrices of shape $n \times n$. For each matrix, we perform PERMANOVA \citep{anderson2014permutational} and PERMDISP \citep{anderson2006distance} to test respectively whether the centroids and the dispersions of the two groups are the same. Both tests are performed $m$ times on $1000$ permutations, providing finally $m$ associated p-values for each test, the distribution of which will assess the dissimilarity between original and generated data. PERMANOVA and PERMDISP tests are detailed in Appendix~\ref{sec:testing_beta_diversity} and implemented in the \texttt{scikit-bio}\footnote{\url{https://github.com/scikit-bio/scikit-bio}} package.

Finally, to compare the distribution of the generated data with the original data, we evaluate the empirical Wasserstein distance between generated samples and the initial dataset in normalized forms (proportion hierarchical data) at each layer using the \texttt{emd2} function from \textit{POT} \citep{flamary2021pot}. Additionally, we employ correlation measures between the original data and their reconstructions to assess the quality of the variational approximations at the reconstruction task.  Computational efficiency between implementations is discussed in Appendix \ref{app:experiments_setup}.

\paragraph{Selection of the variational architectures} To provide a comprehensive and equitable evaluation of the PLN-Tree variants, we determine efficient architectures for the variational approximations tailored to each experimental scenario. To that end, we propose several network architectures and assess their generative capabilities, leveraging the above evaluation metrics. The model demonstrating superior overall performance is identified by averaging its rank across all computed metrics. The considered architectures and numerical considerations are detailed in Appendix~\ref{app:experiments_setup}. Since the models are trained using variational approximations, convergence may result in different model parameters depending on the initialization. Specifically, the analysis of training variability in Appendix \ref{app:performance_benchmark_plntree_synthetic} reveals that the mean-field approximation is less stable compared to the proposed residual backward approach, but this does not affect the performance ranking of the two methods. Consequently, training is conducted once for each model, and performance variability is assessed based on the generations.

\subsection{Synthetic data}
\label{sec:synthetic_data}
\subsubsection{PLN-Tree retrieval}
To evaluate the efficiency of the proposed backward variational approximation \eqref{eq:top_down_markov_variational_approx}, we conduct an initial study on data generated from a PLN-Tree model. We begin by defining a tree illustrated on Figure \ref{fig:synthetic_tree_plntree}, a reference PLN-Tree model with parameters $\btheta^*$, and a synthetic dataset $(\bX, \bZ)$ generated using the PLN-Tree dynamic specified in Section \ref{sec:plntree} with $\btheta = \btheta^*$ (see Figure \ref{fig:abundance_samples_synthetic_plntree}), consisting of $n = 2000$ samples. In our experiments, we ensure that the latent dynamic is parameterized by identifiable parameters as detailed in Section~\ref{sec:identifiability}. Upon selecting candidate architectures (see Appendix~\ref{app:experiments_setup_synthetic_plntree}), we conduct the training procedure for each model until convergence. Then, we generate data by sampling $M = 25$ times $2000$ samples from the trained models and aggregate the results to address sampling variability. The considered tree of Figure \ref{fig:synthetic_tree_plntree} has a small depth and not many leaf nodes for computational speed reasons, but it is sufficient to explore scenarios of interest in this benchmark. A comparison between the artificial data and real-word microbiome dataset from \cite{pasolli2016machine} in terms of $\alpha$-diversity is provided in Appendix \ref{fig:comparison_alpha_diversity_datasets}.

\paragraph{PLN-Tree successfully outperforms others under its model}
We start our evaluation by analyzing the performance on the synthetic dataset using $\alpha$-diversity metrics, summarized in Table \ref{tab:synthetic_plntree_benchmark_alpha_diversityj_wasserstein} using Wasserstein distance (see other distances in Table \ref{tab:synthetic_plntree_benchmark_alpha_diversity} in Appendix~\ref{app:performance_benchmark_plntree_synthetic}). As anticipated, PLN-Tree models exhibit superior performance compared to the other method, with the backward variational approximation outperforming the mean-field variant despite being in an amortized setting. Upon delving into the layers of the tree, we observe a gradual decrease in performance across all criteria in the PLN and SPiEC-Easi models, attributable to the Markov tree propagation of the counts, a factor not accounted for by these approaches.
\begin{table*}
    \centering
    \begin{tabular}{@{}lcccc@{}}
    \toprule
    {\textbf{$\alpha$-diversity}} & \textbf{PLN-Tree} & \textbf{PLN-Tree (MF)} & \textbf{PLN} & \textbf{SPiEC-Easi}\\
    \midrule
    \multicolumn{5}{c}{\textbf{Wasserstein Distance} ($\times 10^{2}$)} \\ \midrule
    Shannon $\ell=1$ & \textbf{1.57} (0.50) & 11.23 (0.73) & 14.64 (1.15) & 46.72 (1.63) \\
    Shannon $\ell=2$ & \textbf{3.67} (1.33) & 5.14 (1.20) & 32.04 (1.62) & 89.62 (2.31) \\
    Shannon $\ell=3$ & \textbf{5.82} (1.51) & 7.86 (1.47) & 35.03 (1.68) & 98.49 (2.31) \\
    Simpson $\ell=1$ & \textbf{0.62} (0.21) & 2.69 (0.27) & 4.91 (0.41) & 15.91 (0.65) \\
    Simpson $\ell=2$ & \textbf{0.71} (0.24) & 1.40 (0.31) & 7.35 (0.41) & 22.13 (0.72) \\
    Simpson $\ell=3$ & \textbf{0.85} (0.24) & 1.55 (0.34) & 7.21 (0.41) & 22.05 (0.70) \\
    \bottomrule
    \end{tabular}
\caption{Wasserstein distance between $\alpha$-diversity distributions from synthetic data sampled under the original PLN-Tree model and simulated data under each model trained, averaged over the samplings, with standard deviation.}
\label{tab:synthetic_plntree_benchmark_alpha_diversityj_wasserstein}
\end{table*}

\begin{table*}
    \centering
    \begin{tabular}{@{}lcccc@{}}
    \toprule
    & \textbf{PLN-Tree} & \textbf{PLN-Tree (MF)} & \textbf{PLN} & \textbf{SPiEC-Easi}\\
    \midrule
    \multicolumn{5}{c}{\textbf{Wasserstein Distance} ($\times 10^{2}$)} \\
    \midrule
    $\ell=1$ & \textbf{5.20} (0.62) & 8.61 (0.11) & 10.70 (0.34) & 24.21 (0.73) \\
    $\ell=2$ & \textbf{13.01} (0.14) & 16.37 (0.29) & 17.59 (0.28) & 31.35 (0.67) \\
    $\ell=3$ & \textbf{14.08} (0.13) & 18.13 (0.32) & 20.04 (0.03) & 37.36 (0.87) \\
    \bottomrule
    \end{tabular}
    \caption{Empirical Wasserstein distance between normalized synthetic data sampled under the original PLN-Tree model and normalized simulated data under each modeled trained, for each layer, averaged over the trainings, with standard deviation.}
    \label{tab:synthetic_plntree_benchmark_samples}
\end{table*}

Analyzing $\beta$-diversity through PERMANOVA and PERMDISP tests (see Figure \ref{fig:synthetic_experiments_pvalues}) reveals that, at the deepest layer ($\ell = L$), the centroids and dispersions of PLN and SPiEC-Easi significantly deviate from the original data. Specifically, the rejection rates at $5\%$ significance level are $82\%$ and $96\%$ for PLN, and $100\%$ for both tests applied to SPiEC-Easi. In contrast, PLN-Tree models with backward approximation exhibits rejection rates of only $8\%$ for PERMANOVA and $6\%$ for PERMDISP, suggesting that our model preserves the $\beta$-diversity patterns of the original data compared to the competing methods. Interestingly, the mean-field approximation of PLN-Tree displays a considerably higher rejection rate of around $90\%$ for both tests. At upper layers ($\ell < L$), the backward PLN-Tree model continues to be accepted, demonstrating its robustness across the hierarchy. In comparison, the acceptance rate of PLN  improves from $18\%$ at $\ell = L$ to $80\%$ at $\ell = 1$, while SPiEC-Easi remains consistently rejected across all layers at the 95\% confidence level. These results highlight the consistency and improved performance of our method in modeling hierarchical $\beta$-diversity and the specific interest of the backward approximation over the mean-field approach.
\begin{figure}
    \centering
    \begin{subfigure}[b]{0.98\linewidth}
        \centering
        \includegraphics[width=\linewidth]{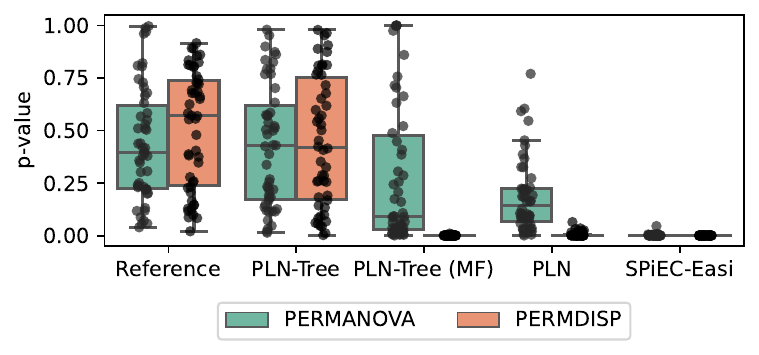} 
        \caption{$\ell = 1$}
    \end{subfigure}
    \hfill
    \begin{subfigure}[b]{0.98\linewidth}
        \centering
        \includegraphics[width=\linewidth]{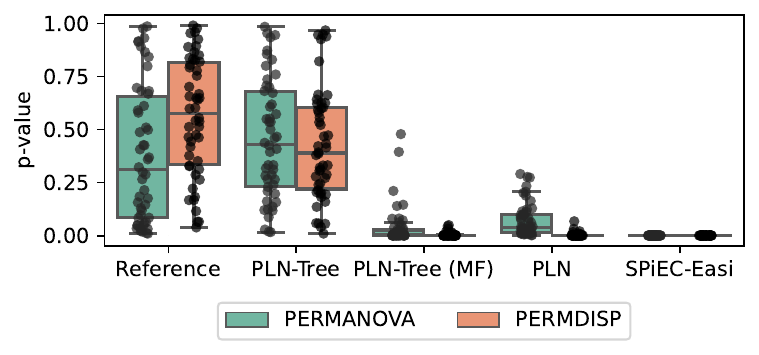} 
        \caption{$\ell = 2$}
    \end{subfigure}
    
    \begin{subfigure}[b]{0.98\linewidth}
        \centering
        \includegraphics[width=\linewidth]{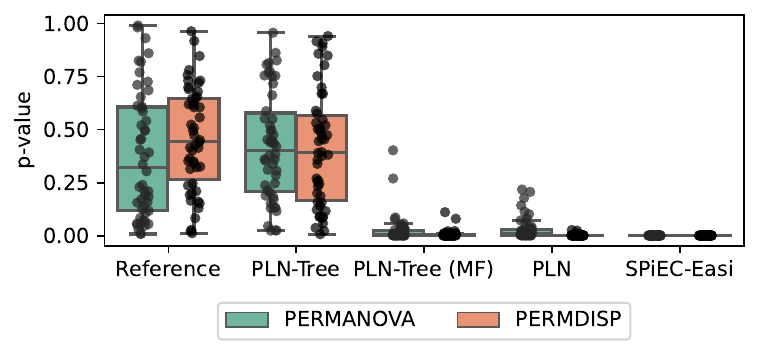} 
        \caption{$\ell = 3$}
    \end{subfigure}
    \hfill
    \caption{$p$-values for PERMANOVA and PERMDISP tests applied on Bray Curtis dissimilarities (layer-wise) computed between $100$ generated data with each model and $100$ sampled PLN-Tree generated data from the training dataset, repeated $50$ times. Reference model corresponds to generated data from the original model to assess the bootstrap variability.}
    \label{fig:synthetic_experiments_pvalues}
\end{figure}

Additionally, Table \ref{tab:synthetic_plntree_benchmark_samples} demonstrates that PLN-Tree-based approaches consistently approximate the distribution of the proportions of the entities at each depth of the tree, contrasting with the other approaches, which exhibits a noticeable performance decline as we descend the tree, matching with the $\alpha$-diversity observations. Looking at the encoders performance in Table \ref{tab:synthetic_plntree_benchmark_encoding}, it appears the backward approximation conserves more information than the mean-field approach in an ideal PLN-Tree framework on unseen samples, illustrating the upside of considering the backward Markov structure of the true posterior for model inference.
\begin{table}
    \centering
    \begin{tabular}{@{}lcccc@{}}
    \toprule
    & \textbf{PLN-Tree} & \textbf{PLN-Tree (MF)} \\
    \midrule
    $\ell=1$ & \textbf{0.999} (0.002) & 0.901 (0.209) \\
    $\ell=2$ & \textbf{0.993} (0.050) & 0.910 (0.137) \\
    $\ell=3$ & \textbf{0.996} (0.020) & 0.990 (0.028) \\
    \bottomrule
    \end{tabular}
\caption{Correlation between reconstructed counts and the test dataset ($1000$ samples) from the original PLN-Tree model, averaged over the samples, with standard deviation.}
\label{tab:synthetic_plntree_benchmark_encoding}
\end{table}

\subsubsection{Artificial data from Markovian Dirichlet}
\label{sec:markov_dirichlet}
\paragraph{Markov-Dirichlet model and benchmark parameters}
In order to provide fair comparisons of the performances of each model in a controlled setup, we simulate hierarchical count data from a process unrelated to PLN framework, extended from the synthetic experiments protocol of \cite{chiquet_pln_network}. First, we define a hierarchical tree fixing the dataset structure, illustrated on Appendix \ref{fig:synthetic_tree_markov_dirichlet}. Then, the steps of the generative process are defined as follows.
\begin{itemize}
    \item \textbf{Base network generation.} Sample an adjacency matrix $\BS{G} \in \mathcal{M}_{K_1 \times K_1}$ using a random graph model like Erdös-R\'enyi (no particular structure), preferential attachment (scale-free property) or affiliation models (community structure). Choose $u,v > 0$ to control the partial correlation and conditioning of the network at the first layer, and deduce a precision matrix $\bOmega = v\BS{G} + \mathrm{diag}(|\mathrm{min}(\mathrm{eig}(v\BS{G}))| + u)$. In our experiments, $v = 0.3$ and $u = 0.1$.
    \item \textbf{First counts generation.} Draw counts $\BF{a} \in \NN^{K_1}$ such that $\log(\BF{a}) \sim \gaussian{\bmu}{\bOmega}$. Compute a probability vector $\BS{\pi} = \softmax{\BF{a}}$ and draw a sampling effort $\mathrm{N} = \exp(\offset)$ from a negative binomial distribution. We obtain the counts of the first layer using a multinomial distribution $\bX^1 \sim \multinomial{\mathrm{N}}{\BS{\pi}}$.
    \item \textbf{Counts propagation.} For each $k \leq K_1$, compute $\balpha_k^1(\bX^1) \in \RR^{\# \childindex{1}{k}}_{>0}$, where $\balpha_k^1(.)$ is an arbitrary function, like a neural network with softplus output in our experiments. Sample weights $\omega_k^1 \in \simplex{\#\childindex{1}{k}}$ from a Dirichlet of parameters $\balpha_k^1(\bX^1)$. Draw the counts of the children of the node $k$ using a multinomial with total count $\rmX_k^1$ and probabilities $\omega_k^1$. Repeat that procedure for the next layers using the counts of the previous layer.
\end{itemize}

To derive the covariance matrix of the first layer, we generate a random adjacency matrix using the Erdös-Rényi graph model. In our architecture, for all layers $\ell$ up to $L$ and nodes $k$ up to $K_{\ell}$, $\alpha^\ell_k$ is structured as a one-layer network with softplus output and a random weight matrix. We set the sampling effort to $N = 20000$ and sample $n = 2000$ hierarchical count data points constituting our synthetic dataset. Following the selection of candidate architectures (detailed in Appendix~\ref{app:experiments_setup_synthetic_markovdirichlet}), we conduct a single training procedure for each model. Subsequently, we sample data from the trained models $M = 10$ times and aggregate the results to address sampling variability. As for the previous synthetic experiment, a comparison between the artificial data and real-word microbiome samples from \cite{pasolli2016machine} in terms of $\alpha$-diversity is performed in Appendix \ref{fig:comparison_alpha_diversity_datasets}.

\paragraph{PLN-Tree outperforms others in hierarchical scenarios} We provide a summary of the model performances in Table~\ref{tab:synthetic_markovdirichlet_benchmark_alpha_diversity_wasserstein},  Table~\ref{tab:synthetic_markovdirichlet_benchmark_samples},  (see Table \ref{tab:synthetic_markovdirichlet_benchmark_alpha_diversity} for other distances), and Table~\ref{tab:synthetic_markovdirichlet_benchmark_encoding}. Notably, PLN-Tree models exhibit superior performance compared to PLN and SPiEC-Easi approaches, which do not account for the underlying Markovian tree structure of the data. Similar to our previous synthetic experiment, we observe that as we delve deeper into the tree structure, the performance of PLN and SPiEC-Easi deteriorates significantly. When looking at the $\alpha$-diversity in Table \ref{tab:synthetic_markovdirichlet_benchmark_alpha_diversity}, the backward variational approach demonstrates superior performance compared to the mean-field approach, which is supported by its higher efficiency at the reconstruction task on unseen samples summarized in Table \ref{tab:synthetic_markovdirichlet_benchmark_encoding}. 
\begin{table*}
    \centering
    \begin{tabular}{@{}lcccc@{}}
    \toprule
{\textbf{$\alpha$-diversity}} & \textbf{PLN-Tree} & \textbf{PLN-Tree (MF)} & \textbf{PLN} & \textbf{SPiEC-Easi} \\
\midrule
\multicolumn{5}{c}{\textbf{Wasserstein Distance} ($\times 10^{2}$)} \\ \midrule
Shannon $\ell=1$ & \textbf{17.70} (0.47) & 21.42 (0.59) & 72.27 (1.70) & 125.10 (1.25) \\
Shannon $\ell=2$ & \textbf{22.23} (0.94) & 29.10 (1.06) & 111.53 (1.81) & 177.18 (1.50) \\
Shannon $\ell=3$ & \textbf{24.32} (0.83) & 37.72 (1.14) & 142.28 (1.99) & 224.07 (1.62) \\
Simpson $\ell=1$ & \textbf{5.69} (0.16) & 5.84 (0.16) & 21.74 (0.60) & 39.01 (0.46) \\
Simpson $\ell=2$ & \textbf{5.21} (0.17) & 5.90 (0.19) & 26.70 (0.59) & 46.26 (0.54) \\
Simpson $\ell=3$ & \textbf{3.91} (0.11) & 5.16 (0.16) & 28.55 (0.59) & 50.12 (0.54) \\
    \bottomrule
    \end{tabular}
\caption{Wasserstein distance on the distribution of $\alpha$-diversity at each layer computed between synthetic data sampled under the Markov Dirichlet model and simulated data under each modeled trained, averaged over the trainings, with standard deviation.}
\label{tab:synthetic_markovdirichlet_benchmark_alpha_diversity_wasserstein}
\end{table*}
\begin{table*}
    \centering
    \begin{tabular}{@{}lcccc@{}}
    \toprule
    & \textbf{PLN-Tree} & \textbf{PLN-Tree (MF)} & \textbf{PLN} & \textbf{SPiEC-Easi} \\
    \midrule
    \multicolumn{5}{c}{\textbf{Wasserstein distance} ($\times 10^{2}$)} \\ \midrule
    $\ell=1$ & \textbf{11.51} (0.25) & 12.47 (0.30) & 25.50 (0.59) & 41.84 (0.50) \\
    $\ell=2$ & \textbf{19.68} (0.25) & 22.02 (0.36) & 43.26 (0.61) & 59.09 (0.55) \\
    $\ell=3$ & \textbf{24.33} (0.24) & 27.15 (0.30) & 51.84 (0.57) & 68.21 (0.52) \\
    \bottomrule
    \end{tabular}
\caption{Empirical Wasserstein distance between normalized synthetic data sampled under the  Markov Dirichlet model and normalized simulated data under each modeled trained, for each layer, averaged over the trainings, with standard deviation.}
\label{tab:synthetic_markovdirichlet_benchmark_samples}
\end{table*}
\begin{table}
    \centering
    \begin{tabular}{@{}lcccc@{}}
    \toprule
    & \textbf{PLN-Tree} & \textbf{PLN-Tree (MF)} \\
    \midrule
    $\ell=1$ & \textbf{0.995} (0.062) & 0.967 (0.103) \\
    $\ell=2$ & \textbf{0.989} (0.065) & 0.967 (0.078) \\
    $\ell=3$ & \textbf{0.987} (0.075) & 0.973 (0.087) \\
    \bottomrule
    \end{tabular}
\caption{Correlation between reconstructed abundances and the test dataset from the Markov Dirichlet model ($1000$ samples), averaged over the samples, with standard deviation.}
\label{tab:synthetic_markovdirichlet_benchmark_encoding}
\end{table}
The results of the $\beta$-diversity tests, presented in Figure \ref{fig:markovdirichlet_experiments_pvalues}, reveal a 100\% rejection rate for non-tree-based methods at the 5\% significance level for both PERMANOVA and PERMDISP tests, confirming their inability to capture $\beta$-diversity patterns in this hierarchical context. Among PLN-Tree methods, the backward approximation shows a notably lower rejection rate ($4\%$ to $58\%$ for PERMANOVA) compared to the mean-field approach ($2\%$ to $84\%$ for PERMANOVA), highlighting the residual backward approximation superiority over the mean-field in learning PLN-Tree models. However, PERMDISP tests at the deepest layer ($\ell = L$) reveal a 100\% rejection rate for all models, indicating that even PLN-Tree methods still struggle to fully capture $\beta$-diversity patterns at the deepest levels in this particular hierarchical dynamic defined by the Markov Dirichlet framework.
\begin{figure}
    \centering
    \begin{subfigure}[b]{0.98\linewidth}
        \centering
        \includegraphics[width=\linewidth]{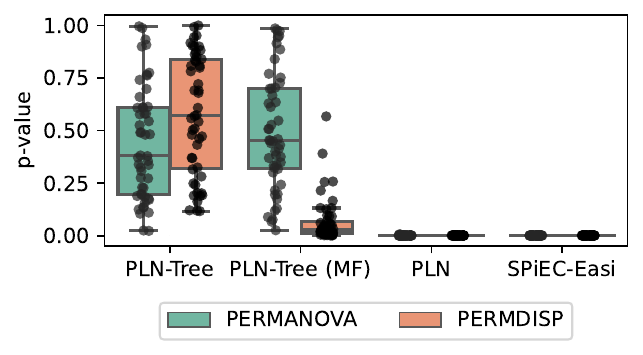} 
        \caption{$\ell = 1$}
    \end{subfigure}
    \hfill
    \begin{subfigure}[b]{0.98\linewidth}
        \centering
        \includegraphics[width=\linewidth]{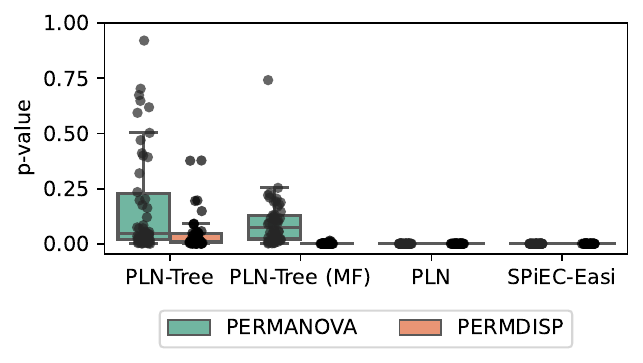} 
        \caption{$\ell = 2$}
    \end{subfigure}
    
    \begin{subfigure}[b]{0.98\linewidth}
        \centering
        \includegraphics[width=\linewidth]{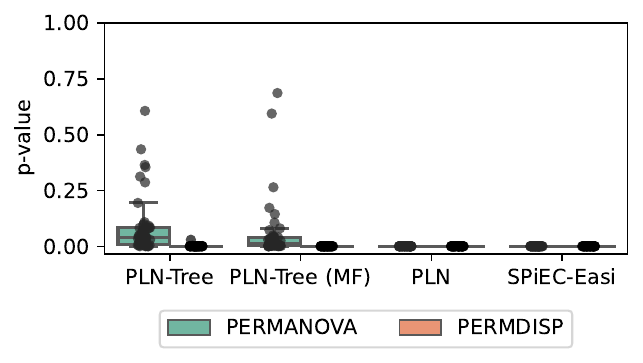} 
        \caption{$\ell = 3$}
    \end{subfigure}
    \hfill
    \caption{$p$-values for PERMANOVA and PERMDISP tests applied on Bray Curtis dissimilarities (layer-wise) computed between $100$ generated data with each model and $100$ sampled Markov Dirichlet generated data from the training dataset, repeated $50$ times.}
    \label{fig:markovdirichlet_experiments_pvalues}
\end{figure}

Thus, this experiment demonstrates the inability for non-tree-based method to capture count data distributions in hierarchical context, as well as the interest of considering the backward structure of the true posterior when doing variational inference to learn PLN-Tree. However, progress is still to be made for PLN-Tree methods to fully capture counts distributions in generalized hierarchical context.

\subsection{Metagenomics dataset: application to the gut microbiome}
\label{sec:metagenomics_exp}
\paragraph{Description of the dataset and preprocessing} We assess the efficacy of the PLN-Tree model using a metagenomics dataset introduced in \cite{pasolli2016machine}. This dataset comprises microbial compositions from both control individuals and patients with various diseases, totaling $3610$ samples. Our analysis focuses exclusively on the gut microbial compositions of disease-associated patients, as recapitulated in Table \ref{tab:metagenomics_dataset_desc}. Each sample is characterized by hierarchical proportion data, with the base tree representing the taxonomy of Archaea, Eukaryota, and Bacteria. Sequencing was conducted using MetaPhlAn2, optimized for bacterial sequencing \citep{truong2015metaphlan2}, thus restricting our study to bacteria. Besides, for computational speed reasons, we limit our analysis to the layers of the taxonomy comprised between the second and fifth layers, which respectively correspond to the \textit{class} and the \textit{family}, as these levels yield sufficient performance disparities between the considered models of this benchmark. To convert the proportions of taxa within each patient's gut into count data, we sample counts from a multinomial distribution with a total count of $\exp(12)$ and gut sample compositions as probabilities, allowing for fair comparison of $\alpha$-diversity and $\beta$-diversity by fixing the total count \citep{schloss2024rarefaction}. Additionally, we implement prevalence filtering using a threshold of $1\times\rme^{-12}$ to filter very rare Operational Taxonomic Units (OTUs). The considered taxonomy after prevalence filtering is illustrated in Figure \ref{fig:taxonomy_metagenomics_and_samples}.
\begin{table*}
    \centering
    \begin{tabular}{lcc|c}
    \toprule
    {\textbf{Label}} & \textbf{Nb of training samples} & \textbf{Nb of test samples} & \textbf{Total}\\
    \midrule 
    IBD (Crohn) & 20 & 5 & 25 \\
    Colorectal Cancer & 38 & 10 & 48\\
    Leaness & 71 & 18 & 89\\
    Liver Cirrhosis & 94 & 24 & 118 \\
    IBD (UC) & 118 & 30 & 148 \\
    Obesity & 131 & 33 & 164 \\
    Type 2 Diabetes & 178 & 45 & 223 \\
    \midrule
    \textbf{Total} & 650 & 165 & 815 \\
    \bottomrule
    \end{tabular}
    \caption{Metagenomics dataset considered in our experiments, extracted from \cite{pasolli2016machine}. The samples are drawn randomly for each label to satisfy these counts.}
    \label{tab:metagenomics_dataset_desc}
\end{table*}
\begin{figure}
    \centering
    \begin{subfigure}[b]{1\linewidth}
        \centering
        \includegraphics[width=1\linewidth]{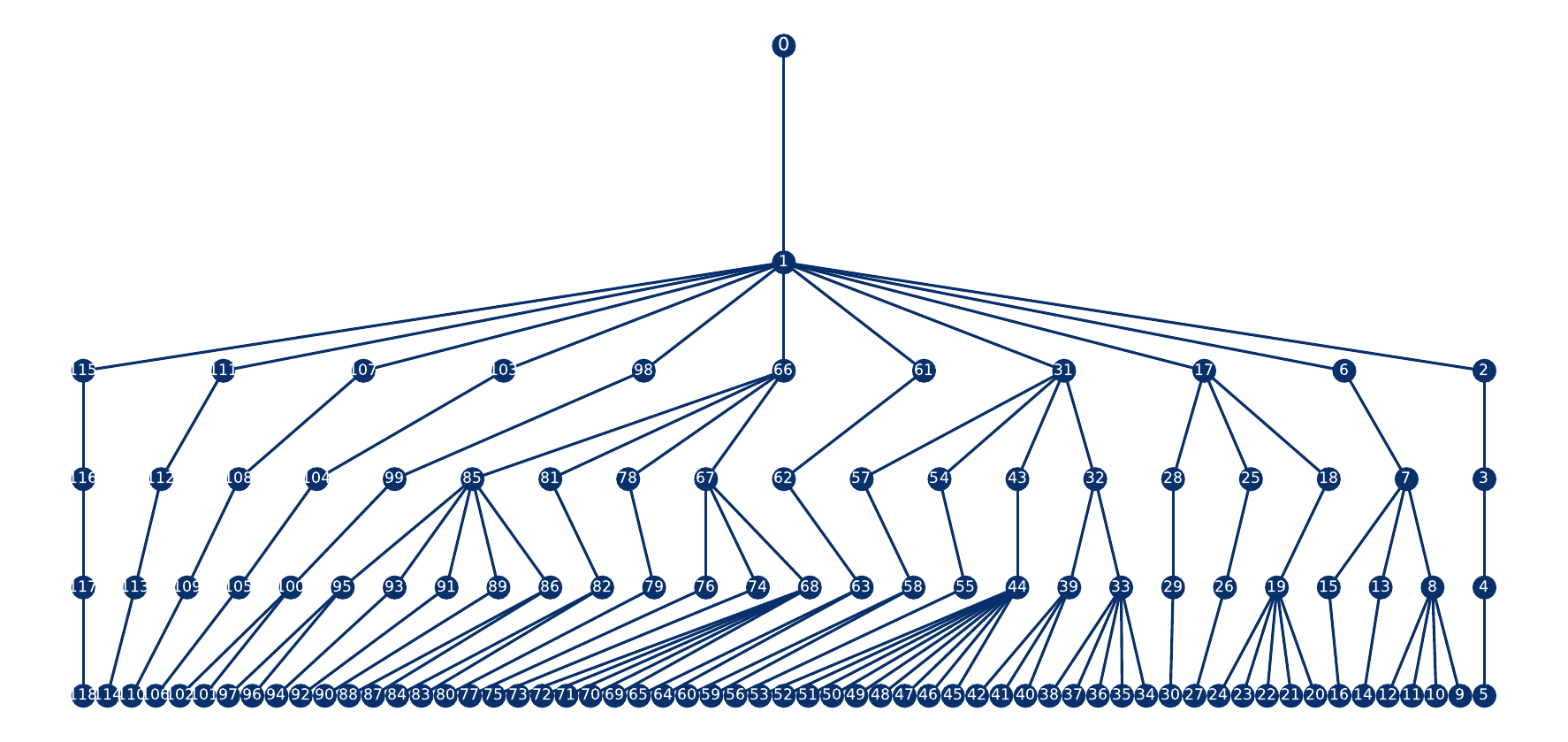}
        \label{fig:metagenomics_taxonomy}
    \end{subfigure}
    \\
    \begin{subfigure}[b]{1\linewidth}
        \centering
        \includegraphics[width=1\linewidth]{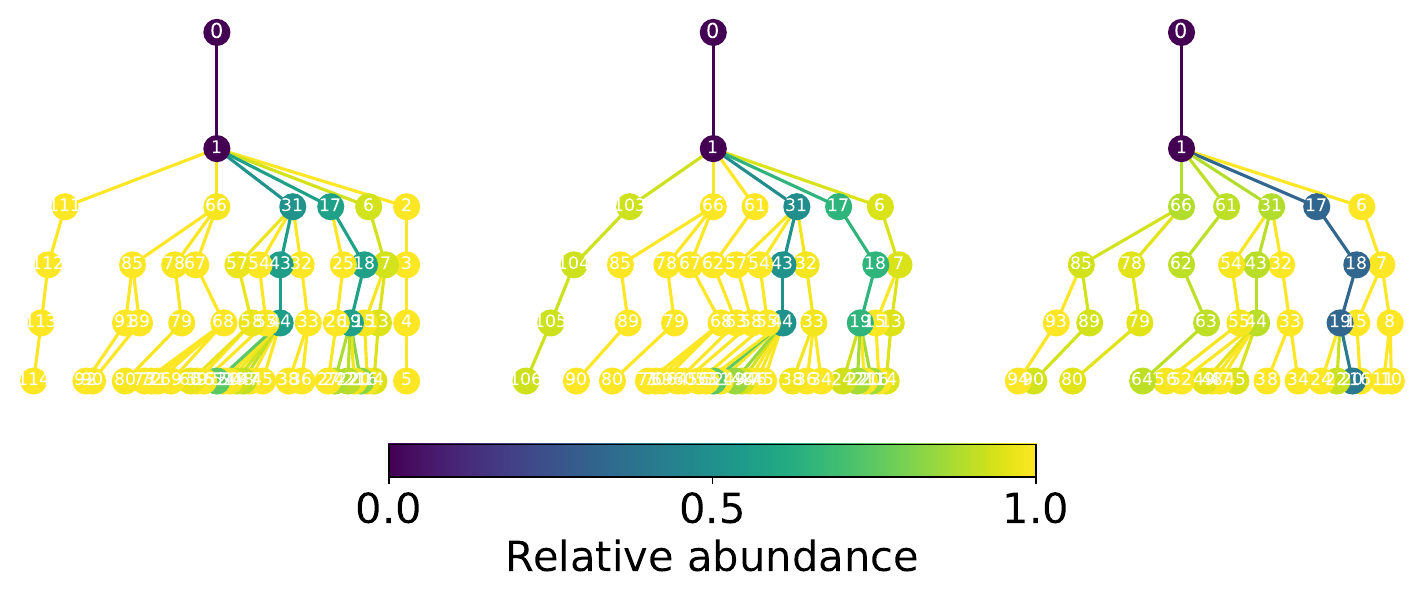}
        \label{fig:abundance_samples_metagenomics}
    \end{subfigure}
    \caption{Graph of the taxonomy considered in the metagenomics experiments (top), and four samples from the dataset (bottom). Indices on the nodes identify taxa and matches with the boxplots indexing in Appendix \ref{fig:metagenomics_abundances_boxplot}.}
    \label{fig:taxonomy_metagenomics_and_samples}
\end{figure}

\subsubsection{Generating microbiome compositions with PLN-Tree}
\label{sec:microbiome_generation}
We provide a summary of the tested and selected architectures for the PLN-Tree models, in Appendix \ref{app:experiments_setup_metagenomics}. Each compared model is trained once, while sampling is repeated $M = 25$ with $2000$ samples to account for sampling variability in the model evaluation. 

\paragraph{Exploiting taxonomic relationships improves synthetic microbiome realism} We provide a summary of the model performances in Tables \ref{tab:metagenomics_benchmark_samples} and \ref{tab:metagenomics_benchmark_alpha_diversity_wasserstein}, while Figure \ref{fig:metagenomics_abundances_boxplot} illustrates the variability of the generations for each model. Notably, the tree-based models exhibit superior performance for $\alpha$-diversity and distribution of proportions modeling compared to state-of-art approaches, which do not account for the taxonomy. Specifically, as we delve deeper into the tree structure, the performance of PLN models declines, while PLN-Tree models maintain consistency with depth. At the deepest layer ($\ell = L$) in Figure \ref{fig:metagenomics_pvalues}, rejection rates obtained from the PERMANOVA test at the 5\% significance level show that the PLN-Tree model with backward approximation is rejected in only 36\% of the tests, compared to 48\% with the mean-field approximation, while PLN and SPiEC-Easi are always rejected. Similarly, for the PERMDISP test, PLN-Tree with backward approximation is rejected in only 8\% of tests, compared to 46\% with the mean-field approach, while the other methods are consistently rejected. These results suggest that PLN-Tree models provide a significantly better approximation of the original $\beta$-diversity than PLN and SPiEC-Easi. For layers above the deepest ($\ell < L$), the acceptance rate for PLN-Tree residual backward model continues to rise over $80\%$ on average, whereas the benchmark models remain largely rejected for both tests, showing only marginal improvements. This highlights the robustness and consistency of PLN-Tree models across different layers of the taxonomy.
\begin{table*}
\centering
\begin{tabular}{@{}lcccc@{}}
\toprule
& \textbf{PLN-Tree} & \textbf{PLN-Tree (MF)} & \textbf{PLN} & \textbf{SPiEC-Easi} \\
\midrule
\multicolumn{5}{c}{\textbf{Wasserstein distance} ($\times 10^{2}$)} \\ \midrule
$\ell=1$ & 5.89 (0.29) & \textbf{4.67} (0.25) & 15.57 (0.59) & 36.16 (1.02) \\
$\ell=2$ & 8.83 (0.28) & \textbf{7.55} (0.14) & 20.65 (0.71) & 42.52 (1.17) \\
$\ell=3$ & 9.27 (0.27) & \textbf{7.76} (0.12) & 20.86 (0.70) & 42.75 (1.16) \\
$\ell=4$ & 17.00 (0.22) & \textbf{15.59} (0.13) & 29.29 (0.72) & 56.19 (0.88) \\
\bottomrule
\end{tabular}
\caption{Empirical Wasserstein distance between normalized metagenomics data and normalized simulated data under each modeled trained, for each layer, averaged over the trainings, with standard deviation.}
\label{tab:metagenomics_benchmark_samples}
\end{table*}
\begin{table*}
    \centering
    \begin{tabular}{@{}lcccc@{}}
    \toprule
{\textbf{$\alpha$-diversity}} & \textbf{PLN-Tree} & \textbf{PLN-Tree (MF)} & \textbf{PLN} & \textbf{SPiEC-Easi} \\
\midrule
\multicolumn{5}{c}{\textbf{Wasserstein distance} ($\times 10^{2}$)} \\ \midrule
Shannon $\ell=1$ & \textbf{1.73} (0.44) & 3.00 (0.44) & 16.49 (1.14) & 43.12 (1.57) \\
Shannon $\ell=2$ & \textbf{2.22} (0.73) & 5.70 (0.97) & 23.21 (1.64) & 57.73 (2.02) \\
Shannon $\ell=3$ & \textbf{2.29} (0.63) & 6.58 (1.02) & 23.96 (1.67) & 59.16 (2.00) \\
Shannon $\ell=4$ & \textbf{2.08} (0.62) & 20.39 (1.08) & 55.32 (2.38) & 127.11 (3.03) \\
Simpson $\ell=1$ & 0.84 (0.14) & \textbf{0.71} (0.12) & 7.18 (0.48) & 17.99 (0.71) \\
Simpson $\ell=2$ & 0.92 (0.24) & \textbf{0.73} (0.19) & 7.49 (0.57) & 19.59 (0.81) \\
Simpson $\ell=3$ & 0.91 (0.23) & \textbf{0.72} (0.19) & 7.46 (0.57) & 19.50 (0.80) \\
Simpson $\ell=4$ & \textbf{0.53} (0.13) & 2.41 (0.21) & 12.91 (0.67) & 31.62 (0.99) \\
    \bottomrule
    \end{tabular}
\caption{Wasserstein distance on $\alpha$-diversity distributions computed between metagenomics data and simulated data under each modeled trained, averaged over the trainings, with standard deviation. The best-performing model in each row is indicated in bold.}
\label{tab:metagenomics_benchmark_alpha_diversity_wasserstein}
\end{table*}
\begin{figure*}
    \centering
    \begin{subfigure}[b]{0.48\linewidth}
        \centering
        \includegraphics[width=\linewidth]{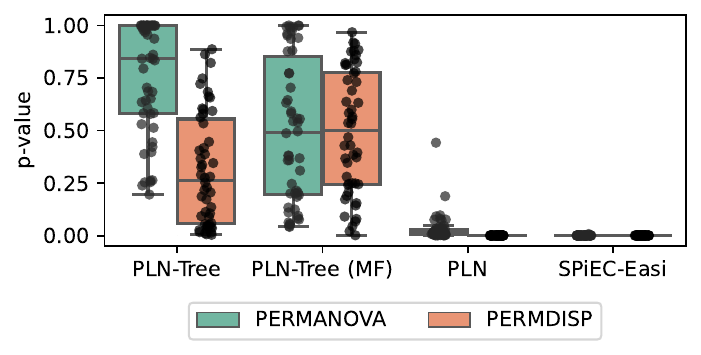} 
        \caption{$\ell = 2$}
    \end{subfigure}
    \hfill
    \begin{subfigure}[b]{0.48\linewidth}
        \centering
        \includegraphics[width=\linewidth]{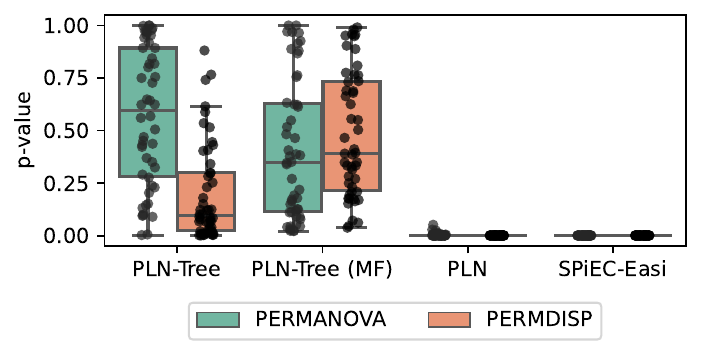} 
        \caption{$\ell = 3$}
    \end{subfigure}
    
    \begin{subfigure}[b]{0.48\linewidth}
        \centering
        \includegraphics[width=\linewidth]{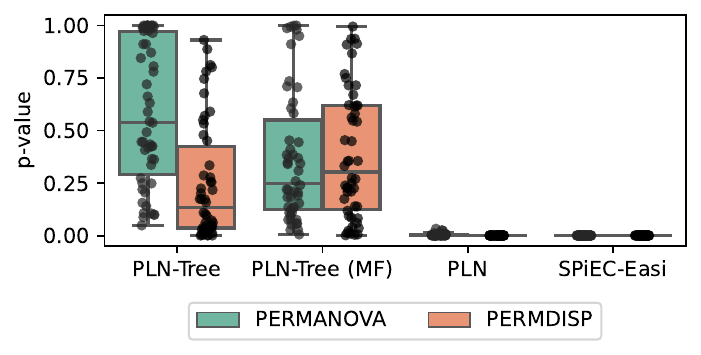} 
        \caption{$\ell = 4$}
    \end{subfigure}
    \hfill
    \begin{subfigure}[b]{0.48\linewidth}
        \centering
        \includegraphics[width=\linewidth]{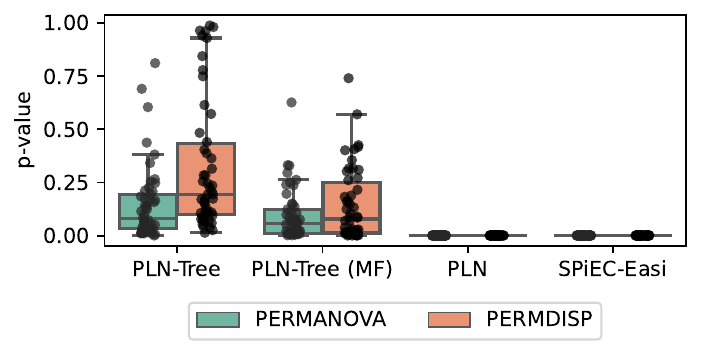} 
        \caption{$\ell = 5$}
    \end{subfigure}
    \caption{$p$-values for PERMANOVA and PERMDISP tests applied on Bray Curtis dissimilarities (layer-wise) computed between $100$ generated data with each model and $100$ sampled microbiome data from the metagenomics dataset, repeated $50$ times.}
    \label{fig:metagenomics_pvalues}
\end{figure*}

These findings suggest that the taxonomy provides pertinent insights into the distribution of bacteria and their interactions within the host's ecosystem, bearing significant biological implications. However, as shown in Appendix \ref{fig:metagenomics_abundances_boxplot}, PLN-Tree approaches struggle with modeling zero-valued abundances (see Bacteria 2, 61, 107 for instance), particularly when using the mean-field approximation. This issue, which accumulates across layers due to the top-down dynamic of the model, could be addressed by using zero-inflation techniques, similar to the approach taken for PLN in \cite{batardiere2024zero}.

\paragraph{Variational approximation performances}
The analysis of the $\alpha$-diversity (see Appendix \ref{tab:metagenomics_benchmark_alpha_diversity}) and the $\beta$-diversity underscores the consistently superior performance of the residual amortized backward approximation compared to the mean-field approach. This observation is further supported by the reconstruction task results summarized in Table \ref{tab:metagenomics_benchmark_encoding}, where structured variational inference exhibits a distinct advantage over the conventional mean-field method in this practical context. Even when the mean-field approximation outperforms the backward approach, as evidenced by the sample distributions in Table \ref{tab:metagenomics_benchmark_samples}, the backward approach remains competitive, indicating its overall effectiveness as the preferred variational approximation method on the metagenomics dataset.
\begin{table}
    \centering
    \begin{tabular}{@{}lcccc@{}}
    \toprule
    & \textbf{PLN-Tree} & \textbf{PLN-Tree (MF)} \\
    \midrule
    $\ell=1$ & \textbf{0.971} (0.113) & 0.850 (0.184) \\
    $\ell=2$ & \textbf{0.971} (0.084) & 0.843 (0.185) \\
    $\ell=3$ & \textbf{0.826} (0.243) & 0.804 (0.258) \\
    $\ell=4$ & \textbf{0.917} (0.165) & 0.736 (0.212) \\
    \bottomrule
    \end{tabular}
\caption{Correlation between reconstructed abundances and the test samples from the metagenomics dataset (see Table \ref{tab:metagenomics_dataset_desc}), averaged over the samples, with standard deviation.}
\label{tab:metagenomics_benchmark_encoding}
\end{table}

\subsubsection{Data preprocessing using PLN-Tree for classification tasks}
\label{sec:data_preprocessing_benchmark}
The metagenomics dataset from \cite{pasolli2016machine} involves a one-vs-all disease classification problem using microbiome proportion data, which are highly sparse and compositional, presenting challenges for direct use in machine learning algorithms \citep{rodriguez2022advances}. To mitigate these constraints, several preprocessing techniques have been proposed, including the additive, centered, and isometric log-ratio transforms, which are commonly used for standard preprocessing \citep{greenacre2021compositional} even though they struggle in highly sparse context and lack theoretical groundings \citep{o2010not}. More recently, \cite{chiquet_plnpca} introduced the use of PLN models to perform PCA in the latent space, demonstrating that latent variables can facilitate machine learning tasks. Therefore, PLN-based approaches can serve as preprocessing pipelines by encoding observations into a latent space, using the identifiable latent variables as input data for machine learning models instead of the raw observations (see Section \ref{sec:identifiability}).

In this context, we explore the use of LP-CLR identifiable features \eqref{eq:latent_tree_counts} derived from our identifiability results in Section~\ref{sec:identifiability} against using regular PLN latent variables or conventional compositional preprocessings such as CLR. Specifically, we benchmark two preprocessing pipelines: \begin{itemize}
    \item The LP-CLR features derived from PLN-Tree restricted at the leaf level of hierarchy, which are computed using the variational encoder learned on the metagenomics dataset using the taxonomy as the underlying hierarchy. This preprocessing is denoted by Latent Tree Proportions - CLR (LTP-CLR).
    \item The LP-CLR features derived from PLN learned at the leaf level of the hierarchy on the metagenomics data, which are computed by projecting the latent variables on $\vect{\mathbf{1}_{K_L}}^\perp$. We denote this preprocessing by Proj-PLN.
\end{itemize}
To illustrate the benefits of these approaches, we focus on the T2D-vs-all classification problem, with the IBD-vs-all scenario detailed in Appendix \ref{app:classif_appendix_preprocessing}. The dataset description is provided in Table \ref{tab:metagenomics_dataset_desc}, and the taxonomic levels considered remain consistent with those in our generative experiments.

\paragraph{PLN limitations and benchmark procedure} We seek to compare the influence of the preprocessing techniques using the conventional PLN latent features, the CLR transform used in SPiEC-Easi, and the LP-CLR features, LTP-CLR and Proj-PLN, against the raw normalized data employed in the study of \cite{pasolli2016machine}. However, the PLN approach proposed in \cite{chiquet_pln} relies on a per-sample variational parameterization, which precludes the development of a generalizable encoder for mapping counts into the latent space. As a result, we train the PLN model on the entire dataset to preprocess the counts—even though, in practical applications, preprocessing models typically involves learning an encoder on a training set and then applying it to unseen data, thereby raising concerns about generalization. In contrast, PLN-Tree employs a flexible neural network parameterization to learn a continuous mapping from counts to latent representations, thus enabling the encoding of unseen data, as explored in Table \ref{tab:metagenomics_benchmark_encoding}. As such, to maintain a fair comparison with PLN, which is trained on the complete dataset, we also train PLN-Tree on all available data; this choice, however, does not demonstrate PLN-Tree's full potential for encoding generalization in real-world applications where Proj-PLN preprocessing is not applicable. 

Upon training both models, we then obtain a latent representation of the counts data by encoding them as averaged samples of each model's variational approximation. Subsequently, the latents obtained from PLN-Tree are mapped using the LP-CLR identifiable transform \eqref{eq:latent_tree_counts} to obtain the LTP-CLR as the leaf level output of LP-CLR, while the LP-CLR transformed PLN features (Proj-PLN) are obtained by projecting the latents on $\vect{\mathbf{1}}^\perp$, yielding two offset-independent microbiome data representations derived from our identifiability results. Then, we select several tabular classifiers and proceed to a $50$ stratified K-Fold training for each model, which allows to account for the training variability on the performances, using $80\%$ of the most precise taxa-abundance data to train the models. To enhance the robustness in each fold, we perform $50$ random grid search iteration to sample the classifiers' hyperparameters, and select the best performing model using $20\%$ of the training set for cross-validation. The hyperparameters' exploration grid is provided in Appendix \ref{app:classif_appendix_preprocessing}.

\paragraph{T2D-vs-all experiment} We consider the classification task of patients with type 2 diabetes against patients with other diseases. In Table \ref{tab:t2d_classif_performances}, we present the performance of various classifiers using the raw proportion data, count data preprocessed with the CLR transform and PLN latents, as well as Proj-PLN and LTP-CLR, employing either the residual backward amortized variational approximation or the mean-field. Overall, our results indicate that both our proposed Proj-PLN and LTP-CLR preprocessings enhance the performances of the classifiers, except for random forests, which, as noted in previous work \citep{yerke2024proportion}, do not benefit well from compositional preprocessing with microbiome data. Nonetheless, the regular CLR and PLN preprocessings are consistently outperformed by the identifiable features derived from the theoretical results across all classifiers, thereby underscoring the practical benefits of identifiability in real-world applications. We also observe that the mean-field approximation is generally outperformed by the proposed residual backward amortized variational approximation, illustrating again the superiority of this approach. The IBD-vs-all experiment conducted in Appendix \ref{app:classif_appendix_preprocessing} highlights similar findings.
\begin{table*}
    \centering
    \resizebox{0.99\textwidth}{!}{
    \begin{tabular}{@{}lc|cc|ccc@{}}
    \toprule
    & \textbf{Proportions} & \textbf{CLR} & \textbf{PLN} & \textbf{Proj-PLN} & \textbf{LTP-CLR (MF)} & \textbf{LTP-CLR} \\
    \midrule
    \multicolumn{7}{c}{\textbf{Logistic Regression}} \\ \midrule
    \textbf{Balanced accuracy} & 0.505 (0.008) & 0.686 (0.035) & 0.665 (0.044) & \textbf{0.689} (0.044) & 0.661 (0.065) & 0.687 (0.046) \\
    \textbf{Precision} & 0.613 (0.119) & 0.76 (0.027) & 0.744 (0.033) & \textbf{0.768} (0.031) & 0.744 (0.036) & 0.767 (0.028) \\
    \textbf{Recall} & 0.726 (0.005) & 0.762 (0.029) & 0.749 (0.034) & \textbf{0.771} (0.027) & 0.741 (0.029) & 0.765 (0.028) \\
    \textbf{F1} & 0.615 (0.01) & 0.756 (0.025) & 0.741 (0.033) & \textbf{0.762} (0.025) & 0.729 (0.033) & 0.756 (0.025) \\
    \textbf{ROC AUC} & 0.646 (0.049) & 0.798 (0.032) & 0.776 (0.041) & \textbf{0.812} (0.036) & 0.762 (0.044) & 0.792 (0.036) \\
    \textbf{PR AUC} & 0.414 (0.064) & 0.596 (0.058) & 0.58 (0.064) & \textbf{0.629} (0.064) & 0.521 (0.058) & 0.581 (0.062) \\
    \midrule
    \multicolumn{7}{c}{\textbf{Linear SVM}} \\ \midrule
    \textbf{Balanced accuracy} & 0.505 (0.02) & 0.688 (0.043) & 0.677 (0.046) & \textbf{0.701} (0.047) & 0.61 (0.075) & 0.674 (0.045) \\
    \textbf{Precision} & 0.554 (0.08) & 0.755 (0.032) & 0.748 (0.034) & 0.764 (0.031) & 0.723 (0.049) & \textbf{0.765} (0.03) \\
    \textbf{Recall} & 0.72 (0.019) & 0.751 (0.032) & 0.74 (0.038) & 0.755 (0.029) & 0.732 (0.023) & \textbf{0.772} (0.03) \\
    \textbf{F1} & 0.611 (0.01) & 0.749 (0.03) & 0.737 (0.033) & 0.755 (0.028) & 0.698 (0.04) & \textbf{0.758} (0.031) \\
    \textbf{ROC AUC} & 0.612 (0.119) & 0.781 (0.032) & 0.766 (0.038) & \textbf{0.796} (0.032) & 0.761 (0.036) & 0.786 (0.037) \\
    \textbf{PR AUC} & 0.409 (0.1) & 0.575 (0.06) & 0.564 (0.064) & \textbf{0.61} (0.063) & 0.52 (0.064) & 0.581 (0.064) \\
    \midrule
    \multicolumn{7}{c}{\textbf{Neural Network}} \\ \midrule
    \textbf{Balanced accuracy} & 0.676 (0.047) & 0.732 (0.04) & 0.709 (0.03) & 0.741 (0.033) & 0.694 (0.042) & \textbf{0.75} (0.042) \\
    \textbf{Precision} & 0.76 (0.037) & 0.797 (0.03) & 0.782 (0.028) & 0.802 (0.027) & 0.775 (0.031) & \textbf{0.808} (0.031) \\
    \textbf{Recall} & 0.77 (0.032) & 0.802 (0.029) & 0.789 (0.027) & 0.805 (0.026) & 0.784 (0.027) & \textbf{0.812} (0.029) \\
    \textbf{F1} & 0.759 (0.034) & 0.797 (0.03) & 0.782 (0.026) & 0.802 (0.026) & 0.773 (0.03) & \textbf{0.808} (0.03) \\
    \textbf{ROC AUC} & 0.787 (0.039) & 0.847 (0.032) & 0.832 (0.033) & \textbf{0.86} (0.028) & 0.783 (0.035) & 0.842 (0.031) \\
    \textbf{PR AUC} & 0.607 (0.064) & 0.691 (0.058) & 0.653 (0.059) & 0.697 (0.052) & 0.615 (0.069) & \textbf{0.704} (0.063) \\
    \midrule
    \multicolumn{7}{c}{\textbf{Random Forest}} \\ \midrule
    \textbf{Balanced accuracy} & \textbf{0.652} (0.072) & 0.619 (0.055) & 0.604 (0.064) & 0.598 (0.051) & 0.637 (0.058) & 0.648 (0.063) \\
    \textbf{Precision} & \textbf{0.809} (0.054) & 0.79 (0.041) & 0.79 (0.055) & 0.772 (0.046) & 0.779 (0.04) & 0.8 (0.037) \\
    \textbf{Recall} & \textbf{0.799} (0.038) & 0.779 (0.027) & 0.774 (0.032) & 0.768 (0.027) & 0.78 (0.029) & 0.792 (0.032) \\
    \textbf{F1} & \textbf{0.758} (0.064) & 0.731 (0.049) & 0.717 (0.061) & 0.713 (0.049) & 0.743 (0.046) & 0.753 (0.055) \\
    \textbf{ROC AUC} & \textbf{0.882} (0.037) & 0.837 (0.044) & 0.858 (0.043) & 0.839 (0.037) & 0.811 (0.042) & 0.844 (0.034) \\
    \textbf{PR AUC} & \textbf{0.751} (0.069) & 0.68 (0.062) & 0.695 (0.084) & 0.666 (0.072) & 0.644 (0.078) & 0.695 (0.067) \\
    \bottomrule
    \end{tabular}
    }
\caption{Classification T2D-vs-all performances for several classifiers on the metagenomics dataset using different preprocessing strategies, averaged over training, with standard deviation. We perform $50$ stratified K-folds using $80\%$ of the dataset with $50$ nested random grid search loops for hyperparameters tuning in each fold using $20\%$ of the training set for cross-validation. The dataset is restricted to the \textit{family} level of the taxonomy.}
\label{tab:t2d_classif_performances}
\end{table*}
Overall, these results demonstrate that identifiable PLN-based features can improve classification performances. In particular, using the latent features rather than the true proportions generally enhances the preprocessing quality of the CLR transform, significantly outperforming the results obtained with the true proportions. 

Further improvements could potentially be attained by using PLN-Tree's identifiable hierarchical LP-CLR features, rather than their leaf level restriction LTP-CLR, by exploring specific recurrent or convolutional deep architectures, but exploring such preprocessing methods is out of the scope of this paper. In particular, hierarchy-aware architectures such as MIOSTONE \citep{jiang2025modeling} could potentially benefit from hierarchical LP-CLR representations instead of using a layer-wise CLR.

%% file: conclusion.tex
In this paper, we introduced the PLN-Tree model as an extension of the Poisson log-normal framework, designed to accommodate hierarchical count data. To learn the parameters of PLN-Tree models, we proposed a structured variational inference approximation to effectively learn the model's parameters by exploiting the true form of the posterior distribution using deep learning parameterizations, showing highly competitive performances against the regular mean-field approximation. Additionally, we established the identifiability properties of PLN-Tree models, providing insights into its theoretical foundations and validating its practical reliability.

To assess the performance of PLN-Tree models, we conducted comprehensive experiments on both synthetic and real-world datasets, benchmarking it against established interaction-based count data models on generative and classification tasks. By using the underlying tree structure, our results underscored the efficacy and consistency of PLN-Tree models in capturing the diversity of the data at all depths, contrary to the regular PLN and SPiEC-Easi approaches. These findings highlight the relevance of hierarchical structures organizing entities, such as the taxonomy, in modeling complex biological systems like the microbiome. Additionally, we demonstrated on a disease classification problem that identifiable features derived from our theoretical results yield performance improvements over traditional compositional preprocessing transforms, such as CLR or PLN. Overall, our contribution offers valuable insights into the practical utility of considering knowledge graphs in modeling approaches, particularly in domains characterized by intricate data structures such as ecology or microbiology, while highlighting the practical benefits of identifiability in generative modeling.

However, PLN-Tree models have certain limitations. While it precisely models proportion-based $\alpha$-diversities, it does not account for sparse structures effectively due to its propagation dynamics. Inspired by the ZI-PLN model \citep{batardiere2024zero}, a zero-inflated PLN-Tree variant could address this limitation and represents a promising direction for future research.
Furthermore, the identification of meaningful interaction networks using PLN-Tree remains an open question. In the general framework, the interaction networks are conditioned on the preceding latent variable, preventing a general interpretation. Although our framework includes a variant where the interaction networks are independent of the preceding latent variable, it would be valuable to investigate this parameterization in specific applications, and explore statistical guarantees analogous to that provided by the faithful correlation property of PLN \citep{chiquet_pln}.

Additionally, while our experiments provide valuable insights into the capabilities of the PLN-Tree framework, further applications to real datasets would be beneficial to illustrate its potential. Specifically, given the significant generative improvements of PLN-Tree over the traditional PLN model in capturing $\alpha$-diversity and $\beta$-diversity characteristics of complex ecosystems, a promising research lead would be to explore its utility in data augmentation tasks. In such instances, the proposed conditional PLN-Tree could be a compelling tool for practitioners aiming to balance datasets or artificially augment small cohorts while investigating the impact of covariates on the quality of the generated data. However, achieving this integration requires substantial deep neural network design to effectively incorporate the covariates. In the same wake, exploring deep recursive or convolutional architectures that could fully leverage the Markov structure of the identifiable latent variables, such as MIOSTONE \citep{jiang2025modeling}, represents another interesting lead for improving the classification performance of PLN-Tree-based preprocessing. Finally, while the PLN-Tree framework offers means to evaluate specific hierarchical clusterings of count data by comparing the generative performances, it requires prior knowledge of the tree structure. In particular in microbiology, although taxonomic clustering is typically available, there is no guarantee that the resulting hierarchy is optimal for a given task—for example, in reconstructing diversity metrics. Therefore, it would be of great practical interest to explore how to infer tree structures from count data. While this extension represents a natural perspective of our work, it poses several theoretical and practical challenges, and remains a widely open topic as illustrated in \citep{momal2020, mao2022dirichlet}.

%% file: acknowledgments.tex
We would like to gratefully thank Harry Sokol, co-director of Alexandre Chaussard's PhD program and medical advisor for this work, as well as the reviewers who thoroughly contributed to improve the clarity of the manuscript. The final publication is available at Springer via \url{https://doi.org/10.1007/s11222-025-10668-w}.

%% file: appendix.tex
\section{Diversity metrics}
\subsection{Alpha diversity}
\label{sec:alpha_diversity}
Alpha diversities are a set of metrics used in ecology and biology to quantify the variety and distribution of species within a particular ecosystem \citep{gotelli2001quantifying,thukral2017review}. These measures consider the diversity within a single sample (a given ecosystem) without considering interactions with other samples. There exist numerous indices to compute $\alpha$-diversity, which evaluate species richness and/or evenness. Species richness refers to the total number of different species present in the sample, while evenness measures how evenly the entities are distributed among the species. High $\alpha$-diversity often indicates a healthy ecosystem with a wide variety of species, while low $\alpha$-diversity suggests a less diverse or possibly disturbed ecosystem.

\paragraph{Shannon entropy}
Originally introduced for information theory, the Shannon entropy is a widely used $\alpha$-diversity metric in ecology to measure species diversity within a given community \citep{thukral2017review}. It evaluates both species richness and evenness by considering the relative abundance of each species. The Shannon entropy calculates the uncertainty or randomness in species composition, reflecting the information content of the community. Higher values of Shannon entropy indicate greater diversity, where species are more evenly distributed, while lower values suggest lower diversity or dominance by a few species. Denoting by $p_s$ the empirical proportion of the species $s$ in the ecosystem, the Shannon entropy is computed as
\[
H = - \sum_{s=1}^S p_s \log p_s \eqsp.
\]
The interpretation of the Shannon entropy as an $\alpha$-diversity is described for instance in \cite{jost2006entropy}.

\paragraph{Simpson index}
The Simpson $\alpha$-diversity metric assesses species diversity within a specific habitat \citep{thukral2017review}. It focuses on the probability that two individuals randomly selected from the community belong to different species. Letting $p_s$ the empirical proportion of species $s$ in the ecosystem, the Simpson index is computed as
\[
S = \sum_{s=1}^S p_s^2 \eqsp.
\]
This metric emphasizes the importance of species evenness in a community, giving more weight to rare species. The interpretation of the Simpson index as an $\alpha$-diversity is given by its reciprocal as the Inverse Simpson index \citep{jost2006entropy}.

\subsection{Beta diversity}
\label{sec:beta_diversity}
Beta diversity measures the variation in species composition between different communities, providing insight into how ecosystems differ from one another, and are thus often referred to as dissimilarity metrics. Unlike $\alpha$-diversity, which quantifies species richness and evenness within a single community (sample), $\beta$-diversity assesses differences in species composition across multiple ecosystems (pairwise dissimilarity). This measure is crucial in ecology studies, where understanding community structure, biogeography, or the effects of environmental changes is essential. Common $\beta$-diversity metrics include Bray-Curtis dissimilarity \cite{beals1984bray}, which evaluates compositional differences based on species abundances, UniFrac (both unweighted and weighted) \cite{lozupone2005unifrac}, which incorporates phylogenetic distances between communities, and the Jaccard index, which compares species presence and absence. These metrics enable biologist to unveil patterns in communities, going further in the environment's characteristics than agglomerated statistics like $\alpha$-diversity.

\paragraph{Bray Curtis dissimilarity}
The Bray-Curtis $\beta$-diversity is used to quantify the compositional dissimilarity between two communities based on species abundances. It ranges from 0 (completely identical) to 1 (completely dissimilar). The metric emphasizes species abundances, making it sensitive to both shared species and their relative quantities, and is widely used in ecological and microbiome studies for comparing community compositions. Given two samples $i,j$, let $C_{ij}$ the amount of entities shared in both samples, $S_i$ the total count in site $i$ and $S_j$ the total count in site $j$, then the Bray Curtis dissimilarity between $i$ and $j$ is given by
\[
    \mathrm{BC}_{ij} = 1 - \frac{2C_{ij}}{S_i + S_j} \eqsp.
\]

\subsubsection{Comparing Beta diversities}
\label{sec:testing_beta_diversity}

Computing the pairwise $\beta$-diversity between multiple ecosystems results in a matrix which captures the dissimilarity in species composition across samples. To quantify and assess the overall similarity between these ecosystems, this matrix can be further used in statistical analyses such as PERMANOVA and PERMDISP, thus providing a statistical framework for comparing ecosystem differences based on $\beta$-diversity metrics.

\paragraph{PERMANOVA} Permutational Multivariate Analysis of Variance (PERMANOVA) \citep{anderson2014permutational} is a non-parametric multivariate statistical test based on permutations. In our context, it is used to compare $\beta$-diversity between two ecosystems by testing the null hypothesis that the centroids and dispersions of these groups are the same, as defined in the measured space given by the dissimilarity matrix. A rejection of the null hypothesis indicates significant differences between groups regarding their centroids, their dispersion, or both.

\paragraph{PERMDISP} Permutational Analysis of Multivariate Dispersions (PERMDISP) \citep{anderson2006distance} is a non-parametric multivariate test that assesses the homogeneity of group dispersions. It tests whether the spread of $\beta$-diversity within each ecosystem differs significantly, regardless of group centroids, according to the dissimilarity matrix provided by the $\beta$-diversity. The test is commonly used in conjunction with PERMANOVA to distinguish whether differences between groups arise from variability in dispersion rather than differences in central tendency. A rejection of the null hypothesis in PERMDISP suggests that the groups exhibit different degrees of variability, making it particularly valuable for interpreting $\beta$-diversity in ecological studies.

\section{ELBO derivation for PLN-Tree}
\label{app:elbo}
Consider a joint probability density $p_{\btheta}(\bZ, \bX) = p_{\btheta,\bZ}(\bZ)p_{\btheta}(\bX \mid \bZ)$. Then, given a variational approximation $q_{\bphi}(. \mid \bX)$, the ELBO is defined as 
$$
\mathcal{L}(\btheta, \bphi) = \EE{q_{\bphi}(. \mid \bX)}{\log p_{\btheta}(\bX \mid \bZ)} - \KL{q_{\bphi}(. \mid \bX)}{p_{\btheta,\bZ}}\eqsp.
$$
The dependency of $\mathcal{L}$ on $\bX$ is dropped for better readability. In addition, when there is no possible confusion, the dependency of the variational distributions on the observations is dropped.
\begin{proposition}
\label{proposition:elbo_pln_tree}
    Consider the PLN-Tree model of Section \ref{sec:plntree}.  Then, when using the backward variational approximation \eqref{eq:top_down_markov_variational_approx}, the ELBO of PLN-Tree models writes 
    \begin{align*}
        \mathcal{L}(\btheta, \bphi) = \sum_{\ell=1}^L & \frac{1}{2}\EE{\varapprox{1:L}}{\log |\bOmega_{\btheta^\ell}(\bZ^{\ell-1})| - \trace(\Sigmahat_{\ell} \bOmega_{\theta^{\ell}}(\bZ^{\ell-1})) + \log |\bS_{\bphi^\ell}(\bZ^{\ell+1}, \bX^{1:\ell})|} \\
        & + \sum_{\kk{} = 1}^{K_{\ell}} \left( \rmX_{\kk{}}^{\ell} \EE{\varapprox{1:L}}{\BF{m}_{\bphi^\ell,\kk{}}(\bZ^{\ell+1}, \bX^{1:\ell})} - \indicator{\ell=1} \EE{\varapprox{1:L}}{\mathrm{M}_{\ell \mid \ell+1}^k(\bZ^{\ell+1})} \right) \\
        & - \indicator{\ell > 1} \sum_{\kk{-1} = 1}^{K_{\ell-1}} \rmX_{\kk{-1}}^{\ell-1} \EE{\varapprox{1:L}}{\log \sum_{\jj{}\in \childindex{\ell-1}{\kk{-1}}} \rme^{Z_{\jj{}}^\ell}}
        - \indicator{\ell=L} \sum_{\kk{}=1}^{K_\ell} \log \rmX_{\kk{}}^\ell ! - \frac{1}{2} K_{\ell} \eqsp,
    \end{align*}
    such that $\bOmega_{\btheta^1}(\bZ^{0}) = \bOmega_1$, $\bmu_{\btheta^1}(\bZ^{0}) = \bmu_1$, $\bS_{\bphi^L}(\bZ^{L+1}, \bX^{1:L}) = \bS_{\bphi^L}(\bX^{1:L})$, $\BF{m}_{\bphi^L} (\bZ^{L+1}, \bX^{1:L}) = \BF{m}_{\bphi^L}(\bX^{1:L})$,
    and for all $1\leq \ell \leq L-1, 1 \leq k \leq K_\ell \eqsp$,
    \begin{equation}
        \begin{aligned}
            \Sigmahat_{\ell} = &\left(\bmu_{\btheta_\ell}(\bZ^{\ell-1}) - \BF{m}_{\bphi^\ell}(\bZ^{\ell+1}, \bX^{1:\ell})\right) \left(\bmu_{\btheta_\ell}(\bZ^{\ell-1}) - \BF{m}_{\bphi^\ell}(\bZ^{\ell+1}, \bX^{1:\ell})\right)^\top + \bS_{\bphi^\ell}(\bZ^{\ell+1}, \bX^{1:\ell}) \eqsp,
        \end{aligned}
    \label{eq:def:sigmahat}
    \end{equation}
    \begin{equation*}
        \mathrm{M}_{\ell \mid \ell+1}^k(\bZ^{\ell+1}) = \exp\left(\frac{\bS_{\bphi^\ell, \kk{}}(\bZ^{\ell+1}, \bX^{1:\ell})}{2} + \BF{m}_{\bphi^\ell, \kk{}}(\bZ^{\ell+1}, \bX^{1:\ell})\right) \eqsp.
    \end{equation*}
\end{proposition}
\begin{proof}
    The prior distribution of $\bZ$ is denoted by $\prior{\bZ}(\bZ) = \prior{1}(\bZ^1) \prod_{\ell=1}^{L-1} \prior{\ell+1 | \ell}(\bZ^{\ell+1} | \bZ^\ell)$.
    By definition of the ELBO,
    \[
    \ELBO(\btheta, \bphi) = \EE{\varapprox{1:L}}{\log \posterior{1:L}(\bX|\bZ)} - \KL{\varapprox{1:L}}{\prior{\bZ}} \eqsp.
    \]
    Using the Markov tree structure of the observed counts \eqref{eq:plntree_joint_density} yields
    \begin{equation}
    \label{eq:decomposition_conditional_likelihood}
        \EE{\varapprox{1:L}}{\log \posterior{1:L}(\bX|\bZ)} = \EE{\varapprox{1:L}}{\log \posterior{1}(\bX^1 | \bZ^1)} + \sum_{\ell=1}^{L-1} \sum_{\kk{}=1}^{K_{\ell}} \EE{\varapprox{1:L}}{\log \posterior{k,\ell}(\child{X}{k}{\ell} | \child{Z}{k}{\ell}, \rmX_k^{\ell})} \eqsp.
    \end{equation}
    Since the first layer is modeled by a Poisson log-normal distribution, the emission distribution writes 
    $$
    \posterior{1}(\bX^1 \mid \bZ^1) = \prod_{k=1}^{K_1} \rme^{-\rme^{\rmZ^1_k}}\frac{\rme^{\rmZ^1_k \rmX_k^1}}{\rmX_k^1 !} \eqsp.
    $$
    Using the tower property, conditioning on $\bZ^{\ell+1}$ yields $\mathds{E}_{q_{\bphi,1:L}}[{\rmZ_j^{\ell}}] = \EE{q_{\bphi,1:L}}{\BF{m}_{\bphi^\ell,j}(\bZ^{\ell+1}, \bX^{1:\ell})}\eqsp$, the first term can then be expressed as
    \begin{multline*}
        \EE{\varapprox{1:L}}{\log \posterior{1}(\bX^1 | \bZ^1)} = \sum_{\pkk{1}=1}^{K_1} \rmX_{\pkk{1}}^1 \EE{\varapprox{1:L}}{\BF{m}_{\bphi^1, \pkk{1}}(\bZ^2, \bX^{1})} \\ - \EE{\varapprox{1:L}}{\exp\left(\frac{\bS_{\bphi^1, \pkk{1}}(\bZ^2, \bX^{1})}{2} + \BF{m}_{\bphi^\ell, \pkk{1}}(\bZ^2, \bX^{1})\right)} - \log (\rmX_{\pkk{1}}^1 !) \eqsp.
    \end{multline*}
    The propagation of the counts along the tree conditionally to the respective latent variables and the parent counts is given by a multinomial distribution (see Section \ref{sec:plntree}) such that
    $$
    \posterior{k,\ell}(\child{\bX}{k}{\ell} \mid \child{\bZ}{k}{\ell}, \rmX_k^\ell) = \frac{\rmX_k^\ell !}{\prod_{j \in \childindex{\ell}{k}} \rmX_j^{\ell+1}!} \prod_{j \in \childindex{\ell}{k}}\left( \frac{\rme^{\rmZ_j^{\ell+1}}}{\sum_{v\in \childindex{\ell}{k}} \rme^{\rmZ_v^{\ell+1}}} \right)^{\rmX_j^{\ell+1}} = \frac{\rmX_k^\ell !}{\prod_{j \in \childindex{\ell}{k}} \rmX_j^{\ell+1}!} \frac{\rme^{\sum_{j \in \childindex{\ell}{k}} \rmZ_j^{\ell+1} \rmX_j^{\ell+1}}}{\left(\sum_{j\in \childindex{\ell}{k}} \rme^{\rmZ_j^{\ell+1}} \right)^{\rmX_k^{\ell}}} \eqsp.
    $$
    Noticing that $\sum_{k=1}^{K_{\ell}} \sum_{j \in \childindex{\ell}{k}} \log (\rmX_j^{\ell+1} !) = \sum_{k=1}^{K_{\ell+1}} \log (\rmX_k^{\ell+1} !)$, the second term of \eqref{eq:decomposition_conditional_likelihood} can be explicated as
    \begin{multline*}
        \sum_{\ell=1}^{L-1} \sum_{\kk{}=1}^{K_{\ell}} \EE{\varapprox{1:L}}{\log \posterior{k,\ell}(\child{X}{k}{\ell} | \child{Z}{k}{\ell}, \rmX_k^{\ell})} = \sum_{\pkk{1}=1}^{K_1} \log (\rmX_{\pkk{1}}^1!) - \sum_{\pkk{L}=1}^{K_{L}} \log (\rmX_{\pkk{L}}^L!)  \\+ \sum_{\ell=1}^{L-1} \sum_{\kk{}=1}^{K_{\ell}} \left\{\sum_{\jj{+1} \in \childindex{\ell}{\kk{}}} \rmX_{\jj{+1}}^{\ell+1} \EE{\varapprox{1:L}}{\rmZ_{\jj{+1}}^{\ell+1}}  
        - \rmX_{\kk{}}^\ell \EE{\varapprox{1:L}}{\log \left(\sum_{\jj{+1}\in \childindex{\ell}{\kk{}}} \rme^{\rmZ_{\jj{+1}}^{\ell+1}}\right)}\right\} \eqsp.
    \end{multline*}
Combining the previous results provides the expected conditional log-likelihood as
    \begin{align*} 
        \EE{\varapprox{1:L}}{\log \posterior{1:L}(\bX|\bZ)} = \sum_{\ell=1}^L
        & \sum_{\kk{} = 1}^{K_{\ell}} \bigg( \rmX_{\kk{}}^{\ell} \left(\indicator{\ell < L}\EE{\varapprox{1:L}}{\BF{m}_{\bphi^\ell,\kk{}}(\bZ^{\ell+1}, \bX^{1:\ell})} + \indicator{\ell = L} \BF{m}_{\bphi^L,k}(\bX^{1:L})\right) \\ 
        &\quad  - \indicator{\ell=1} \EE{\varapprox{1:L}}{\exp\left(\frac{\bS_{\bphi^\ell, \kk{}}(\bZ^{\ell+1}, \bX^{1:\ell})}{2} + \BF{m}_{\bphi^\ell, \kk{}}(\bZ^{\ell+1}, \bX^{1:\ell})\right)} \bigg) \\
        &- \indicator{\ell > 1} \sum_{\kk{-1} = 1}^{K_{\ell-1}} \rmX_{\kk{-1}}^{\ell-1} \EE{\varapprox{1:L}}{\log \left(\sum_{\jj{}\in \childindex{\ell-1}{\kk{-1}}} \rme^{\rmZ_{\jj{}}^\ell}\right)} - \indicator{\ell=L} \sum_{\kk{}=1}^{K_\ell} \log (\rmX_{\kk{}}^\ell !) \eqsp.
    \end{align*}
    The divergence term can be expressed as
    \begin{align*}
        \KL{\varapprox{1:L}}{\prior{\bZ}} &
        \\&\hspace{-2cm}=\EE{\varapprox{1:L}}{\log \left( \frac{\varapprox{1\mid 2}(\bZ^1 \mid \bZ^2, \bX^{1:2})}{\prior{1}(\bZ^1)} \prod_{\ell=2}^{L-1} \frac{\varapprox{\ell \mid \ell+1}(\bZ^\ell \mid \bZ^{\ell+1}, \bX^{1:\ell})}{\prior{\ell \mid \ell-1}(\bZ^\ell \mid \bZ^{\ell-1})} \frac{\varapprox{L}(\bZ^L \mid \bX^{1:L})}{\prior{L \mid L-1}(\bZ^L \mid \bZ^{L-1})} \right)} \\
        &\hspace{-2cm}= \EE{\varapprox{1:L}}{\KL{\varapprox{1 \mid 2}}{\prior{1}}} + \sum_{\ell=2}^{L-1} \EE{\varapprox{1:L}}{\KL{\varapprox{\ell \mid \ell+1}}{\prior{\ell \mid \ell-1}}} \\
        &\hspace{6.5cm}+\EE{\varapprox{1:L}}{\KL{\varapprox{L}}{\prior{L \mid L-1}}}.
    \end{align*}
    For $1 < \ell < L$, the Kullback-Leibler divergence writes
    \begin{multline*}
        \KL{\varapprox{\ell \mid \ell+1}}{\prior{\ell \mid \ell - 1}} = - \frac12 \left[ \log|\bOmega_{\theta^{\ell}}(\bZ^{\ell-1})| + \log|\bS_{\bphi^\ell}(\bZ^{\ell+1}, \bX^{1:\ell})| + K_\ell\right]\\
        + \frac{1}{2}\trace\bigg(\Sigmahat_{\ell}\bOmega_{\btheta^\ell}(\bZ^{\ell-1})\bigg)\eqsp,
    \end{multline*}
    where $\Sigmahat_{\ell}$ is defined in \eqref{eq:def:sigmahat}.
    Following the same steps for the other terms yields
    \begin{multline*}
    \KL{\varapprox{1:L}}{\prior{\bZ}} \\ = -\frac{1}{2} \sum_{\ell=1}^{L} \EE{\varapprox{1:L}}{\log |\bOmega_{\btheta^\ell}(\bZ^{\ell-1})| + \log |\bS_{\bphi^\ell}(\bZ^{\ell+1}, \bX^{1:\ell})| - \trace\left(\Sigmahat_{\ell} \bOmega_{\btheta^\ell}(\bZ^{\ell-1})\right)} + K_\ell \eqsp,
    \end{multline*}
    which concludes the proof.
\end{proof}

\subsection{PLN-tree ELBO with offset modeling}
\label{sec:plntree_elbo_offset}
For a sample $i$, let $\offset_i \in \RR$ its offset, following H\ref{hypH:plntree_offset}:
\begin{hypH}
\label{hypH:plntree_offset}
\begin{itemize}
    \item The $(\offset_i, \bZ_i, \bX_i)_{1 \leq i \leq n}$ are i.i.d., and for $\ell \leq L-1$, conditionally on $\offset, \bZ, (\bX^v)_{1 \leq v \leq \ell}$, the random variables $(\child{X}{k}{\ell})_{1 \leq k \leq K_\ell}$ are independent and the conditional law of $\child{X}{k}{\ell}$ depends only on $\child{Z}{k}{\ell}$ and and $\rmX_{\kk{}}^\ell$.
    \item The distribution of the offset $\offset$ is a Gaussian mixture, and conditionally on $\bX$, the offset $\offset$ and the latent variables $\bZ$ are independent.
    \item The latent process $(\bZ^\ell)_{1\leq \ell \leq L}$ is a Markov chain with initial distribution $\bZ^1 \sim \gaussian{\bmu_1}{\bSigma_1}$ and such that for all $1\leq \ell \leq L-1$, the conditional distribution of $\bZ^{\ell + 1}$ given $\bZ^{\ell}$ is Gaussian with mean $\bmu_{\btheta_{\ell+1}}(\bZ^{\ell})$ and variance $\bSigma_{\btheta_{\ell+1}}(\bZ^{\ell})$.
    \item Conditionally on $\bZ^1$, $\bX^1$ has a Poisson distribution with parameter $\exp({\bZ^1 + \offset})$ and for all $1\leq \ell \leq L-1$, $1\leq k \leq K_\ell$, conditionally on $\rmX_{\kk{}}^\ell$ and $\child{Z}{k}{\ell}$, $\child{X}{k}{\ell}$ has a multinomial distribution with parameters $\softmax{\child{Z}{k}{\ell}}$ and $\rmX_{k}^\ell$.
\end{itemize}    
\end{hypH}
We define the following variational approximation to compute the unknown posterior:
\begin{hypH}
\label{hypH:offset_variational}
    \begin{itemize}
        \item Inheriting the property of the true posterior, under the variational approximation, $\offset$ and $\bZ$ are independent conditionally to $\bX$.
        \item The variational approximation $\varapproxoffset(\offset | \bX)$ is a Gaussian with mean $\rmm_o(\bX)$ and variance $\rms_o^2(\bX)$.
        \item The latent posterior $\varapproxZ{1:L}(\bZ | \bX)$ is a backward Markov chain as defined in \eqref{eq:top_down_markov_variational_approx}.
    \end{itemize}
\end{hypH}

\begin{proposition}
\label{prop:elbo:offset}
    Assume that H\ref{hypH:plntree_offset} and H\ref{hypH:offset_variational} hold. Denote by $\ELBO_{\mid \offset}(\btheta, \bphi)$ the ELBO of the generative model from Proposition \ref{proposition:elbo_pln_tree} with shifted latent means $\bmu_1 + \offset$ and $\BF{m}_{\bphi^1}(.) + \offset$, then the ELBO of the offset-modeled PLN-Tree is given by
    \begin{align*}
        \ELBO_{\mathrm{offset}}(\btheta, \bphi) &= \ELBO_{\mid \offset}(\btheta, \bphi) + 2 \EE{\varapproxoffset}{\log p_{\btheta}(\offset)} + \frac{1}{2} \log s^2_o(\bX) + \frac{1 + \log 2\pi}{2} \eqsp.
    \end{align*}
\end{proposition}
\begin{proof}
\label{proof:elbo_offset}
    By definition,
    \[
    \ELBO_{\mathrm{offset}}(\btheta, \bphi) = \EE{\varapprox{1:L}}{\log p_{\btheta}(\bX, \bZ, \offset)} - \KL{\varapprox{1:L}}{\distrib{(\offset, \bZ)}{}} \eqsp.
    \]
    Conditioning $(\bX, \bZ)$ by $\offset$ yields
    \[
    \ELBO_{\mathrm{offset}}(\btheta, \bphi) = \ELBO_{\mid \offset}(\btheta, \bphi) + \EE{\varapproxoffset}{\log p_{\btheta}(\offset)} + \KL{\varapproxoffset}{\prioroffset} \eqsp.
    \]
    Using the KL divergence definition
    \[
    \KL{\varapproxoffset}{\prioroffset} = -\entropy_{\varapproxoffset} - \EE{\varapproxoffset}{\log p_{\btheta}(\offset)} \eqsp,
    \]
    since $\varapproxoffset$ is Gaussian, its entropy is given by $\frac{1}{2}\log(2\pi\rme s_o^2(\bX))$, which concludes the proof.
\end{proof}

\section{Identifiability results}

\subsection{PLN identifiability}
\begin{lemma}
\label{lemma:poisson_composed_identifiability}
    Let $\rmZ$ and $\tilde\rmZ$ be supported on $\RR^*_+$, and $\rmX \sim \poisson{\rmZ}$ and $\tilde\rmX \sim \poisson{\rmZ}$. Then, if $\rmX$ and $\tilde\rmX$ have the same distribution, $\rmZ$ and $\tilde\rmZ$ have the same distribution.
\end{lemma} 
\begin{proof}
    Let $h$ be a measurable function, then we have
    \begin{equation*}
        \EE{}{h(\rmX)} = \EE{}{\EE{}{h(\rmX) \mid \rmZ}} = \EE{}{\sum_{x \in \NN} \rme^{-\rmZ} \frac{\rmZ^x}{x!}h(x)} \eqsp.
    \end{equation*}
    For all $t\in\RR$, taking $h(x) = t^x$ yields
    \begin{equation*}
        \EE{}{h(\rmX)} = \EE{}{\rme^{-\rmZ} \sum_{x \in \NN} \frac{(\rmZ t)^x}{x!}} = \EE{}{\rme^{(t-1)\rmZ}} = \mgf_{\rmZ}(t-1)\eqsp.
    \end{equation*}
    Since $\rmX$ and $\tilde\rmX$ have the same law, then we have for all $u \leq 0, \mgf_\rmZ(u) = \mgf_{\tilde\rmZ}(u)$. Write $\rmY = \exp(-\rmZ)$ and $\tilde \rmY = \exp(-\tilde \rmZ)$. The random variables $\tilde \rmY$ and $\rmY$ are compactly supported so by the Stone-Weierstrass theorem their distribution is characterized by their moments $(\EE{}{\rmY^k})_{k \geq 0}$ and $(\mathbb{E}[\tilde \rmY^k])_{k \geq 0}$. Therefore  $\tilde \rmY$ and $\rmY$ have the same law, which concludes the proof.
\end{proof}

\begin{lemma}
\label{lemma:pln_identifiability}
    Let $\rmZ$ and $\tilde\rmZ$ be two real random variables, and $\rmX \sim \poisson{\rme^{\rmZ}}$ and $\tilde\rmX \sim \poisson{\rme^{\tilde\rmZ}}$. Then, if $\rmX$ and $\tilde\rmX$ have the same distribution, $\rmZ$ and $\tilde\rmZ$ have the same distribution.
\end{lemma}
\begin{proof}
By Lemma~\ref{lemma:poisson_composed_identifiability}, $\rme^{\rmZ}$ and $\rme^{\tilde\rmZ}$ have the same distribution which is enough to conclude the proof.
\end{proof}

\subsubsection{Proof of Lemma \ref{lemma:poisson_markov_identifiability}}
\label{app:poisson_markov_identifiability}
    Let $h(\rmX_1, \dots, \rmX_K) = \prod_{k=1}^K h_k(\rmX_k)$ where $\{h_k\}_{1\leq k\leq K}$ are measurable functions. Then,
    \begin{align*}
        \EE{}{h(\rmX_1, \dots, \rmX_K)} = \EE{}{\EE{}{h(\rmX_k, \dots, \rmX_K) \mid \bZ}} 
        &= \EE{}{\prod_{k=1}^K \EE{}{h_{k}(\rmX_k) \mid \rmZ_k}} \\
        &= \EE{}{\prod_{k=1}^K \sum_{x \in \NN} \rme^{-\rmZ_k} \frac{({\rmZ_k})^x}{x!} h_k(x)} \eqsp.
    \end{align*}
    Choosing $h_k(x) =  t_k^{x}$, yields
    \begin{align*}
        \EE{}{h(\rmX_1, \dots, \rmX_K)} &= \EE{}{\prod_{k=1}^K \rme^{(t_k - 1)\rmZ_k}} \eqsp.
    \end{align*}
    By setting $\BF{u} = \{t_k - 1\}_{1 \leq k \leq K}$, we obtain
    \begin{align*}
        \EE{}{h(\rmX_1, \dots, \rmX_K)} = \EE{}{\rme^{\BF{u}^\top \bZ}} = \mgf_{\bZ}(\BF{u}) \eqsp.
    \end{align*}
    The proof is concluded by the same arguments as in Lemma \ref{lemma:poisson_composed_identifiability}.

\subsection{PLN-Tree identifiability}

\subsubsection{Identifiability of parent-children distributions at the first layer}
\begin{lemma}
\label{lemma:identifiability_anylaw_simplex}
    Let $\bZ = (\rmZ^1, \bZ^2)$ be random variables such that $\rmZ^1> 0$, $\bZ^2\in \simplex{K}$, where $\simplex{K}$ denotes the simplex in $\RR^K$. Assume that the observations $\bX = (\rmX^1, \bX^2)$  are such that conditionally on $\rmZ^1$, $\rmX^1 \sim \poisson{\rmZ^1}$ and conditionally on $(\rmX^1,\bZ^2)$, $\bX^2 \sim \multinomial{\rmX^1}{\bZ^2}$.  Then, the law of  $(\rmZ^1, \bZ^2)$ is identifiable from the law of $(\rmX^1,\bX^2)$.
\end{lemma}
\begin{proof}
\label{app:identifiability_anylaw_simplex}
    Let $h$ be a measurable function. For all $x_1\geq 1$, let $x^1 \simplex{K} = \{(x_1^2, \dots, x_K^2) \in \RR^K \mid \sum_{k=1}^K x_k^2 = x^1\}$, then
    \begin{align*}
        \EE{}{h(\rmX^1, \bX^2)} &= \EE{}{\EE{}{h(\rmX^1, \bX^2) \mid \bZ}} \\
        &= \EE{}{\sum_{x^1 \in \NN} \sum_{\BF{x}^2 \in x^1 \simplex{K}} \rme^{-\rmZ^1} \prod_{k=1}^K \frac{(\rmZ^1 \rmZ^2_k)^{x_k^2}}{x_k^2!} h(x^1, \BF{x}^2)} \eqsp.
    \end{align*}
    Using that $\bZ^2$ lies in the simplex yields
    \begin{align*}
        \EE{}{h(\rmX^1, \bX^2)} &= \EE{}{\sum_{\BF{x}^2 \in \NN^K} \prod_{k=1}^K \rme^{-\rmZ^1 \rmZ^2_k} \frac{(\rmZ^1 \rmZ^2_k)^{x_k^2}}{x_k^2!} h\left(\sum_k^{K} x_k^2, \BF{x}^2\right)}\eqsp.
    \end{align*}
    Therefore, $(\bX^2_1,\ldots,\bX^2_K)$ are conditionally independent with  Poisson distribution with parameters  $(\rmZ^1 \rmZ^2_k)_{1 \leq k \leq K}$. Hence, by Lemma~\ref{lemma:poisson_markov_identifiability}, the law of $\bU = (\rmZ^1 \rmZ^2_1, \dots, \rmZ^1 \rmZ^2_K)$ is identifiable. Since $\bZ^2$ lies in the simplex, conditionally on  $ \bU$, $\rmZ^1$ has a Dirac distribution with mass at $\sum_{k=1}^K \bU_k$.
    For any measurable function $f$,
    \begin{align*}
        \EE{}{f(\rmZ^1, \bZ^2)} = \EE{}{\EE{}{f(\rmZ^1, \bZ^2) \mid \bU}}
        &= \EE{}{\EE{}{f\left(\rmZ^1, \frac{\bU}{\rmZ^1}\right) \mid \bU}}\\
        &= \EE{}{f\left(\sum_{k=1}^K \bU_k, \frac{\bU}{\sum_{k=1}^K \bU_k}\right) }\eqsp.
    \end{align*}
    Then, as the law of $\bU$ is identifiable, the law of $(\rmZ^1, \bZ^2)$ is identifiable.
\end{proof}
\subsubsection{Identifiability through softmax transform}
\label{proof:identifiability_projected_softmax}

\begin{lemma}
\label{lemma:identifiability_projected_softmax}
    Let $\bZ$, $\tilde\bZ$ be two random variables in $\RR^d$. Define $\bP = \identity_d -  \BS{1}_{d\times d}/d$ the projector on $\vect{\BS{1}_d}^\perp$. Then, if $\softmax{\bZ}$ and $\softmax{\tilde\bZ}$ have the same distribution, $\bP \bZ$ and $\bP \tilde\bZ$ have the same distribution and conversely.
\end{lemma}
\begin{proof}
We start with the direct sense of the equivalence. Let $B \in \mathcal{B}(\simplex{d})$, since $\softmax{\cdot}$ is surjective on $\simplex{d}$ there exists $C \subseteq \RR^d$ such that $\softmax{C} = B$. Then, assuming $\softmax{\bZ}$ has the same law as $\softmax{\tilde\bZ}$,
    \begin{equation*}
        \PP(\softmax{\bZ} \in B) = \PP(\softmax{\tilde\bZ} \in B)\eqsp,
    \end{equation*}
    so that
    \begin{equation*}
        \PP(\softmax{\bZ} \in \softmax{C}) = \PP(\softmax{\tilde\bZ} \in \softmax{C}) \eqsp.
    \end{equation*}
    On the event $\{\softmax{\bZ} \in \softmax{C}\}$, there exists $\BS{c} \in C$ such that $\softmax{\bZ} = \softmax{\BS{c}}$, which yields
    \[
    \bZ = \BS{c} + K(\BS{c}, \bZ) \BS{1}_d \eqsp, 
    \]
    with $K(\BS{c}, \bZ) = \log(\sum_{k=1}^d \rme^{\bZ_{k}} \big/ \sum_{k=1}^d \rme^{\BS{c}_{k}} )$.
    Since $\bP$ is the projector on $\vect{\BS{1}_d}^\perp$, we have $\bP \BS{1}_d = 0$, which yields $\bP \bZ = \bP \BS{c} \in \bP C$, the projection of $C$ on $\vect{\BS{1}_d}^\perp$ and therefore $\{\softmax{\bZ} \in \softmax{C}\}\subseteq \{\bP \bZ \in \bP C\}$. Additionally, since softmax is invariant by constant translation, we have $ \{\bP \bZ \in \bP C\}\subseteq \{\softmax{\bZ} \in \softmax{C}\}$, and thus 
    \[
    \PP(\bP \bZ \in \bP C) = \PP(\softmax{\bZ} \in \softmax{C}) = \PP(\softmax{\tilde\bZ} \in \softmax{C}) = \PP(\bP \tilde\bZ \in \bP C) \eqsp, 
    \]
    which concludes the direct sense of the equivalence. The converse statement is obtained similarly.
\end{proof}

\subsubsection{Proof of Corollary \ref{corollary:identifiability_plntree_projected}}
\label{proof:identifiability_plntree_projected}
Since $(\bZ^1, \softmax{\bZ^2})$ and $(\tilde\bZ^1, \softmax{\tilde\bZ^2})$ have the same distribution, Lemma \ref{lemma:identifiability_projected_softmax} yields that $(\bZ^1, \bP{\bZ^2})$ and $(\tilde\bZ^1, \bP{\tilde\bZ^2})$ have the same distribution. Furthermore, since conditionally to $\bZ^1$ (resp. $\tilde \bZ^1$), $\bZ^2$ (resp. $\tilde \bZ^2$) is Gaussian, observing that $\bP = \bP^\top$, the law of $\bP \bZ^2$ (resp. $\bP\tilde\bZ^2$) is given by $\mathcal{N}(\bP\bmu(\bZ^1),\bP\bSigma(\bZ^1)\bP)$ (resp. $\mathcal{N}(\bP\tilde\bmu(\tilde\bZ^1),\bP\tilde\bSigma(\tilde\bZ^1)\bP)$), which concludes the proof.

\subsubsection{Proof of Theorem \ref{theorem:plntree_identifiability}}
\label{proof:plntree_identifiability}
For $\ell \in \{1, 2\}$, we denote the events $\{\rmX_k^\ell = \sum_{j \in \childindex{\ell}{k}} \rmX_j^{\ell+1}\}_{k \leq K_{\ell}}$ by $\{\bX^\ell = \parent{\bX}{}{\ell+1}\}$ and $\NN^{\BF{K}} = \NN^{1} \times \NN^{K_2} \times \NN^{K_3}$. Then, for all measurable functions $h$ we have
\begin{align*}
    \EE{}{h(\rmX^1, \bX^2, \bX^3)} &= \EE{}{\EE{}{h(\rmX^1, \bX^2, \bX^3) \mid \rmZ^1, \bZ^2, \bZ^3}} \\
    &= \EE{}{\sum_{\BF{x} \in \NN^{\BF{K}}} h(\mathrm{x}^1,\BF{x}^2,\BF{x}^3) \;.\; \rme^{-\rmZ^1}\frac{(\rmZ^1)^{\mathrm{x}^1}}{\mathrm{x}^1!} \;.\; \indicator{\BF{x}^1 = \parent{\BF{x}}{}{2}, \BF{x}^2 = \parent{\BF{x}}{}{3}} \mathrm{x}^1 ! \prod_{k=1}^{K_2} \frac{{(\rmZ_k^2)}^{\mathrm{x}_k^2}}{\mathrm{x}_k^2!}  \;.\;  \prod_{k=1}^{K_2} \mathrm{x}_k^2! \prod_{j \in \childindex{2}{k}} \frac{{(\rmZ_j^3)}^{\mathrm{x}_j^3}}{\mathrm{x}^3_j!}} \\
    &= \EE{}{\sum_{\BF{x} \in \NN^{\BF{K}}} h(\mathrm{x}^1,\BF{x}^2,\BF{x}^3) \;.\; \rme^{-\rmZ^1}\frac{(\rmZ^1)^{\mathrm{x}^1}}{\mathrm{x}^1!} \;.\; \indicator{\BF{x}^1 = \parent{\BF{x}}{}{2}, \BF{x}^2 = \parent{\BF{x}}{}{3}} \mathrm{x}^1 ! \prod_{k=1}^{K_2} {(\rmZ_k^2)}^{\mathrm{x}_k^2} \prod_{j \in \childindex{2}{k}} \frac{{(\rmZ_j^3)}^{\mathrm{x}_j^3}}{\mathrm{x}^3_j!}} \\
    &= \EE{}{\sum_{\BF{x} \in \NN^{\BF{K}}} h(\mathrm{x}^1,\BF{x}^2,\BF{x}^3) \;.\; \rme^{-\rmZ^1}\frac{(\rmZ^1)^{\mathrm{x}^1}}{\mathrm{x}^1!} \;.\; \indicator{\BF{x}^1 = \parent{\BF{x}}{}{2}, \BF{x}^2 = \parent{\BF{x}}{}{3}} \mathrm{x}^1 ! \prod_{k=1}^{K_2} \prod_{j \in \childindex{2}{k}} \frac{{(\rmZ_k^2 \rmZ_j^3)}^{\mathrm{x}_j^3}}{\mathrm{x}^3_j!}} \eqsp.
\end{align*}
Using that $\{\childindex{2}{k}\}_{k \leq K_2}$ is a partition of $\{1, \dots, K_3\}$ and denoting by $\parent{\rmZ}{k}{\ell}$ the parent of $\rmZ_k^\ell$ yields
\begin{align*}
    \EE{}{h(\rmX^1, \bX^2, \bX^3)} &= \EE{}{\sum_{\BF{x} \in \RR^{\BF{K}}} h(\mathrm{x}^1,\BF{x}^2,\BF{x}^3) \;.\; \rme^{-\rmZ^1}\frac{(\rmZ^1)^{\mathrm{x}^1}}{\mathrm{x}^1!} \;.\; \indicator{\BF{x}^1 = \parent{\BF{x}}{}{2}, \BF{x}^2 = \parent{\BF{x}}{}{3}} \mathrm{x}^1 ! \prod_{k=1}^{K_3} \frac{{(\parent{\rmZ}{k}{3} \rmZ_k^3)}^{\mathrm{x}_k^3}}{\mathrm{x}^3_k!}} \eqsp.
\end{align*}
Remarking that $\sum_{k=1}^{K_3} \parent{\rmZ}{k}{3} \rmZ_k^3 = \sum_{k=1}^{K_2} \rmZ_k^2 \sum_{j \in \childindex{2}{k}} \rmZ_j^3$, since $\bZ^2 \in \simplex{K_2}$ and for all $k \leq K_2, \child{\bZ}{k}{2} \in \simplex{\# \childindex{2}{k}}$, then $\sum_{k=1}^{K_3} \parent{\rmZ}{k}{3} \rmZ_k^3 = 1$.
Consequently, conditionally on $\rmZ^1$, $\rmX^1\sim \poisson{\rmZ^1}$, and conditionally on $(\bX^1, (\parent{\rmZ}{k}{3} \rmZ_k^3)_{k \leq K_3})$, $\bX^3$ has a multinomial distribution with total count $\bX^1$ and probabilities $(\parent{\rmZ}{k}{3} \rmZ_k^3)_{k \leq K_3}$. Then applying Lemma \ref{lemma:identifiability_anylaw_simplex} provides the identifiability of the law of $(\rmZ^1, (\parent{\rmZ}{k}{3} \rmZ_k^3)_{k \leq K_3}) = (\rmZ^1, (\rmZ_k^2 \child{\rmZ}{k}{2})_{k \leq K_2})$. Writing for all $k \leq K_2, \bU_k = \rmZ_k^2 \child{\bZ}{k}{2}$ and $\bU = (\bU_1, \dots, \bU_k)_{k \leq K_2}$, and following the last steps of Lemma \ref{lemma:identifiability_anylaw_simplex} yields for all measurable functions $f$,
\begin{align*}
    \EE{}{f(\rmZ^1, \bZ^2, \bZ^3)} = \EE{}{f\left(\rmZ^1, \left(\sum_{j=1}^{\#\childindex{2}{k}} \bU_{kj}\right)_{k \leq K_2}, \left(\frac{\bU_k}{\sum_{j=1}^{\#\childindex{2}{k}} \bU_{kj}}\right)_{k \leq K_2}\right)} \eqsp.
\end{align*}
Then, as the law of $(\rmZ^1, \bU)$ is identifiable, the law of $(\rmZ^1, \bZ^2, \bZ^3)$ is identifiable.

\section{Experimental setup}
\label{app:experiments_setup}
\paragraph{Latent prior architectures} The latent prior is a Markov chain with Gaussian transition kernels parameterized by neural networks, such that at layer $1 < \ell \leq L$, the mean $\bmu_{\btheta^\ell}(.) \in \RR^{K_{\ell}}$ and precision matrix $\bOmega_{\btheta^\ell}(.) \in \RR^{K_{\ell} \times K_\ell}$ use $\bZ^{\ell-1} \in \RR^{K_{\ell-1}}$ as input. In our experiments, the mean and precision of the latent dynamic are both composed of two modules. The first module consists of a fully connected neural network, such that at layer $1 < \ell \leq L$ of the tree, we fix the number of neurons in the hidden layers to $K_{\ell-1}$ for the mean, and $K_{\ell-1} (K_{\ell} + 1) / 2$ for the precision, and only tune the number of hidden layers. Then, for the mean, we add a module to compute the projector associated with the layer $\ell$ of the tree to ensure the identifiability of the mean parameter (see Section \ref{sec:identifiability}). Similarly, for the precision matrix, we attach a module that turns the output of the first module into a lower triangular matrix $\cholesky_{\btheta^\ell}(.)$ with positive diagonal terms using softplus, thus obtaining the Cholesky decomposition of a positive definite matrix. To prevent computational issues, we add a perturbation term of amplitude $\lambda = 10^{-4}$, ensuring the numerical invertibility of the covariance matrix which is given by $\bSigma_{\btheta^\ell}(.) = \cholesky_{\btheta^\ell}(.) \cholesky_{\btheta^\ell}(.)^\top + \lambda \identity_{K_\ell}$. We then proceed to taking its inverse to obtain the precision matrix.

Finally, we initialize the parameters of the first layers based on PLN initialization such that for all $k \leq K_1$,
$$\bmu_{1,k} = \frac{1}{n} \sum_{i=1}^n \log \rmX_{ik}^1
$$ 
and 
$$
\bSigma_1 = \frac{1}{n-1} (\log \bX^1 - \BS{1}_{n \times K_1} \bmu_1)^\top (\log \bX^1 - \BS{1}_{n \times K_1} \bmu_1)\eqsp,
$$
the other parameters are initialized at random.

\paragraph{Mean-field architectures} In the mean-field approximation, the parametrization of the Gaussian kernels at layer $\ell \leq L$ is made of two neural networks with inputs $\bX$. In our experiments, the input of the networks at layer $\ell$ is limited to $\bX^\ell$ (see \cite{blei2017variational}). At layer $\ell$, the mean $\BS{m}_{\bphi^\ell}(\bX^\ell)$ and the diagonal covariance matrix $\bS_{\bphi^\ell}(\bX^\ell)$ have the same network architecture but consists of two different fully connected neural networks with output of dimension $K_\ell$. In our experiments, the architecture of the networks is solely parameterized by the number of hidden layers, while the number of neurons at each hidden layer is fixed to $K_\ell$ at depth $\ell$ of the tree.

\paragraph{Backward Markov architectures} The backward variational approximation is a backward Markov chain with Gaussian transition kernels, such that at layer $L$ the mean and diagonal covariance matrix use $\bX^{1:L}$ as inputs, and for layer $\ell < L$, the mean and diagonal covariance matrix use $(\bX^{1:\ell}, \bZ^{\ell+1})$ as inputs (see \eqref{eq:top_down_markov_variational_approx}). Due to the computational burdens of the chain $\bX^{1:\ell}$, \cite{chagneux2024additive} suggest performing amortized inference by encoding the chain using a recurrent neural network architecture into $\bE^{\ell}$. Consequently, the backward architecture consists of an embedding block common to all layers, and for each layer $1 \leq \ell < L $ a fully connected network for each parameter of the Gaussian taking as input $\bE^{\ell}$, $\bZ^{\ell+1}$, and $\bX^\ell$ through a residual connection as illustrated on Figure \ref{fig:res_backward_architecture}. In our experiments, we define the embedder using a GRU or LSTM from the PyTorch library \citep{paszke2019pytorch}, and we design the fully connected network at each layers by their number of hidden layers solely, fixing the intermediate hidden neurons to the input size.

\paragraph{Model optimization and numerical considerations} The computation of the ELBO presents several numerical challenges that arise due to the need for exponentiation of parameters and inversion of the precision matrix. To mitigate issues related to numerical overflow, we impose constraints on the variational parameters. Specifically, we restrict the means coordinates to the interval $[-100, 25]$ and the variance terms to $[10^{-8}, 10]$ using tempered sigmoid activation functions $B$ defined as
\begin{equation*}
    \forall x \in \RR, \quad B(x) = m + (M - m) \sigma\left(s \times \left(x - \frac{m + M}{2}\right)\right) \eqsp,
\end{equation*}
where $\sigma(x) = \exp(x) / (1 + \exp(x)), s > 0, m \in \RR, M \in \RR$, $B(x) \in [m, M]$. Additionally, to ensure the invertibility of the considered matrices, we introduce a bias of $\lambda = 10^{-4}$ to the diagonal. Subsequently, we opt to employ the Adam optimizer \citep{kingma2014adam} for training our neural networks with learning rate $10^{-3}$ using PyTorch implementation \citep{paszke2019pytorch}. This choice is motivated by its demonstrated stability and efficacy, surpassing alternative optimization techniques in our experiments.

\paragraph{Computational efficiency} 
Training PLN-Tree is more computationally intensive than classical PLN alternatives, particularly as the depth of the taxonomy increases, the dimensionality of the layers grows (excluding only-child nodes that do not require parameterization), and the dataset size expands. In our experiments, conducted on a CPU with $\mathrm{i}5-1335\mathrm{U}\times12$ configuration, training a single-layer PLN model using the \texttt{pyPLNmodels} v0.0.69 \citep{batardiere2024pyplnmodels} package leads to an average iteration time of $0.01$s. In contrast, training the entire PLN-Tree hierarchy has an average iteration time of $0.36$s (batch size set to $512$). This indicates that while PLN-Tree convergence is achieved, the hierarchical nature of the model and its neural network parameterization significantly slow down the process compared to PLN. It should be noted that both \texttt{pyPLNmodels} and our PLN-Tree implementation support GPU acceleration through \texttt{CUDA}, though we did not benchmark GPU performance for this study.

Despite the slower training times observed in CPU-based experiments, PLN-Tree models are inherently scalable to larger datasets contrary to PLN. The critical difference lies in the parameterization of the variational distributions. Indeed, traditional PLN models use a per-individual parameterization of the variational parameters, allowing for fast optimization through closed-form updates at each iteration \citep{chiquet_pln}. However, this variational method requires as many parameters as the number of data points, which can become a bottleneck for very large datasets. In contrast, PLN-Tree models employs a backward variational approximation that can not be parameterized per individual due to the structured dependencies imposed by the backward Markov chain, which is then parameterized using neural networks. While this makes optimization more challenging and slower at each iteration, the number of parameters in PLN-Tree remains a fixed hyperparameter, regardless of the size of the dataset. This property is crucial for scalability, as it allows PLN-Tree to handle large datasets efficiently. 

\subsection{PLN-Tree generated data experiments}
\label{app:experiments_setup_synthetic_plntree}
\subsubsection{Model selection experiments}

In this experiment, the latent prior optimal architecture is already known from the original model. Consequently, we only  optimize the hyperparameters of the variational approximation. The training dataset consists of $2000$ samples from a PLN-Tree model. For each model, we sample $3000$ samples $5$ times and select the model with the best overall performances regarding the $\alpha$-diversity criteria.

\paragraph{Mean-field architectures} We try $3$ architectures of mean-field variational approximations, where the amount of hidden layers in the variation approximation spans in $\{1, 2, 3\}$. The results indicate the optimal architecture is given for $1$ hidden layer.

\paragraph{Backward Markov architectures} The tested architectures are summarized in Table~\ref{tab:synthetic_plntree_backward_selection_params}. The performances of each architecture orientate the choice of the optimal architecture towards the Model 4.
\begin{table}
\centering
\begin{tabular}{lccccc}
\midrule
\textbf{Parameter} & \textbf{Model 1} & \textbf{Model 2} & \textbf{Model 3} & \textbf{Model 4} \\
\midrule
Embedder type & GRU & GRU & GRU & GRU \\
Hidden layers size & 32 & 32 & 32 & 32 \\
Number of hidden layers & 2 & 2 & 3 & 3 \\
Embedding size & 64 & 64 & 64 & 120 \\
\midrule
Number of layers (Gaussian parameters) & 1 & 2 & 1 & 2 \\
\midrule
\end{tabular}
\caption{Tested backward variational architectures in PLN-Tree synthetic data experiments.}
\label{tab:synthetic_plntree_backward_selection_params}
\end{table}

\subsubsection{Performance benchmark}
\label{app:performance_benchmark_plntree_synthetic}
For each selected model, we perform multiple training runs and present the resulting objective values in Figure~\ref{fig:ELBO_repeat}. We observe that the mean-field approximation does not converge to the same value of the ELBO value across different runs (see Figure \ref{fig:ELBO_repeat_MF}), indicating variability in performance. Conversely, our method consistently converges to the same ELBO values (see Figure \ref{fig:ELBO_repeat_backward}), demonstrating stable performance and consistently outperforming the mean-field approach. Thus, in all our experiments, we do not explore the training variability of the mean-field model, and only account for the sampling variability. 
\begin{figure}[htbp]
    \centering
    \begin{subfigure}[b]{0.48\linewidth}
        \centering
        \includegraphics[width=\linewidth]{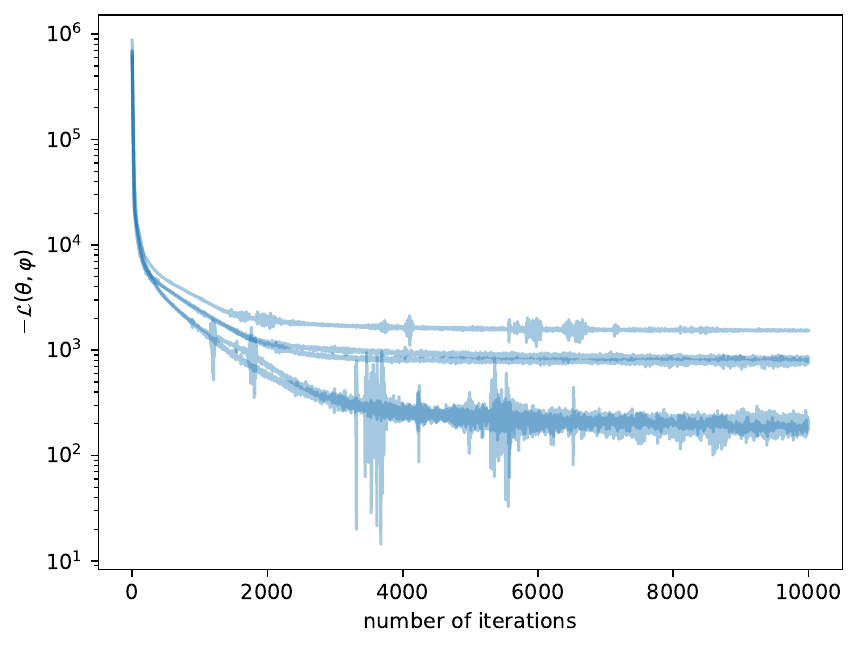}
        \caption{Mean-field approximation}
        \label{fig:ELBO_repeat_MF}
    \end{subfigure}
    \hfill
    \begin{subfigure}[b]{0.48\linewidth}
        \centering
        \includegraphics[width=\linewidth]{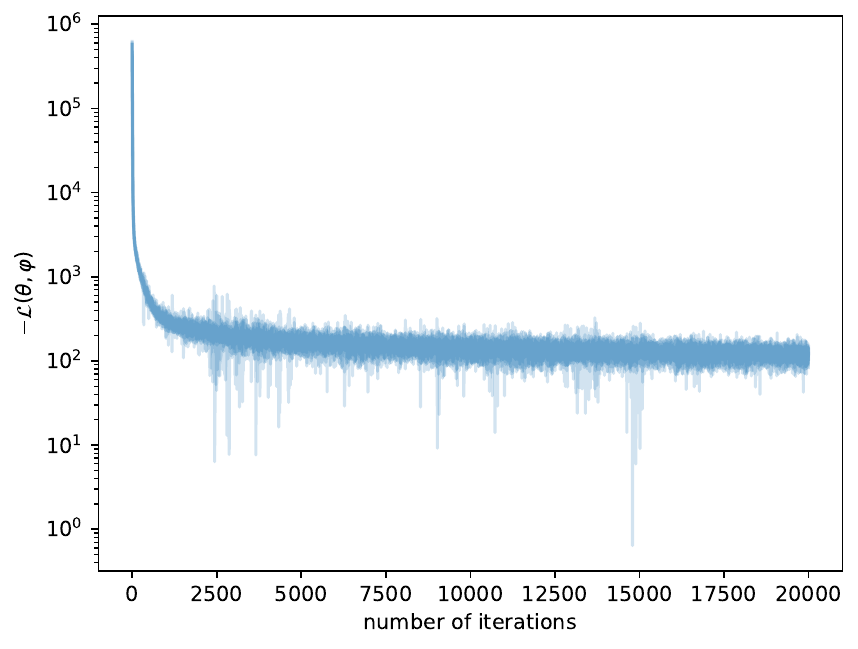}
        \caption{Residual amortized backward approximation}
        \label{fig:ELBO_repeat_backward}
    \end{subfigure}
    \caption{ELBO convergence over iterations for PLN-Tree models on the PLN-Tree generated dataset, repeated $5$ times, performed for mean-field and residual amortized backward variational approximations. Negative values are eluded in log scale.}
    \label{fig:ELBO_repeat}
\end{figure}

For the performance benchmark of the selected models, we sample $2000$ samples $25$ times for each model and show the average result with standard deviation between brackets.
\begin{table}
\centering
\begin{tabular}{@{}lcccc@{}}
\toprule
\multirow{1}{*}{\textbf{Alpha diversity}} & \textbf{PLN-Tree} & \textbf{PLN-Tree (MF)} & \textbf{PLN} & \textbf{SPiEC-Easi} \\
\midrule
    \multicolumn{5}{c}{\textbf{Wasserstein Distance} ($\times 10^{2}$)} \\ \midrule
    Shannon $\ell=1$ & \textbf{1.57} (0.50) & 11.23 (0.73) & 14.64 (1.15) & 46.72 (1.63) \\
    Shannon $\ell=2$ & \textbf{3.67} (1.33) & 5.14 (1.20) & 32.04 (1.62) & 89.62 (2.31) \\
    Shannon $\ell=3$ & \textbf{5.82} (1.51) & 7.86 (1.47) & 35.03 (1.68) & 98.49 (2.31) \\
    Simpson $\ell=1$ & \textbf{0.62} (0.21) & 2.69 (0.27) & 4.91 (0.41) & 15.91 (0.65) \\
    Simpson $\ell=2$ & \textbf{0.71} (0.24) & 1.40 (0.31) & 7.35 (0.41) & 22.13 (0.72) \\
    Simpson $\ell=3$ & \textbf{0.85} (0.24) & 1.55 (0.34) & 7.21 (0.41) & 22.05 (0.70) \\
\midrule
\multicolumn{5}{c}{\textbf{Kolmogorov Smirnov ($\times 10^{-2}$)}} \\ \midrule
Shannon $\ell=1$ & \textbf{2.60} (0.70) & 14.69 (1.06) & 11.0 (0.99) & 32.91 (1.22) \\
Shannon $\ell=2$ & 4.63 (1.29) & \textbf{4.42} (1.00) & 20.68 (0.87) & 47.25 (1.10) \\
Shannon $\ell=3$ & \textbf{5.34} (1.08) & 5.54 (0.99) & 20.2 (1.15) & 45.84 (1.11) \\
Simpson $\ell=1$ & \textbf{2.65} (0.69) & 11.14 (0.93) & 9.87 (0.78) & 28.59 (1.21) \\
Simpson $\ell=2$ & 4.37 (1.24) & \textbf{4.18} (0.89) & 19.14 (0.97) & 42.17 (1.08) \\
Simpson $\ell=3$ & 4.99 (0.92) & \textbf{4.58} (0.79) & 18.92 (0.98) & 42.29 (1.20) \\ \midrule

\multicolumn{5}{c}{\textbf{Total variation ($\times 10^{-2}$)}} \\ \midrule
Shannon $\ell=1$ & \textbf{1.14} (0.29) & 5.67 (0.41) & 4.41 (0.39) & 13.12 (0.49) \\
Shannon $\ell=2$ & \textbf{1.21} (0.32) & 1.24 (0.21) & 5.34 (0.24) & 12.27 (0.29) \\
Shannon $\ell=3$ & \textbf{1.19} (0.22) & 1.31 (0.17) & 4.32 (0.23) & 9.97 (0.22) \\
Simpson $\ell=1$ & \textbf{3.13} (0.71) & 10.84 (0.94) & 11.38 (0.91) & 31.82 (1.42) \\
Simpson $\ell=2$& \textbf{4.29} (0.94) & 4.66 (0.97) & 19.53 (0.91) & 43.60 (1.05) \\
Simpson $\ell=3$ & 4.92 (0.73) & \textbf{4.48} (0.74) & 18.63 (0.94) & 42.23 (1.03) \\
\midrule

\multicolumn{5}{c}{\textbf{Kullback-Leibler Divergence ($\times 10^{-2}$)}} \\ \midrule
Shannon $\ell=1$ & \textbf{0.24} (0.11) & 4.83 (0.50) & 14.17 (1.22) & 30.09 (1.81) \\
Shannon $\ell=2$ & \textbf{0.57} (0.26) & 0.80 (0.23) & 14.39 (1.15) & 71.78 (3.24) \\
Shannon $\ell=3$ & \textbf{1.07} (0.36) & 1.63 (0.45) & 4.56 (0.50) & 87.93 (4.22) \\
Simpson $\ell=1$ & \textbf{0.23} (0.11) & 2.40 (0.31) & 4.56 (0.50) & 23.70 (1.65) \\
Simpson $\ell=2$ & \textbf{0.47} (0.19) & 0.80 (0.28) & 10.64 (0.91) & 51.78 (2.25) \\
Simpson $\ell=3$ & \textbf{0.68} (0.20) & 0.84 (0.26) & 10.57 (0.87) & 52.62 (2.13) \\
\bottomrule
\end{tabular}
\caption{Distribution metrics on $\alpha$-diversities computed between synthetic data sampled under the original PLN-Tree model and simulated data under each modeled trained, averaged over the trainings, with standard deviation.}
\label{tab:synthetic_plntree_benchmark_alpha_diversity}
\end{table}

\subsection{Synthetic data with Markov Dirichlet experiments}
\label{app:experiments_setup_synthetic_markovdirichlet}
\subsubsection{Model selection experiments}
\paragraph{Dataset description and selection procedure} The training dataset consists of $2000$ samples from a Markov Dirichlet model. For each model, when compared to this dataset, we draw $3000$ samples $5$ times and select the model with the best overall performances regarding $\alpha$-diversity criteria.

\paragraph{Mean-field architectures} In this experiment, the number of hidden layers in the latent priors spans in $\{1, 2, 3\}$, while the number of hidden layers in the mean-field approximations spans in $\{1, 2\}$. Trying all combinations, we obtain that the best-performing architecture in our experiment has $2$ hidden layers in the latent prior, and $1$ hidden layer in the variational approximation parameters.

\paragraph{Backward Markov architectures} For the backward architectures, the number of layers tested in the latent priors spans in $\{1, 2\}$. The various tested architectures for the embedders are summarized in Table \ref{tab:synthetic_markovdirichlet_backward_selection_params}. The architecture of the best-performing model is yielded for 1 layers in the latent prior with the embedding architecture E8.

\begin{table}
\centering
\begin{tabular}{lccccc}
\midrule
\textbf{Name} & \textbf{Embedding size} & \textbf{Hidden layers} & \textbf{Nb neurons} \\
\midrule
\textbf{E1} & 16 & 2 & 32 \\
\textbf{E2} & 32 & 2 & 32 \\
\textbf{E3} & 32 & 3 & 32 \\
\textbf{E4} & 32 & 2 & 64 \\
\textbf{E5} & 32 & 3 & 64 \\
\textbf{E6} & 60 & 2 & 64 \\
\textbf{E7} & 60 & 3 & 64 \\
\textbf{E8} & 60 & 3 & 120 \\
\midrule
\end{tabular}
\caption{Tested backward variational architectures in the Embedder in the Markov Dirichlet synthetic experiments. All embedders are GRU, stacked with a $2$ layers neural network to model the parameters.}
\label{tab:synthetic_markovdirichlet_backward_selection_params}
\end{table}

\subsubsection{Performance benchmark}
For the performance benchmark of the selected models, we draw $2000$ samples $25$ times for each model and show the average result with standard deviation between brackets.
\begin{table}
\centering
\begin{tabular}{@{}lcccc@{}}
\toprule
\multirow{1}{*}{\textbf{Alpha diversity}} & \textbf{PLN-Tree} & \textbf{PLN-Tree (MF)} & \textbf{PLN} & \textbf{SPiEC-Easi} \\
\midrule
\multicolumn{5}{c}{\textbf{Wasserstein Distance} ($\times 10^{2}$)} \\ \midrule
Shannon $\ell=1$ & \textbf{17.70} (0.47) & 21.42 (0.59) & 72.27 (1.70) & 125.10 (1.25) \\
Shannon $\ell=2$ & \textbf{22.23} (0.94) & 29.10 (1.06) & 111.53 (1.81) & 177.18 (1.50) \\
Shannon $\ell=3$ & \textbf{24.32} (0.83) & 37.72 (1.14) & 142.28 (1.99) & 224.07 (1.62) \\
Simpson $\ell=1$ & \textbf{5.69} (0.16) & 5.84 (0.16) & 21.74 (0.60) & 39.01 (0.46) \\
Simpson $\ell=2$ & \textbf{5.21} (0.17) & 5.90 (0.19) & 26.70 (0.59) & 46.26 (0.54) \\
Simpson $\ell=3$ & \textbf{3.91} (0.11) & 5.16 (0.16) & 28.55 (0.59) & 50.12 (0.54) \\
\midrule
\multicolumn{5}{c}{\textbf{Kolmogorov Smirnov} ($\times 10^{2}$)} \\ \midrule
Shannon $\ell=1$ & \textbf{16.81} (0.93) & 24.28 (0.9) & 45.12 (1.02) & 66.09 (0.69) \\
Shannon $\ell=2$ & \textbf{19.29} (1.06) & 25.94 (0.97) & 58.83 (0.88) & 76.04 (0.56) \\
Shannon $\ell=3$ & \textbf{20.80} (0.75) & 30.50 (0.98) & 66.62 (0.71) & 83.14 (0.31) \\
Simpson $\ell=1$ & \textbf{13.95} (0.94) & 20.93 (0.9) & 39.77 (1.10) & 61.65 (0.73) \\
Simpson $\ell=2$ & \textbf{18.35} (1.03) & 23.47 (0.97) & 55.42 (0.87) & 70.24 (0.69) \\
Simpson $\ell=3$ & \textbf{22.00} (0.87) & 30.43 (0.82) & 62.10 (0.71) & 77.63 (0.32) \\
\midrule
\multicolumn{5}{c}{\textbf{Total variation} ($\times 10^{2}$)} \\ \midrule
Shannon $\ell=1$ & \textbf{7.75} (0.29) & 9.78 (0.35) & 14.87 (0.33) & 21.50 (0.29) \\
Shannon $\ell=2$ & \textbf{6.67} (0.30) & 8.08 (0.28) & 15.59 (0.21) & 19.54 (0.18) \\
Shannon $\ell=3$ & \textbf{5.60} (0.16) & 7.47 (0.24) & 14.93 (0.14) & 18.23 (0.08) \\
Simpson $\ell=1$ & \textbf{19.33} (0.68) & 23.72 (1.02) & 38.32 (1.02) & 58.01 (0.86) \\
Simpson $\ell=2$ & \textbf{20.11} (0.89) & 24.60 (1.01) & 50.64 (0.79) & 64.60 (0.75) \\
Simpson $\ell=3$ & \textbf{19.21} (0.64) & 26.42 (0.96) & 56.78 (0.66) & 71.10 (0.47) \\
\midrule
\multicolumn{5}{c}{\textbf{Kullback-Leibler divergence} ($\times 10^{2}$)} \\ \midrule
Shannon $\ell=1$ & \textbf{20.72} (2.42) & 23.72 (1.70) & 60.51 (3.14) & 1.8236 (7.55) \\
Shannon $\ell=2$ & \textbf{28.77} (5.75) & 33.04 (3.33) & 153.73 (8.38) & 423.20 (57.72) \\
Shannon $\ell=3$ & \textbf{25.02} (4.03) & 40.96 (3.73) & 226.13 (11.96) & 784.49 (160.02) \\
Simpson $\ell=1$ & \textbf{15.21} (1.71) & 15.75 (1.38) & 39.04 (2.18) & 119.83 (4.30) \\
Simpson $\ell=2$ & \textbf{26.26} (8.32) & 26.84 (5.95) & 81.47 (3.49) & 198.69 (6.90) \\
Simpson $\ell=3$ & \textbf{21.68} (7.61) & 29.71 (7.99) & 106.28 (3.40) & 265.42 (7.48) \\
\bottomrule
\end{tabular}
\caption{Distribution metrics on $\alpha$-diversities computed between synthetic data sampled under the Markov Dirichlet model and simulated data under each modeled trained, averaged over the trainings, with standard deviation.}
\label{tab:synthetic_markovdirichlet_benchmark_alpha_diversity}
\end{table}

\subsection{Metagenomics dataset experiments}
\label{app:experiments_setup_metagenomics}
\subsubsection{Model selection experiments}
\paragraph{Selection procedure} For each model, when compared to the metagenomics dataset, we draw $3000$ samples $5$ times and select the model with the best overall performances regarding the $\alpha$-diversity criteria.

\paragraph{Mean-field architectures} We try all combinations of the number of hidden layers for the latent prior and the variational approximation taking values in $\{1, 2, 3\}$. The best-performing architecture in our experiment has $1$ hidden layers in the latent prior, and $2$ hidden layers in the variational approximation parameters.

\paragraph{Backward Markov architectures} We decide on a grid of embedders summarized in Table \ref{tab:synthetic_metagenomics_backward_selection_params}, which we combine with latent prior architecture with a number of hidden layers in $\{1, 2, 3\}$. In our experiment, the best architecture is yielded by the embedding architecture E4.

\begin{table}
\centering
\begin{tabular}{lccccc}
\midrule
\textbf{Name} & \textbf{Embedding size} & \textbf{Hidden layers} & \textbf{Nb neurons} & \textbf{Parameters layers} \\
\midrule
\textbf{E1} & 16 & 2 & 32 & 2 \\
\textbf{E2} & 32 & 2 & 32 & 2 \\
\textbf{E3} & 32 & 3 & 32 & 2 \\
\textbf{E4} & 32 & 2 & 64 & 2 \\
\textbf{E5} & 32 & 3 & 64 & 2 \\
\textbf{E6} & 32 & 3 & 64 & 3 \\
\textbf{E7} & 60 & 2 & 64 & 2 \\
\textbf{E8} & 60 & 3 & 64 & 2 \\
\textbf{E9} & 60 & 3 & 64 & 3 \\
\textbf{E10} & 60 & 3 & 120 & 2 \\
\textbf{E11} & 60 & 3 & 120 & 3 \\
\midrule
\end{tabular}
\caption{Tested backward variational architectures in the Embedder in the metagenomics experiments. All embedders are GRU.}
\label{tab:synthetic_metagenomics_backward_selection_params}
\end{table}

\begin{table}
\centering
\begin{tabular}{@{}lcccc@{}}
\toprule
\multirow{1}{*}{\textbf{Alpha diversity}} & \textbf{PLN-Tree} & \textbf{PLN-Tree (MF)} & \textbf{PLN} & \textbf{SPiEC-Easi} \\
\midrule
\multicolumn{5}{c}{\textbf{Wasserstein distance} ($\times 10^{2}$)} \\ \midrule
Shannon $\ell=1$ & \textbf{1.73} (0.44) & 3.00 (0.44) & 16.49 (1.14) & 43.12 (1.57) \\
Shannon $\ell=2$ & \textbf{2.22} (0.73) & 5.70 (0.97) & 23.21 (1.64) & 57.73 (2.02) \\
Shannon $\ell=3$ & \textbf{2.29} (0.63) & 6.58 (1.02) & 23.96 (1.67) & 59.16 (2.00) \\
Shannon $\ell=4$ & \textbf{2.08} (0.62) & 20.39 (1.08) & 55.32 (2.38) & 127.11 (3.03) \\
Simpson $\ell=1$ & 0.84 (0.14) & \textbf{0.71} (0.12) & 7.18 (0.48) & 17.99 (0.71) \\
Simpson $\ell=2$ & 0.92 (0.24) & \textbf{0.73} (0.19) & 7.49 (0.57) & 19.59 (0.81) \\
Simpson $\ell=3$ & 0.91 (0.23) & \textbf{0.72} (0.19) & 7.46 (0.57) & 19.50 (0.80) \\
Simpson $\ell=4$ & \textbf{0.53} (0.13) & 2.41 (0.21) & 12.91 (0.67) & 31.62 (0.99) \\
\midrule
\multicolumn{5}{c}{\textbf{Kolmogorov Smirnov} ($\times 10^{2}$)} \\ \midrule
Shannon $\ell=1$ & \textbf{4.71} (1.44) & 8.4 (1.35) & 23.17 (1.26) & 45.26 (1.53) \\
Shannon $\ell=2$ & \textbf{3.42} (0.99) & 10.3 (1.25) & 22.14 (1.58) & 45.65 (1.63) \\
Shannon $\ell=3$ & \textbf{3.48} (0.68) & 10.66 (1.32) & 22.07 (1.47) & 45.57 (1.58) \\
Shannon $\ell=4$ & \textbf{3.64} (1.06) & 22.66 (1.3) & 36.65 (1.50) & 65.15 (1.32) \\
Simpson $\ell=1$ & 4.8 (0.93) & \textbf{4.17} (0.58) & 21.25 (1.47) & 43.30 (1.86) \\
Simpson $\ell=2$ & \textbf{4.46} (1.06) & 5.6 (1.46) & 19.64 (1.45) & 41.76 (1.94) \\
Simpson $\ell=3$ & \textbf{4.17} (1.03) & 5.7 (1.53) & 19.53 (1.42) & 41.53 (1.91) \\
Simpson $\ell=4$ & \textbf{4.09} (1.06) & 12.26 (1.46) & 32.07 (1.69) & 58.34 (1.49) \\
\midrule
\multicolumn{5}{c}{\textbf{Total variation} ($\times 10^{2}$)} \\ \midrule
Shannon $\ell=1$ & \textbf{2.34} (0.63) & 4.51 (0.75) & 10.00 (0.63) & 19.54 (0.75) \\
Shannon $\ell=2$ & \textbf{1.42} (0.34) & 3.87 (0.52) & 7.47 (0.63) & 15.38 (0.63) \\
Shannon $\ell=3$ & \textbf{1.36} (0.24) & 3.81 (0.48) & 7.19 (0.59) & 14.88 (0.62) \\
Shannon $\ell=4$ & \textbf{0.82} (0.27) & 6.3 (0.42) & 8.94 (0.38) & 15.69 (0.30) \\
Simpson $\ell=1$ & \textbf{6.71} (1.18) & 8.17 (1.8) & 27.59 (1.80) & 51.63 (2.13) \\
Simpson $\ell=2$ & \textbf{5.51} (0.89) & 7.41 (1.61) & 22.23 (1.60) & 45.53 (2.12) \\
Simpson $\ell=3$ & \textbf{5.63} (0.93) & 7.35 (1.64) & 21.83 (1.56) & 44.71 (2.05) \\
Simpson $\ell=4$ & \textbf{3.75} (1.15) & 14.88 (1.72) & 34.56 (1.58) & 59.51 (1.28) \\
\midrule
\multicolumn{5}{c}{\textbf{Kullback-Leibler divergence} ($\times 10^{2}$)} \\ \midrule
Shannon $\ell=1$ & \textbf{0.88} (0.32) & 2.32 (0.82) & 15.98 (1.33) & 48.32 (2.94) \\
Shannon $\ell=2$ & \textbf{0.86} (0.32) & 2.68 (0.71) & 15.33 (1.59) & 48.99 (3.24) \\
Shannon $\ell=3$ & \textbf{0.71} (0.29) & 2.87 (0.75) & 15.30 (1.57) & 49.10 (3.13) \\
Shannon $\ell=4$ & \textbf{0.57} (0.22) & 18.02 (4.71) & 35.01 (2.64) & 116.63 (3.97) \\
Simpson $\ell=1$ & \textbf{1.04} (0.33) & 1.54 (0.68) & 15.57 (1.40) & 45.21 (3.10) \\
Simpson $\ell=2$ & \textbf{0.96} (0.28) & 1.14 (0.50) & 13.61 (1.42) & 42.20 (3.24) \\
Simpson $\ell=3$ & \textbf{0.97} (0.32) & 1.11 (0.51) & 13.58 (1.41) & 41.81 (3.18) \\
Simpson $\ell=4$ & \textbf{0.81} (0.30) & 9.70 (3.48) & 29.66 (2.38) & 90.79 (3.80) \\
\bottomrule
\end{tabular}
\caption{Distribution metrics on $\alpha$-diversities computed between metagenomics data and simulated data under each modeled trained, averaged over the trainings, with standard deviation.}
\label{tab:metagenomics_benchmark_alpha_diversity}
\end{table}
\begin{figure}
    \centering
    \includegraphics[width=\linewidth]{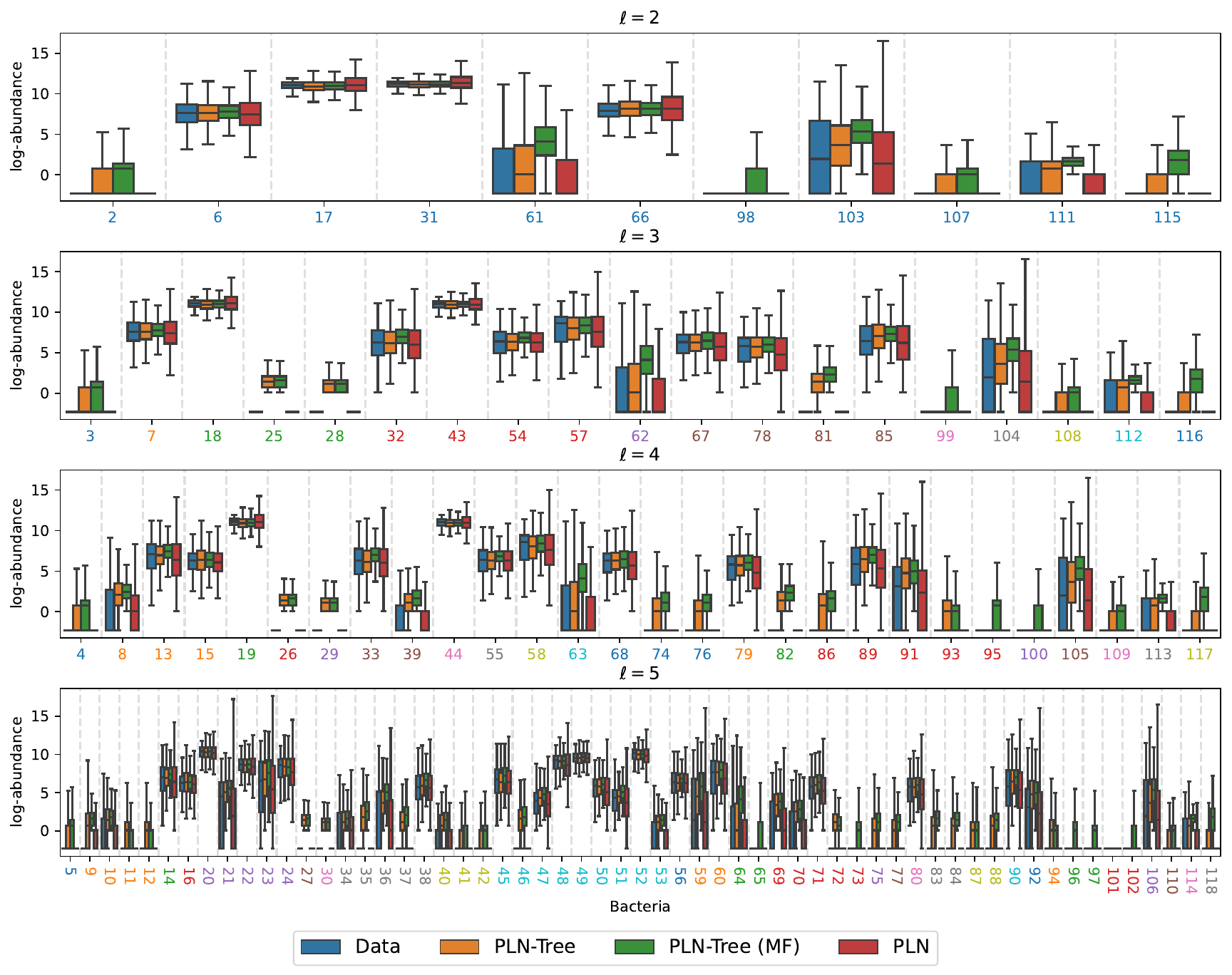}
    \caption{Boxplot of log abundances of the metagenomics dataset and generated data from several PLN-based models learned on this dataset, with $20000$ points per model. Zero abundances are artificially shifted to $10^{-1}$ to represent them in log scale. The bacteria are denoted by a unique integer on the x-axis, with colors indicating the brotherhoods in the taxonomic tree at a given depth.}
    \label{fig:metagenomics_abundances_boxplot}
\end{figure}

\subsubsection{Classification using PLN-based LP-CLR preprocessing}
\label{app:classif_appendix_preprocessing}
For the classification benchmark on the metagenomics dataset, we consider four different types of inputs for various classifiers: the raw data, the CLR transformed counts, the Proj-PLN features, the backward PLN-Tree LTP-CLR latent variables, and the corresponding mean-field variant. Using the same taxa-abundance data as in the previous experiment, we adopt the PLN-Tree architectures selected from our prior model selection on the metagenomics dataset. We then proceed to the training of each model on the entire dataset, and encode the taxa-abundance data into respective latent variables. We then select various classifiers for which the some decisive hyperparameters are tuned within each training fold using a random grid search over Table \ref{tab:classifiers_metagenomics_preprocessing_desc}, with a cross-validation performed on $20\%$ of the training set available. The unspecified hyperparameters are selected from default Scikit-Learn proposals \citep{pedregosa2011scikit}. In this experiment, we only consider the deepest layer of the input data.
\begin{table}
    \centering
    \begin{tabular}{lr}
    \midrule
    \textbf{Model} & \textbf{Parameters} \\
    \midrule
    Logistic Regression & class weight: [balanced, None], \\ & C:$[0.01, 0.05, 0.1, 0.3, 0.6, 1, 1.5, 2]$ \\ & \\
    SVC & probability: true, kernel: linear, \\& class weight: [balanced, None]  \\ & C:$[0.01, 0.05, 0.1, 0.3, 0.5, 0.8, 1.1, 1.5]$ \\ & \\
    MLP & hidden layers sizes: $[(256, 256, 256, 124), (256, 256, 256),$  \\ & $(256, 256, 124), (124, 124, 64),$ \\ &  $(64, 64, 32), (32, 32, 32), (32, 32, 16)]$ \\ & \\
    Random Forests & number of estimators: $[40, 100, 150]$,  class weight: [balanced, None] \\ & max depth: [None, $3, 5, 10, 30$], min leaf: $[1, 3]$ \\ & criterion: [gini, entropy, log loss], min split: $[2, 4]$ \\
    \bottomrule
    \end{tabular}
    \caption{Considered classifiers in the metagenomics preprocessing experiment, with hyperparameters grid based on Scikit-Learn implementation.}
    \label{tab:classifiers_metagenomics_preprocessing_desc}
\end{table}

To further illustrate the impact of the preprocessing on the classifiers' performances, we study the IBD-vs-all problem in addition to the T2D-vs-all presented in the article. The results are presented in Table \ref{tab:ibd_classif_performances}, demonstrating similar interpretations to what is observed in the T2D-vs-all problem.
\begin{table}
    \centering
    \resizebox{0.99\textwidth}{!}{
    \begin{tabular}{@{}lc|cc|ccc@{}}
    \toprule
    & \textbf{Proportions} & \textbf{CLR} & \textbf{PLN} & \textbf{Proj-PLN} & \textbf{LTP-CLR (MF)} & \textbf{LTP-CLR} \\
    \midrule
    \multicolumn{7}{c}{\textbf{Logistic Regression}} \\ \midrule
    \textbf{Balanced accuracy} & 0.502 (0.005) & 0.694 (0.04) & 0.692 (0.041) & \textbf{0.698} (0.047) & 0.636 (0.042) & 0.688 (0.05) \\
    \textbf{Precision} & 0.643 (0.071) & 0.809 (0.03) & 0.809 (0.026) & \textbf{0.821} (0.03) & 0.802 (0.025) & 0.816 (0.026) \\
    \textbf{Recall} & 0.786 (0.003) & 0.817 (0.03) & 0.812 (0.033) & \textbf{0.831} (0.028) & 0.812 (0.028) & 0.824 (0.028) \\
    \textbf{F1} & 0.692 (0.005) & 0.809 (0.028) & 0.805 (0.027) & \textbf{0.819} (0.029) & 0.786 (0.025) & 0.811 (0.028) \\
    \textbf{ROC AUC} & 0.718 (0.042) & 0.823 (0.038) & 0.813 (0.045) & 0.833 (0.038) & 0.795 (0.038) & \textbf{0.842} (0.035) \\
    \textbf{PR AUC} & 0.397 (0.059) & 0.591 (0.082) & 0.591 (0.079) & \textbf{0.631} (0.08) & 0.556 (0.055) & 0.611 (0.075) \\
    \midrule
    \multicolumn{7}{c}{\textbf{Linear SVM}} \\ \midrule
    \textbf{Balanced accuracy} & 0.502 (0.008) & 0.713 (0.037) & 0.707 (0.044) & \textbf{0.73} (0.044) & 0.589 (0.059) & 0.686 (0.058) \\
    \textbf{Precision} & 0.626 (0.034) & 0.818 (0.024) & 0.818 (0.03) & \textbf{0.84} (0.029) & 0.772 (0.073) & 0.824 (0.03) \\
    \textbf{Recall} & 0.785 (0.004) & 0.821 (0.03) & 0.825 (0.028) & \textbf{0.846} (0.03) & 0.807 (0.024) & 0.832 (0.024) \\
    \textbf{F1} & 0.692 (0.008) & 0.816 (0.027) & 0.817 (0.028) & \textbf{0.838} (0.03) & 0.759 (0.043) & 0.815 (0.032) \\
    \textbf{ROC AUC} & 0.658 (0.165) & 0.823 (0.035) & 0.826 (0.038) & \textbf{0.844 (0.032)} & 0.784 (0.037) & \textbf{0.844} (0.032)\\
    \textbf{PR AUC} & 0.371 (0.114) & 0.588 (0.072) & 0.606 (0.079) & \textbf{0.655} (0.064) & 0.551 (0.071) & 0.632 (0.069) \\
    \midrule
    \multicolumn{7}{c}{\textbf{Neural Network}} \\ \midrule
    \textbf{Balanced accuracy} & 0.695 (0.052) & 0.732 (0.041) & 0.727 (0.042) & \textbf{0.738} (0.046) & 0.664 (0.04) & 0.725 (0.038) \\
    \textbf{Precision} & 0.809 (0.038) & 0.833 (0.026) & 0.828 (0.028) & \textbf{0.839} (0.028) & 0.801 (0.028) & 0.824 (0.025) \\
    \textbf{Recall} & 0.822 (0.024) & 0.84 (0.024) & 0.835 (0.026) & \textbf{0.846} (0.024) & 0.812 (0.027) & 0.829 (0.026) \\
    \textbf{F1} & 0.812 (0.031) & 0.834 (0.025) & 0.83 (0.027) & \textbf{0.839} (0.027) & 0.797 (0.025) & 0.825 (0.025) \\
    \textbf{ROC AUC} & 0.813 (0.038) & 0.84 (0.041) & 0.83 (0.039) & \textbf{0.854} (0.036) & 0.776 (0.042) & 0.837 (0.038) \\
    \textbf{PR AUC} & 0.578 (0.078) & 0.649 (0.079) & 0.621 (0.076) & \textbf{0.664} (0.074) & 0.536 (0.077) & 0.636 (0.067) \\
    \midrule
    \multicolumn{7}{c}{\textbf{Random Forest}} \\ \midrule
    \textbf{Balanced accuracy} & \textbf{0.662} (0.057) & 0.627 (0.073) & 0.625 (0.055) & 0.615 (0.055) & 0.579 (0.053) & 0.602 (0.057) \\
    \textbf{Precision} & \textbf{0.857} (0.022) & 0.826 (0.06) & 0.83 (0.041) & 0.815 (0.064) & 0.804 (0.052) & 0.811 (0.057) \\
    \textbf{Recall} & \textbf{0.847} (0.022) & 0.832 (0.028) & 0.83 (0.02) & 0.826 (0.021) & 0.81 (0.017) & 0.82 (0.021) \\
    \textbf{F1} & \textbf{0.816} (0.036) & 0.79 (0.052) & 0.79 (0.036) & 0.784 (0.04) & 0.755 (0.037) & 0.773 (0.042) \\
    \textbf{ROC AUC} & \textbf{0.904} (0.03) & 0.876 (0.034) & 0.865 (0.041) & 0.86 (0.041) & 0.791 (0.03) & 0.848 (0.038) \\
    \textbf{PR AUC} & \textbf{0.771} (0.062) & 0.695 (0.083) & 0.688 (0.071) & 0.681 (0.063) & 0.56 (0.061) & 0.634 (0.071) \\
    \bottomrule
    \end{tabular}
    }
\caption{Classification IBD-vs-all performances for several classifiers on the metagenomics dataset using different preprocessing strategies, averaged over training, with standard deviation. We perform $50$ stratified K-folds using $80\%$ of the dataset with $50$ nested random grid search loops for hyperparameters tuning in each fold using $20\%$ of the training set for cross-validation. The dataset is restricted to the \textit{family} level of the taxonomy.}
\label{tab:ibd_classif_performances}
\end{table}

\section{Additional experiments visualisations}
\begin{figure}
    \centering
    \includegraphics[width=\linewidth]{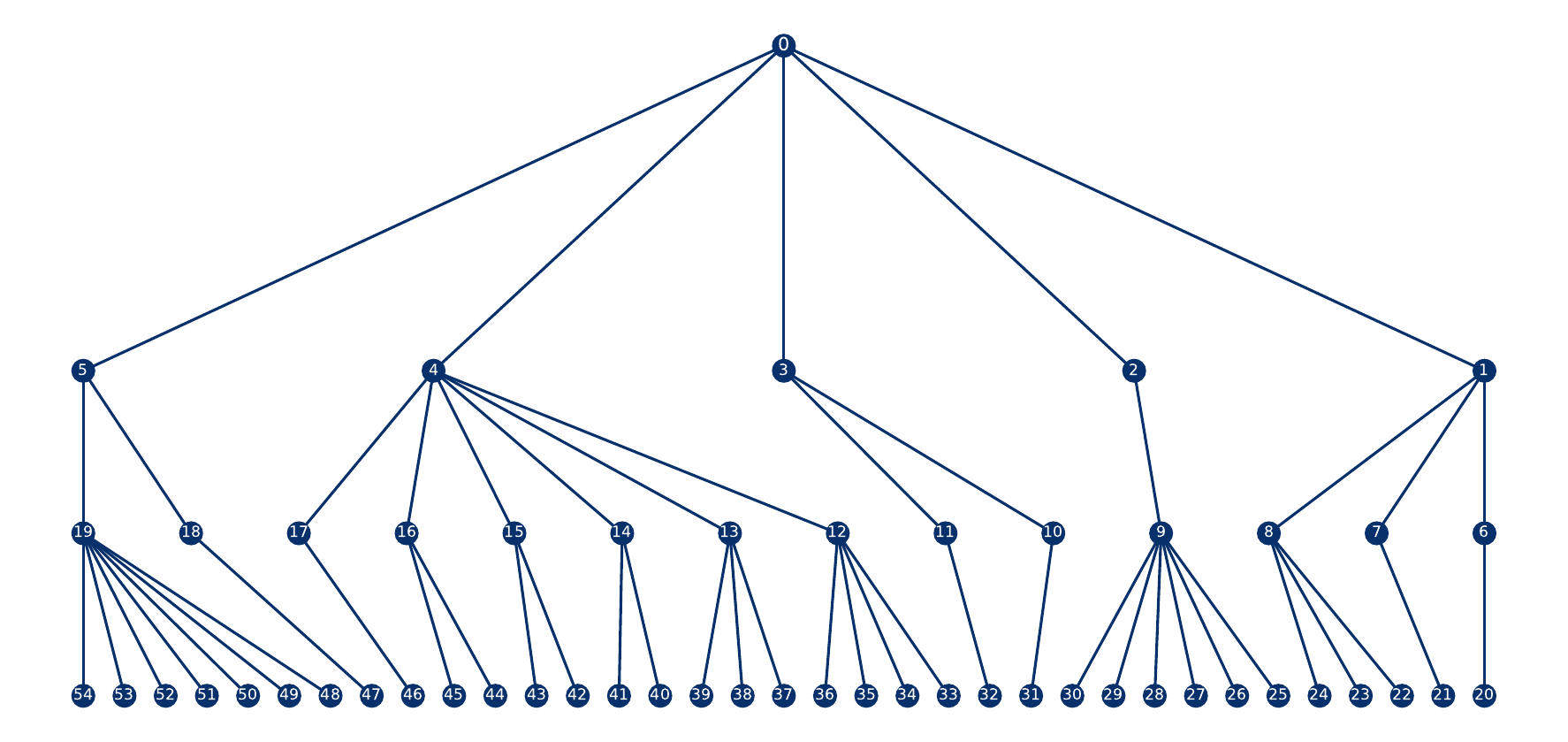}
    \caption{Graph of the tree considered in the PLN-Tree synthetic experiments.}
    \label{fig:synthetic_tree_plntree}
\end{figure}
\begin{figure}
    \centering
    \includegraphics[width=\linewidth]{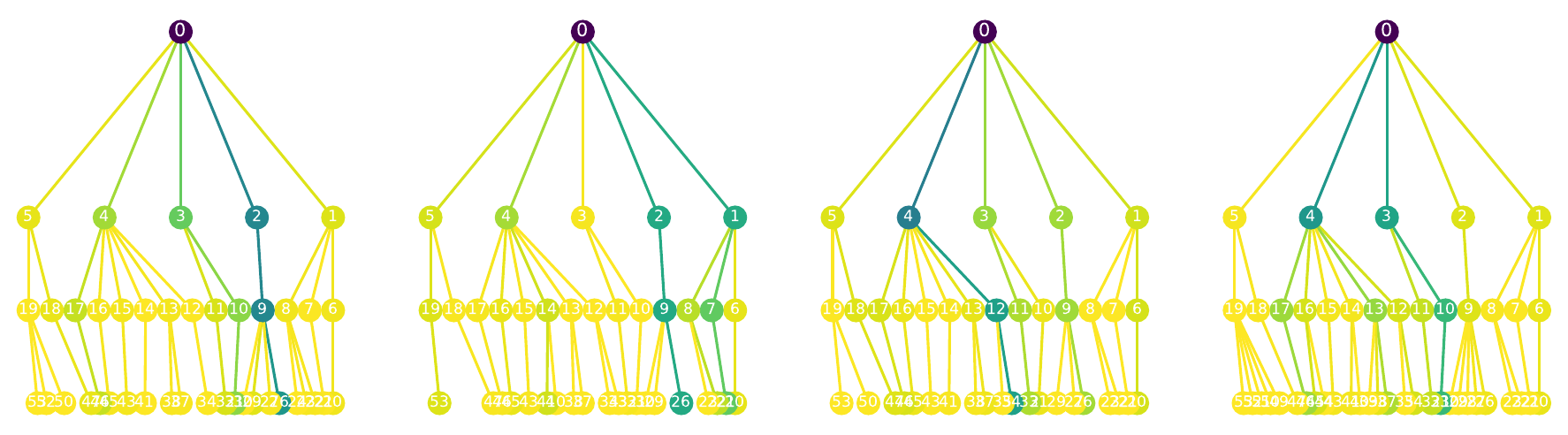}
    \caption{Synthetic hierarchical samples from the artificial dataset $(\bX, \bZ)$.}
    \label{fig:abundance_samples_synthetic_plntree}
\end{figure}
\begin{figure}
    \centering
    \includegraphics[width=\linewidth]{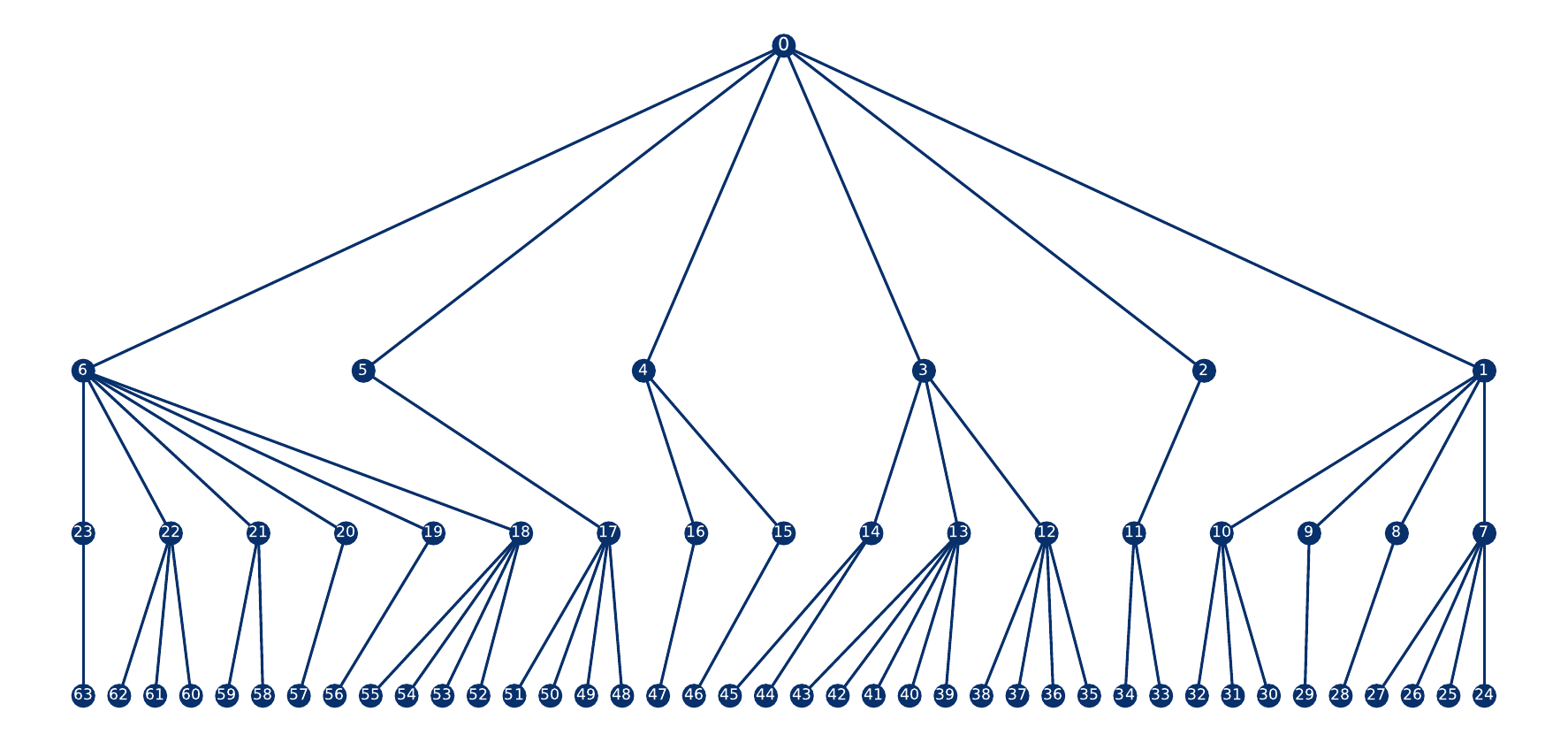}
    \caption{Graph of the tree considered in the Markov Dirichlet synthetic experiments.}
    \label{fig:synthetic_tree_markov_dirichlet}
\end{figure}
\begin{figure}
    \centering
    \includegraphics[width=\linewidth]{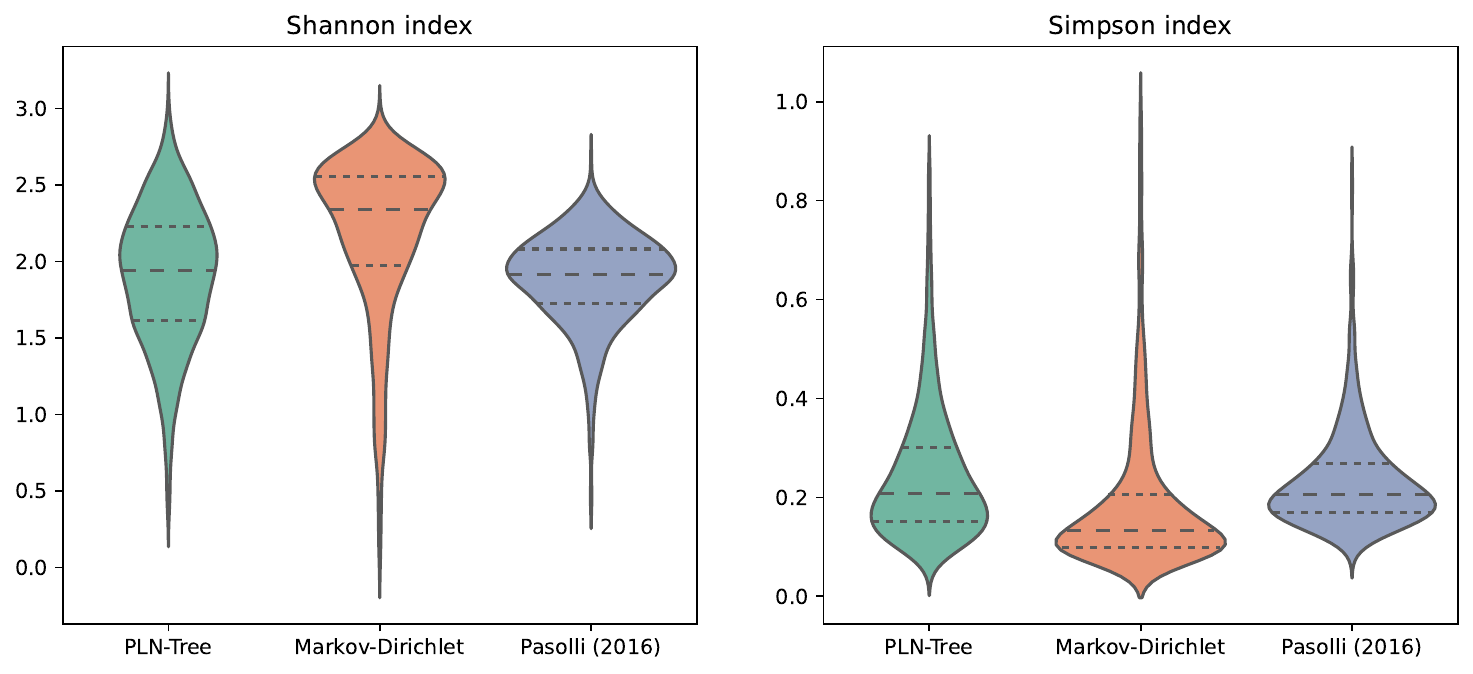}
    \caption{Alpha diversity comparison between the datasets used in each experiment. The $\alpha$-diversities are computed at the deepest level of the hierarchy considered in the respective problems. The tag Pasolli (2016) refers to the metagenomics dataset from \cite{pasolli2016machine}, PLN-Tree refers to the synthetic dataset generated using a PLN-Tree model, Markov-Dirichlet refers to the dataset generated using a Markov-Dirichlet dynamic.}
    \label{fig:comparison_alpha_diversity_datasets}
\end{figure}